\documentclass[aps,pra,superscriptaddress,aps,nofootinbib]{revtex4}
\usepackage{srcltx}
\usepackage{graphicx}
\graphicspath{{./figsSubm/}}
\usepackage{epstopdf}
\usepackage[american]{babel}
\usepackage{amssymb}
\usepackage{amsfonts}
\usepackage{color}
\usepackage[normalem]{ulem}
\usepackage{comment}
\usepackage[abs]{overpic}

\begin{document}
\newtheorem{corollary}{Corollary}[section]
\newtheorem{remark}{Remark}[section]
\newtheorem{definition}{Definition}[section]
\newtheorem{theorem}{Theorem}[section]
\newtheorem{proposition}{Proposition}[section]
\newtheorem{lemma}{Lemma}[section]
\newtheorem{help1}{Example}[section]
\newenvironment{proof}[1][Proof]{\textit{#1:} }{ $\square$}
\renewcommand{\theequation}{\arabic{section}.\arabic{equation}}

\newcommand{\tl}{\textlatin}
\newcommand{\no}{\noindent}

\newcommand{\bbR}{\mathbb{R}}

\newcommand{\bbN}{\mathbb{N}}
\newcommand{\bbC}{\mathbb{C}}
\newcommand{\bbZ}{\mathbb{Z}}

\newcommand{\bfR}{\mathbf{R}}
\newcommand{\bfE}{\mathbf{E}}
\newcommand{\bfC}{\mathbf{C}}
\newcommand{\mcL}{\mathcal{L}}
\newcommand{\mcU}{\mathcal{U}}
\newcommand{\mcA}{\mathcal{A}}
\newcommand{\mcR}{\mathcal{R}}
\newcommand{\ran}{\mathrm{ran}}
\newcommand{\mcV}{\mathcal{V}}
\newcommand{\mcB}{\mathcal{AB}}   
\newcommand{\dom}{\mathrm{dom}}
\newcommand\fh{{\mathfrak{H}}}
\newcommand\ch{{\mathcal{H}}}
\newcommand\ci{{\mathcal{I}}}
\newcommand\cb{{\mathcal{B}}}
\newcommand\cj{{\mathcal{J}}}
\newcommand\ck{{\mathcal{K}}}
\newcommand\cd{{\mathcal{D}}}
\newcommand\cw{{\mathcal{W}}}
\newcommand\ct{{\mathcal{T}}}
\newcommand\cx{{\mathcal{X}}}
\newcommand\cf{{\mathcal{F}}}
\newcommand\cy{{\mathcal{Y}}}

\newcommand\cm{{\mathcal {M}}}
\newcommand\bbc{{\mathbb{C}}}
\newcommand\bfh{{\mathbf{B}(\mathfrak {H}})}

\newcommand\hilbert{{\tl {Hilbert}\;}}
\newcommand\inpr{{\langle\cdot, \cdot\rangle }}

\title{Dynamical transitions between equilibria in a dissipative Klein-Gordon lattice
}
\author{D. J. Frantzeskakis}
\affiliation{Department of Physics, University of Athens, Panepistimiopolis, Zografos, Athens 15784, Greece}
\author{N. I. Karachalios}
\affiliation{Department of Mathematics, University of the Aegean, Karlovassi, 83200 Samos, Greece}
\author{P. G. Kevrekidis}
\affiliation{Department of Mathematics and Statistics, University of Massachusetts, Amherst MA 01003-4515, USA}
\author{V. Koukouloyannis}
\affiliation{Department of Mathematics, Statistics and Physics, College of Arts and
Sciences, Qatar University, P. O. Box 2713, Doha, Qatar}
\author{K. Vetas}
\affiliation{Department of Mathematics, University of the Aegean, Karlovassi, 83200 Samos, Greece}
\begin{abstract}
  We consider the energy landscape of a dissipative Klein-Gordon lattice with
  a $\phi^4$ on-site potential. Our analysis is
  based on suitable energy arguments, combined with a discrete version of the  \L{}ojasiewicz inequality, in order to justify the convergence
to a single, nontrivial equilibrium for all initial configurations of the lattice. Then, global bifurcation theory is explored, to illustrate that
in the discrete regime  all linear states lead to nonlinear
generalizations of equilibrium states. Direct numerical simulations reveal the rich structure of the equilibrium set, consisting of non-trivial topological (kink-shaped) interpolations between the adjacent minima of the on-site potential, and the wealth of dynamical convergence possibilities. 
These dynamical evolution results also
provide insight on the potential stability of the equilibrium branches, 
and glimpses of the emerging global bifurcation structure, elucidating the role of the interplay between discreteness, nonlinearity and dissipation. 

\end{abstract}

\maketitle

\section{Introduction}
\setcounter{equation}{0}

In the present paper, we investigate the dynamics of 
the following discrete Klein-Gordon equation (DKG) model:
\begin{eqnarray}
	\label{eq0}
	\ddot{U}_n-k(U_{n+1}-2U_n+U_{n-1})+ \delta \dot{U}_n = \omega_d^2 (U_{n}-\beta U_{n}^3) ,\;\;\beta,\;\delta>0.
\end{eqnarray}
In Eq.~(\ref{eq0}), the real valued function $U_n(t)$ is the unknown displacement of the oscillator occupying the lattice site $n$, and $k=h^{-2}$ denotes the discretization parameter,
with $h$ playing the role of the lattice spacing (controlling
the distance from the continuum limit of $h \rightarrow 0$). The chain incorporates linear dissipation of strength $\delta>0$, while the parameter $\beta>0$, stands for the strength of the cubic nonlinear term.  Equation~(\ref{eq0}) implies that each individual oscillator of unit mass evolves within the quartic on-site potential of strength $\omega_d^2$, namely: 
\begin{eqnarray}
	\label{eq01}
	W(U)=-\frac{\omega_d^2}{2}U^2+\frac{\beta\omega_d^2}{4}U^4.
\end{eqnarray}
When $\delta=0$, Eq.~(\ref{eq0}) is known as the discrete $\phi^4$ model,  and is one of the fundamental nonlinear lattices. The subject of
discretization of the Hamiltonian variant of the
$\phi^4$ model is one that has been of intense interest
over the past two decades;
see, e.g.,~\cite{speight,physdpgk,pelinovsky,barash,dmitrievpre,pgk_dmitriev,pgk_dmitriev2}
to name only a few examples in the context of devising translational 
invariant discretizations and topological coherent structures
bearing the ability to travel in the lattice setting. Needless to say
that the discretizations critically affect~\cite{dmitrievpre}
also key properties of
the continuum non-integrable version of the model such as
its fractal soliton collisions~\cite{dkc0,anninos}.

In turn, both the continuum and --where relevant-- the discrete
form of the model have
been used in numerous distinct physical contexts; these include 
crystals and metamaterials, ferroelectric and ferromagnetic domain
walls, Josephson junctions, nonlinear optics, topological excitations in hydrogen-bonded chains or chains of base pairs in DNA, complex
electromechanical devices,
and so on --cf. the monographs \cite{DPbook, Panosbook, Chrisbook} and the review articles \cite{reviewsA,reviewsB,reviewsC,reviewsD}. The DKG model, as well as 
its complex analogue, namely the Discrete Nonlinear Schr\"{o}dinger equation (DNLS)~\cite{Panosbook}, 
which is also a cubically-nonlinear dynamical lattice,   
have attracted extensive
interest. This is not only due to their physical relevance alluded to
above, but also
in part due to constituting prototypical playgrounds
where the interplay of discreteness and nonlinearity can be assessed
and compared to the corresponding continuum limit --the NLS and
KG partial differential equations (PDEs).
In general, the crucial differences between the PDEs and the discrete models
are triggered by the breaking of the translational invariance
in the latter. This, in turn, results in the non-equivalence of the configurations of the lattice, the emergence of the so-called Peierls-Nabarro barrier and related features~\cite{DPbook,Panosbook}.

Due to the underlying double-well potential energy (\ref{eq01}), the discrete $\phi^4$ model has been one of the prototypical models for the study of transition state phenomenology for chains of coupled objects. Such a phenomenology refers to the case where the considered objects, initiating from the domain of attraction of a locally stable state, escape to a neighboring stable state
crossing a separating energy barrier \cite{Dirk1, Dirk2, PhysD2013}. The barrier corresponds to the potential's local maximum, associated with 
the saddle point, which separates the two local minima of the potential. Then, ``kink shaped'' topological excitations,  interpolating between the adjacent minima may exist and
play a critical role in the study of transition probabilities. While it is well known that the discrete models admit stationary kink solutions \cite{combs,sega}, the breaking of the translational invariance, 
suggests an important question, concerning the existence of traveling discrete kinks (and of traveling waves solutions in general). The above references, as well as numerous other related works, including (but not limited to)~\cite{chris1,panos0,ioss1,ioss2}, 
provide a sense of the interest that these questions have triggered. 

Turning now to the dissipative variant of the model, the above 
transition-state phenomenology has been
widely explored in numerous physical settings where dissipation effects have a prominent role, including, among others, phase transitions, chemical kinetics, pattern formation and the
kinematics of biological waves --see, e.g., Refs.~\cite{Kramers, PF, KMI,YN,AS} and references therein. 
In particular, when $\delta>0$, the linearly damped counterpart of the discrete $\phi^4$ model (\ref{eq0}),
is of obvious significance to the discrete physical set-ups mentioned above, when friction effects (and possibly driving forces) cannot be neglected, \cite{comte, panos0, Boris1,Boris2}. 

The starting point and aim of the present work are rather
different from the above studies. Here, our scope is  to reveal and discuss the structure of the set of the possible equilibria of  (\ref{eq0}), which serve as potential attractors \cite{temam} for its  dynamics when $\delta>0$.
Supplementing Eq.~(\ref{eq0}) with Dirichlet boundary conditions, the system belongs in the class of a  finite dimensional second-order gradient system, for which the Hamiltonian energy serves as a Lyapunov function. In the above discrete set-up, we use {\em a discrete variant of the \L{}ojasiewicz inequality} \cite{Jen1998,Jen2011}, to prove that all bounded solutions converge to a single equilibrium, i.e., a solution of the stationary problem:
\begin{eqnarray}
	\label{eqSP}
	-k(U_{n+1}-2U_n+U_{n-1}) = \omega_d^2 (U_{n}-\beta U_{n}^3).
\end{eqnarray}
Let us comment at this point on the strength of the \L{}ojasiewicz inequality approach, which seems to be applied for the first time in nonlinear lattices: the case of the DKG chain with a $\phi^4$ on-site potential is only a prototypical example of a second-order lattice dynamical system, on which the method is applicable. The approach may cover a wide class of  gradient dissipative lattices, involving analytical nonlinearities, hence we expect this technique
to be of considerably wider applicability than the specific application
selected for demonstration purposes herein.

While in terms of the topology of the phase space, the structure of the global attractor is trivial (since it consists of the single equilibrium), we show that the structure of the set of  equilibrium (steady-state) solutions, is
quite non-trivial. Our analytical considerations implement global bifurcation results \cite{RB71,Smo94} on  (\ref{eqSP}), considering $\omega_d^2$ as a bifurcation parameter, to rigorously prove that  each linear eigenstate of the system
(i.e., solution in the absence of nonlinearity) can be continued to a nonlinear counterpart. Such an analytical approach, offers the advantage of
characterizing the bifurcating equilibrium branches by the number of sign-changes of the associated equilibrium solutions.
This is a rather natural partition, given the self-adjoint nature
of the underlying linear operator and the considerations of Sturm-Liouville
theory in the one-dimensional setting of interest herein.
It is important to remark that due to these sign-changes,  the nonlinear equilibria, as elements of the bifurcating branches, may define nontrivial topological interconnections between the steady-states associated with the symmetric minima of the quartic potential. This provides, in turn, the connection of
the equilibrium states considered herein with multi-kink variants
of the (one- or two-kink) states that have been extensively studied in earlier works.

An important question, in the framework of the DKG system~(\ref{eq0}),  concerns the underlying mechanism
that leads to the selection of the eventual state of convergence. This question is investigated by analytical arguments corroborated by
direct numerical simulations. In  light of the global bifurcation analysis, such a combined approach wishes to examine the potential dynamical stability of the equilibrium branches, by considering two distinct scenarios for the initial conditions $U_n(0)$, and the bifurcation parameter $\omega_d^2$.  These scenarios, which both consider spatially extended initial conditions $U_n(0)$ (resembling linear eigenstates) and zero velocities $\dot{U}_n(0)=\bf{0}$, are as follows.

The first scenario (I), considers pairs $(||U_n(0)||, \omega_d^2)$, consisting of small values of $\omega_d^2$, and initial conditions whose norm defines a point of the local bifurcation diagram as follows: the initial condition has the same number of sign-changes with the one identifying a particular branch. We call such an initial condition, \emph{similar} to a branch of equilibria. Then, the point $(||U_n(0)||, \omega_d^2)$ is selected to be in a sufficiently small neighborhood of an equilibrium solution of the similar branch. 
Naturally, scenario~(I), is intended to take advantage of the local structure of the bifurcating equilibrium branches, together with the associated  linear stability analysis.  
We verify numerically --with an excellent agreement-- the analytical predictions on the geometric structure of the branches (according to the global bifurcation theorem \cite{RB71,Smo94}),  as well as their instability guiding the convergence dynamics. The latter is  manifested by metastability, with orbits connecting distinct equilibria,
transitioning progressively from more to less unstable states.
In that regard, the evolution of the Hamiltonian energy is an effective
diagnostic for the potential metastable dynamics.

The second scenario (II), considers an arbitrary value for the bifurcation parameter $\omega_d^2$, while the initial condition is still similar to a specific branch. However, its norm is
considerably larger, so that the relevant point $(||U_n(0)||, \omega_d^2)$, is far from the 
relevant local bifurcation diagram.  
In fact, with scenario~(II), we wish to initiate herein our studies
towards revealing the structure of the global bifurcation diagram and
the associated stability properties. 
Our numerical investigations consider fixed values of the  parameter $\omega_d^2$, just above the linear spectrum, while the amplitude of a similar to a branch, initial condition, is progressively decreased.  
Remarkably, decreasing the amplitude of the initial conditions, our findings showcase, at first, that they  may converge to geometrically distinct
equilibrium states,
of a non-similar branch without the metastable transition observed in the first scenario.
Second, decreasing the amplitude further, we reveal the existence of an intermediate amplitude-values interval of the initial condition, for which the latter converges interchangeably to an equilibrium of either a similar or a non-similar branch. 
Third, we illustrate the existence of (upper) thresholds for the amplitude of the initial condition,
below which, the traced orbit of the initial condition converges to an equilibrium of its similar branch. This is an indication that above the threshold values,
the similar branch may not be dynamically accessible (via a similar initial condition), and hence, the system
follows a different type of dynamics.

All the above findings, suggest that the global bifurcation diagram may be
rather complex. The summary of the above observations, is that the dynamics of Eq.~(\ref{eq0})--the prototypical dissipative DKG chain-- serves as a case example, for the demonstration of an 
energy landscape bearing an elaborate equilibrium set and  associated stability properties affecting the resulting dynamics.

The paper is structured as follows. In Section~II, we present
the analytical considerations on  Eq.~(\ref{eq0}), concerning the convergence to a single equilibrium.  
Section~III, is devoted to the analytical results, concerning the global bifurcation of nonlinear equilibria. In Section~IV, we report the results of our numerical simulations. Finally, in Section~V, we summarize and discuss the implications of our results with an eye towards future work.

\section{Convergence to nontrivial equilibria.}  
\setcounter{equation}{0}
\paragraph{Preliminaries.}
For the analysis of Eq.~(\ref{eq0}), we will consider an arbitrary number of $K+2$ oscillators equidistantly occupying an interval $[-\frac{L}{2}, \frac{L}{2}]$ of length $L$, with spacing $h=\frac{L}{K+1}$. Thus, the oscillators are occupying the points $x_n=-L/2+nh$, $n=0,1,2,\ldots,K+1$ of the interval $[-\frac{L}{2}, \frac{L}{2}]$, discretized as $-\frac{L}{2}=x_0<x_1<\ldots<x_{K+1}=\frac{L}{2}$. Then,  Eq. (\ref{eq0}), is written in the standard shorthand notation $U(x_n,t):=U_n(t)$. In some cases, we shall also use the shorthand notation $U$ for the vectors of $\mathbb{R}^{K+2}$, i.e., $U:=\left\{U_n\right\}_{n=0}^{K+1}$. For the DKG chain (\ref{eq0}), we will consider the initial-boundary value problem, with initial conditions
\begin{eqnarray}
	\label{eq02}
	U_n(0)=U_{n,0}\;\;\mbox{and}\;\; \dot{U}_n(0)=U_{n,1}\in\mathbb{R}^{K+2},
\end{eqnarray}
and Dirichlet boundary conditions at the endpoints $x_0=-L/2$ and $x_{K+1}=L/2$, namely:
\begin{eqnarray}
	\label{eq03}
	U_0=U_{K+1}=0,\;\;t\geq 0.
\end{eqnarray}
We shall also use, when convenient, the short-hand notation $\Delta_d$ for the
one-dimensional discrete Laplacian
\begin{eqnarray}
	\label{eqD}
	\left\{\Delta_dU\right\}_{n\in\mathbb{Z}}=U_{n+1}-2U_n+U_{n-1},
\end{eqnarray}
defined in our case in $\mathbb{R}^{K+2}$, which can be considered as the finite dimensional subspace of the sequence space $\ell^2$ of square summable sequences.  In other words, we shall work in  the finite dimensional subspaces of the sequence spaces: 
$\ell^p$, $1\leq p\leq\infty$,
\begin{eqnarray}
	\label{eq04}
	\ell^p_{K+2}=\left\{X\in\ell^p\;:\;X_0=X_{K+1}=0\right\}.
\end{eqnarray}
Clearly, $\ell^p_{K+2}\equiv \mathbb{R}^{K+2}$, which may be endowed with the norm
\begin{eqnarray*}
	||X||_{\ell^p}=\left(\sum_{n=0}^{K+1}|X_n|^p\right)^{\frac{1}{p}}.
\end{eqnarray*}
We recall the well known equivalence of norms,
\begin{eqnarray}
	\label{eq07}
	||X||_{\ell^q}\leq ||X||_{\ell^p}\leq(K+2)^{\frac{(q-p)}{qp}}||X||_{\ell^q},\;\;1\leq p\leq q<\infty.
\end{eqnarray}
We shall also denote by
\begin{eqnarray*}
	(X,Y)_{\ell^2}=\sum_{n=0}^{K+1}X_nY_n,\;\;||X||^2_{\ell^2}=\sum_{n=0}^{K+1}X_n^2,
\end{eqnarray*}
the squared-$\ell^2$ inner product and norm respectively, which will serve as the finite dimensional phase-space for the dynamical system defined by Eq.~(\ref{eq0}).
Note that in the case of the infinite lattice $\mathbb{Z}$, the following inequalities:
\begin{eqnarray}
	\label{eq05}
	||X||_{\ell^q}&\leq& ||X||_{\ell^p}, \;\;\;\;\;1\leq p\leq q\leq\infty\\
	\label{eq06}
	0&\leq& (-\Delta_d X, X)_{\ell^2}\leq 
	4 \sum_{n\in\mathbb{Z}}|X_n|^2,
\end{eqnarray}
hold, which are however valid in the above finite-dimensional setup. 
%
\paragraph{Proof of convergence to the equilibria.}
Proceeding to our proof, we start with some useful observations on the energy quantities possessed  by the problem (\ref{eq0})-(\ref{eq02})-(\ref{eq03}). First, we recall that when the dissipation parameter is $\delta=0$, the DKG system (\ref{eq0}) describes the equations of motion derived by the Hamiltonian:
%
\begin{eqnarray}
	\label{eq08}
	\mathcal{H}(t)=\frac{1}{2}\sum_{n=0}^{K+1}\dot{U}_n^2
	+\frac{k}{2}\sum_{n=0}^{K+1}(U_{n+1}-U_n)^2
	-\frac{\omega_d^2}{2}\sum_{n=0}^{K+1}U_n^2+\frac{\beta\omega_d^2}{4}\sum_{n=0}^{K+1} U_n^4, 
\end{eqnarray}
which is conserved, i.e., $\frac{d}{dt}\mathcal{H}(t)=0$. 
On the other hand, in the linearly damped case $\delta>0$, the Hamiltonian energy (\ref{eq08})  is dissipated according to the energy-balance law:
\begin{eqnarray}\label{eq09}
	\frac{d}{dt}\mathcal{H}(t)=\frac{d}{dt} \left[\frac{1}{2}\sum_{n=0}^{K+1}\dot{U}_n^2+\frac{k}{2} \sum_{n=0}^{K+1}(U_{n+1} - U_n)^2 -
	\frac{\omega_d^2}{2}\sum_{n=0}^{K+1} U_n^2 + \frac{\beta \omega_d^2}{4}
	\sum_{n=0}^{K+1}U_n^4\right] =- \delta \|\dot{U}\|_{\ell^2}^2.
\end{eqnarray}
It will be convenient for the computations, to consider from Eq. (\ref{eq09}), the functional
\begin{eqnarray}
	\label{eq010} 
	F(U_n) = -\frac{k}{2} \sum_{n=0}^{K+1}(U_{n+1} - U_n)^2 +
	\frac{\omega_d^2}{2}\sum_{n=0}^{K+1} U_n^2 - \frac{\beta \omega_d^2}{4}
	\sum_{n=0}^{K+1}U_n^4,
\end{eqnarray}
so that Eq.~(\ref{eq09}) can be rewritten as: 
\begin{eqnarray}
	\label{eq011}
	\frac{1}{2} \frac{d}{dt} \|\dot{U}\|_{\ell^2}^2 + \delta \langle\dot{U},\dot{U}\rangle _{\ell^2} = 
	\frac{d}{dt}F(U_n).
\end{eqnarray}
Since solutions of (\ref{eq0})-(\ref{eq02})-(\ref{eq03}) exist globally in time \cite{PhysD2013},  we may integrate Eq.~(\ref{eq011}) in the interval $[0,t]$, for arbitrary $t\in [0,\infty)$, to  get its integral form  
\begin{eqnarray}
	\label{eq012}
	\frac{1}{2}\|\dot{U}(t)\|_{\ell^2}^2 - \frac{1}{2}\|\dot{U}(0)\|_{\ell^2}^2 + 
	\delta \int_0^t \langle\dot{U}(s),\dot{U}(s)\rangle ds = F(U_n(t)) - F(U_{n,0}).
\end{eqnarray}
\begin{figure}[tbp]
	\centering
	\includegraphics[scale=0.09]{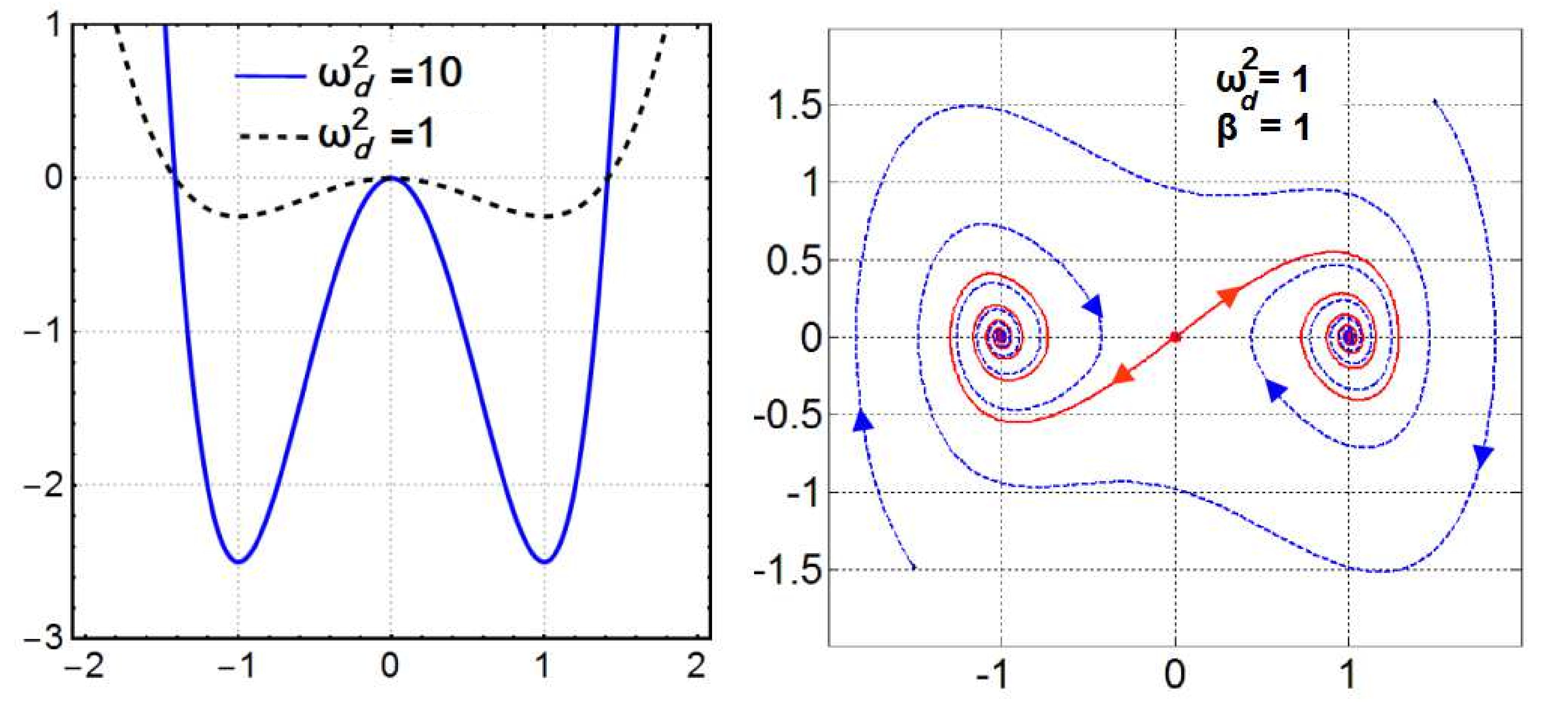} \vspace{0cm}
	\caption{(Color Online) Left panel: Graphs of the potential energy (\ref{eq01}) for $\omega_d^2=10$ [continuous (blue) curve] and $\omega_d^2=1$ [dashed (black) curve)], when $\beta=1$. Right Panel: Dynamics in the case of the anti-continuum limit $k=0$, when $\omega_d^2=1$, $\beta=1$, $\delta=0.01$: each individual oscillator $U_n$ is governed by a linearly damped Duffing equation. The continuous (red) orbit defines the global attractor (an heteroclinic connection, between the symmetric asymptotically stable fixed points, located at the two minima of the quartic potential.)}
	\label{Figu1}
\end{figure}
The equilibria for the problem (\ref{eq0})-(\ref{eq02})-(\ref{eq03}), have the form
$\overline{\Phi}=(\Phi,\mathbf{0})\in \ell^2_{K+2}\times \ell^2_{K+2}$ where $\Phi=\left(0,\Phi_1, \Phi_2,\ldots, \Phi_K, 0\right)$, in order to satisfy also the Dirichlet boundary conditions (\ref{eq03}). In addition, $\Phi_n$ must satisfy the stationary problem
\begin{eqnarray}
	\label{eq063}
	-k\Delta_d \Phi_n &=& \omega_d^2 \Phi_{n}-\beta\omega_d^2 \Phi_{n}^3,\\[2ex]
	\label{eq063A}
	\Phi_{0}&=&\Phi_{K+1}=0.
\end{eqnarray}
In particular, in the uncoupled (alias anti-continuum) limit $k=0$ the equilibrium states of the system are defined through the equilibria of the individual oscillators of the lattice. In this limit, each oscillator is governed by the linearly damped Duffing equation $$\ddot{U}_n+ \delta \dot{U}_n=\omega_d^2U_n-\omega_d^2\beta U_n^3.$$ We recall their dynamics in the right panel of Fig.~\ref{Figu1} (the left panel of this figure depicts the potential function for two different values of $\omega_d^2$). As it is evident, every oscillator possesses three distinct
equilibrium configurations induced  by the potential (\ref{eq01})  --cf. the left panel of Fig.~\ref{Figu1}: the  unstable maximum
at $U_{\mathrm{max}}^0=0$, corresponding to the rest energy $E_{\mathrm{max}}=W(0)=0$, and two stable minima located at $U^{\mp}_{\mathrm{min}}=\mp\frac{1}{\sqrt{\beta}}$, corresponding to $E_{\mathrm{min}}=W\left(\mp\frac{1}{\sqrt{\beta}}\right)=\frac{\omega_d^2}{4\beta}$. Let the set of these three equilibrium configurations be $\Phi^*=\{U_\mathrm{min}^-, U_\mathrm{max}^0, U_\mathrm{min}^+\}$. Then, in the uncoupled limit of $k=0$, every configuration 
\begin{eqnarray}
	\label{eq013}
	\Phi=\left(0,\Phi_1, \Phi_2,\ldots, \Phi_K, 0\right),\ \text{with}\ \ \Phi_n \in \Phi^*\ \ \text{and}\ \ n=1\ldots K,
\end{eqnarray}   
corresponds to an equilibrium of (\ref{eq0})-(\ref{eq02})-(\ref{eq03}).
It is expected that these configurations can typically
be continued for $k\neq0$ to form mono-parametric families of
equilibria with $k$ as a parameter, but this bifurcation problem
will not be considered in the present work. 

On the other hand, in the linear limit $\omega_d=0$, the trivial configuration  
$$\Phi=\left(0, 0,\ldots, 0, 0 \right)\in\ell^2_{K+2},$$
corresponds to a solution of (\ref{eq063})-(\ref{eq063A}), and thus, is an equilibrium  of (\ref{eq0})-(\ref{eq02})-(\ref{eq03}) for every value of $k$. The families of equilibria which bifurcate from the zero solution will be a subject of an extended study later in this work. Here, we have to note that there may be also solutions of the stationary problem, which do not belong to the configurations that bifurcate from the limiting cases mentioned above. 

At this point, we should remark that upon supplementing the chain with Dirichlet boundary conditions, 
the system is forced to maintain two oscillators at the unstable state of the local maximum of the on-site potential. Therefore, 
the lattice is driven out of equilibrium 
by this ``unconventional'' choice, with respect to the stationary states induced by the double-well potential (\ref{eq01}). 

Due to Eq.~(\ref{eq09}), the Hamiltonian energy $\mathcal{H}(t)$ defines a Lyapunov function for the problem (\ref{eq0})-(\ref{eq02})-(\ref{eq03}). Then by \cite[Theorem 4.1, pg. 491]{temam}, all the above equilibria are included in the $\omega$-limit set of the flow $S(t):\ell^2\times\ell^2\rightarrow \ell^2\times\ell^2$.  We recall that if $\mathcal{B}\subseteq \ell^2\times\ell^2$ is a bounded set, then
\begin{eqnarray}
	\label{limset}
	\omega(\mathcal{B})=\left\{\overline{\Phi}_*\;\;:\;\;\exists\;\;\mbox{sequence $t_n$ such that $t_n\rightarrow\infty$ as $n\rightarrow\infty$, and}\;\; \lim_{n\rightarrow\infty}S(t_n)\overline{\Phi}_0=\overline{\Phi}_*,\;\;\forall \overline{\Phi}_0\in\mathcal{B} \right\}.
\end{eqnarray}	

Therefore, potential convergence of solutions of the problem (\ref{eq0})-(\ref{eq02})-(\ref{eq03}) to an equilibrium, should be associated with the convergence 
$\dot{U}(t)\rightarrow \mathbf{0}$ as $t\rightarrow\infty$, for the velocities of the oscillators.  The following Proposition, assures that the above convergence holds for any solution of (\ref{eq0})-(\ref{eq02})-(\ref{eq03}). 
\begin{proposition}
	\label{L01}. Let $\overline{U}(t)=(U(t), \dot{U}(t))\in \mathrm{C}([0,\infty),\ell^2_{K+2}\times\ell^2_{K+2})$, be the unique solution of the initial-boundary value problem  
	(\ref{eq0})-(\ref{eq02})-(\ref{eq03}), for any initial data $\overline{U}(0)=(U_{n,0}, U_{n,1})\in \ell^2_{K+2}\times\ell^2_{K+2}$. Then,
	\begin{eqnarray}
		\label{eq014} 
		\lim_{t\to +\infty} \|\dot{U}(t)\|_{\ell^2} = 0.
	\end{eqnarray} 
\end{proposition} 
\begin{proof}{}
	Since the unique solution $\overline{U}(t)$ is uniformly bounded in $\mathrm{C}([0,\infty),\ell^2_{K+2}\times\ell^2_{K+2})$, there exist constants $c_1, c_2>0$, such that
	\begin{eqnarray}
		\label{eq014A}
		\|U\|_{\ell^2}\leq c_1,\;\mbox{and}\; \|\dot{U}\|_{\ell^2}\leq c_2, 
	\end{eqnarray} 
	respectively.  
	We will verify first, that the above uniform bounds, imply further, that  $\ddot{U}\in \mathrm{C} \left([0,\infty), \ell^2\right)$ is also uniformly bounded in time, i.e., 
	$\|\ddot{U}\|_{\ell^2}\leq c$, for some constant $c>0$. 
	Indeed, let us rewrite the Eq.~(\ref{eq0}) as
	\begin{eqnarray}
		\label{eq015}
		\ddot{U}_n = -\delta \dot{U}_n + k (U_{n+1} - 2U_n + U_{n-1}) + \omega_d^2 U_{n}-\beta
		\omega_d^2 U_{n}^3.
	\end{eqnarray} 
	Taking norms in Eq.~(\ref{eq015}), we get that
	\begin{eqnarray}\label{eq016}
		\|\ddot{U}\|_{\ell^2} = \|-\delta\dot{U}+k \Delta_d U + f(U)\|_{\ell^2},
	\end{eqnarray}
	where
	\begin{eqnarray}
		\label{eq017}
		f(U_n)=\omega_d^2 U_{n}-\beta \omega_d^2 U_{n}^3.
	\end{eqnarray} 
	Now, from (\ref{eq016}), and the boundedness of $-\Delta_d:\ell^2_{K+2}\to\ell^2_{K+2}$, asserting that
	\begin{eqnarray}
		\label{eq018}
		||-\Delta_d X||_{\ell^2}^2\leq 4||X||_{\ell^2}^2,\;\;\mbox{for all}\;\;X\in \ell^2_{K+2}, 
	\end{eqnarray}
	we have the inequality
	\begin{eqnarray}
		\label{eq019}
		\|\ddot{U}\|_{\ell^2} &=& \|-\delta\dot{U}+k \Delta_d U + f(U)\|_{\ell^2}\nonumber\\[2ex]
		&\leq& \delta\|\dot{U}\|_{\ell^2} + k \|\Delta_d U\|_{\ell^2} + \|f(U)\|_{\ell^2}\nonumber\\[2ex]
		&\leq& \delta \|\dot{U}\|_{\ell^2} + 2k\|U\|_{\ell^2} + \|f(U)\|_{\ell^2}\nonumber\\[2ex]
		&\leq&\delta c_2+2kc_1+\|f(U)\|_{\ell^2}.
	\end{eqnarray}
	Furthermore, by taking the $\ell^2_{K+2}$-norm of the on-site forces (\ref{eq017}), we get the inequality
	\begin{eqnarray}
		\label{eq020}
		\|f(U)\|_{\ell^2}^2 &=& \omega_d^4 \sum_{n=0}^{K+1} (U_{n}-\beta U_{n}^3)^2 
		= \omega_d^4 \|U-\beta U^3 \|_{\ell^2}^2 \nonumber\\[2ex]
		&\leq& \omega_d^4\left(\|U\|_{\ell^2}^2 + \beta \|U\|_{\ell^2}^3\right)^2\nonumber\\[2ex]
		&\leq& \omega_d^4(c_1^2 + \beta c_1^3)^2:=c_3.
	\end{eqnarray}
	Note, that for the derivation of the estimate (\ref{eq020}),
	we have used (\ref{eq07}) for $q=6$ and $p=2$: it allows to estimate the norm of the  cubic term in (\ref{eq020}) by the inequality
	$\left(\sum_{n=0}^{K+1} U_n^6\right)\leq
	\left(\sum_{n=0}^{K+1} U_n^2\right)^{3}$.
	Inserting (\ref{eq020}) into (\ref{eq019}), we get that 
	$\|\ddot{U}\|_{\ell^2}\leq c:=\delta c_2+2kc_1+c_3$. 
	This uniform bound on $\|\ddot{U}\|_{\ell^2}$, implies that $\dot{U}(t)$ is uniformly continuous in $[0,\infty)$ as follows: 
	by using the mean value theorem for the function $V(t) = \|\dot{U}(t)\|^{2}$, we have 
	that for arbitrary $t_1\geq 0$ , $t_2>0 \in\bbR ^{+} $,  there exists $t^*\in(t_1, t_2)$, such that 
	\begin{eqnarray}\label{eq021}
		|V(t_1) - V(t_2)|\leq |\dot{V}(t^*)|\, |t_1 - t_2|.
	\end{eqnarray}
	Besides,  for the time-derivative of $V(t)$, given by $\dot{V}(t)  
	= 2 \langle \ddot{U}(t),\dot{U}(t)\rangle $, we get, by applying the Cauchy-Schwartz inequality, that
	\begin{eqnarray}\label{eq022}
		|\dot{V}(t)| = 2 |\langle \ddot{U}(t^*),\dot{U}(t^*)\rangle | \leq 2\|\ddot{U}(t^*)\|_{\ell^2}\;\|\dot{U}(t^*)\|_{\ell^2}\leq m:=2c\;c_2
	\end{eqnarray} 
	Since the interval $(t_1, t_2)$  is arbitrary, it readily follows from (\ref{eq022}), that the function $\dot{V}(t)$ is bounded and that 
	\begin{eqnarray}\label{eq023}
		|V(t_1)-V(t_2)| = \big| \|\dot{U}(t_1)\|_{\ell^2}^2 - \|\dot{U}(t_2)\|_{\ell^2}^2 \big|\leq m|t_1 - t_2|.
	\end{eqnarray}
	Consequently, $V(t)$ is globally \tl{Lipschitz}  continuous and uniformly bounded in $[0,\infty)$, thus locally integrable in $[0,\infty)$.  Furthermore, we may derive, by using (\ref{eq012}), the boundedness of $U$ and $\dot{U}$ in $\ell^2_{K+2}$ (\ref{eq014A}), and the bound (\ref{eq020}), the  following inequality: 
	\begin{eqnarray}
		\label{eq024}
		\delta \int_0^t \langle\dot{U}(s),\dot{U}(s)\rangle ds \leq \frac{1}{2}\big|\|\dot{U}(t)\|_{\ell^2}^2 - \|\dot{U}(0)\|_{\ell^2}^2\big| + |F(U_n(t)) - F(U_{n,0})|<C,
	\end{eqnarray}
	where $C$ is a constant independent of $t$. Letting $t\rightarrow\infty$ in (\ref{eq024}), implies
	the integrability of $V(t)$ in $[0,\infty)$ and, as a result, the claim (\ref{eq014}).
\end{proof}

The rest of the section, is devoted in showing that $\lim_{t\rightarrow\infty}||U(t)-\Phi||_{\ell^2}=0$, which will complete (on the account of Proposition \ref{L01}),  the proof of the convergence of the solutions of (\ref{eq0})-(\ref{eq02})-(\ref{eq03}), to equilibrium. We should exploit the information on the limiting behavior of $\dot{U}$, e.g. (\ref{eq014}), in considering the difference $U-\Phi$ for large times, in terms of the obvious equation
\begin{eqnarray}
	\label{obvi1}
	U(t)-\Phi = U(t_N)-\Phi + \int_{t_N}^{t}\frac{d}{ds}(U(s)-\Phi) ds,
\end{eqnarray}
for some finite $t_N\in [0,\infty)$. Consequently, the  $\ell^2$-energy of the difference $U(t)-\Phi$, satisfies the inequality
\begin{eqnarray}
	\label{obvi2}
	\|U(t)-\Phi\|_{\ell^2} 
	\leq \|U(t_N)-\Phi\|_{\ell^2} + \int_{t_N}^{t}\|\frac{d}{ds}(U(s)-\Phi)\|_{\ell^2} ds.
\end{eqnarray}
Then, the claim is that passing to the limit as $t\rightarrow\infty$, both terms of the right hand side of (\ref{obvi2}), become arbitrarily small, establishing the convergence of $U$ to $\Phi$. For the first term $\|U(t_N)-\Phi\|_{\ell^2}$, the definition of the $\omega$-limit set, will allow for the selection of a sufficiently large $t_N$, so that $U(t_N)$ will be arbitrarily close to $\Phi$. It should be clear that the existence of such $U(t_N)$-orbit points, does not guarantee itself, the convergence of the whole orbit to $\Phi$, since the orbit may escape from the vicinity of $\Phi$, for $t> t_N$. However, the latter scenario will be excluded, by establishing that the second term becomes arbitrarily small, for sufficiently large times.   To this end, we will control its growth by suitable energy estimates. Such estimates, should naturally involve the Hamiltonian of the conservative limit $\delta=0$ of the DKG lattice (which is the Lyapunov-like functional for the full system), as follows.

For convenience, we rewrite the Hamiltonian energy (\ref{eq08}) in the shorthand notation
\begin{eqnarray}\label{eq025}
	\mathcal{H}(t)=\mathcal{H}(U(t),\dot{U(t)}) =\frac{1}{2} \|\dot{U}\|_{\ell^2}^2 - F(U),
\end{eqnarray}
where $F(U)$, denotes the functional (\ref{eq010}). Notice that for an orbit $(U(t), \dot{U}(t))$ converging to a single equilibrium configuration $(\Phi,\mathbf{0})$, its Hamiltonian energy (being a continuous functional itself), should converge as $\mathcal{H}(U(t)),\dot{U}(t)\rightarrow\mathcal{H}(\Phi,\mathbf{0})=-F(\Phi)$, for $t\rightarrow\infty$. Thus, as a first step, it is natural to consider the difference of Hamiltonian energies 
\begin{eqnarray}
	\label{eq270a}
	\mathcal{H}(U(t),\dot{U}(t))-\mathcal{H}(\Phi,\mathbf{0})= \frac{1}{2} \|\dot{U}\|_{\ell^2}^2 - (F(U)-F(\Phi)).
\end{eqnarray}
Since $\lim_{t\to +\infty} \|\dot{U}(t)\|_{\ell^2} = 0$, 
the second term in the difference (\ref{eq270a}) should converge as  $F(U)\rightarrow F(\Phi)$, for $t\rightarrow\infty$, suggested also by the continuity of the functional $F$.  Therefore, we are guided to handle the second term of (\ref{obvi2}), by the difference (\ref{eq270a}). However, the fact that $F$ is locally Lipschitz will, potentially, provide solely a linear differential inequality on the considered difference $||U(t)-\Phi||_{\ell^2}$ (of the unknown limiting behavior); this inequality, however, is insufficient to establish convergence of $U$ to $\Phi$. 
At this point, the \L{}ojasiewicz-type inequality \cite{Jen1998,Jen2011} comes into play, in order to suggest an appropriate perturbation of the difference of the Hamiltonian energies (\ref{eq270a}). This inequality  involves the functional $F$, and the nonlinear operator 
\begin{eqnarray}
	\label{eq026}
	J(U_n) := k(\Delta_d U)_n +\omega_d^2 U_{n}-\beta\omega_d^2  U_{n}^3,
\end{eqnarray}
with shorthand notation $J(U)$. According to the results discussed in the Appendix \ref{appen1}, the \L{}ojasiewicz inequality in our discrete setting, is stated in
\begin{lemma}
	\label{LINEQ}  
	There exists $\tilde{\epsilon}>0$ and $0<\theta<1/2$, such that
	\begin{eqnarray}
		\label{LSin}
		||J(U)||_{\ell^2}\geq \nu_0|F(U)-F(\Phi)|^{1-\theta},
	\end{eqnarray}
	for all $U\in\ell^2_{K+2}$ such that $\|U-\Phi\|_{\ell^2}<\tilde{\epsilon}$, and some constant $\nu_0>0$.
\end{lemma}
In the light of Lemma \ref{LINEQ}, we shall consider a perturbation of the difference of the Hamiltonian energies (\ref{eq270a}), defined as
\begin{eqnarray}\label{eq027}
	E(t) = \frac{1}{2} \|\dot{U}\|_{\ell^2}^2 - (F(U)-F(\Phi)) - \varepsilon \langle J(U),\dot{U}\rangle_{\ell^2},\;\;\varepsilon>0,
\end{eqnarray}

Let us now elucidate further the role of the perturbation term $-\varepsilon \langle J(U),\dot{U}\rangle_{\ell^2}$ in (\ref{eq027}), namely the inner-product of $J(U)$ and $\dot{U}$. The norm $||\dot{U}|||_{\ell^2}$ becomes arbitrary small for large times, while $J(U)$ is uniformly bounded
in $\ell^2$.  Furthermore, observe from the inequality (\ref{LSin}), that $||J(U)||_{\ell^2}$ is bounded from below, from the $1-\theta$-power of the difference $|F(U)-F(\Phi)|$. As will be proved in the sequel by the energy-estimates algebra, this lower bound allows to derive a key-estimate  of $E(t)$ [see (\ref{ALI})] only in terms of the bounded quantity $||J(U)||_{\ell^2}$, and  the asymptotically vanishing $||\dot{U}||_{\ell^2}$; this way, the problematic difference $|F(U)-F(\Phi)|$ is eliminated. Elaborating this key-estimate further, we may bound the second term  of (\ref{obvi2}) only in terms of $E(t)$, which is also proved to have an additional property: If $E(t)>0$, for all $t>0$, then $\lim_{t\rightarrow\infty}E(t)=0$ and, hence, the second term of (\ref{obvi2}) becomes arbitrarily small for large times. Notice that if 
$E(t)\leq 0$ for $t\geq t_0$, convergence easily follows in a straightforward manner.

The above procedure will be completed in various steps. The first step concerns the proof of two useful differential inequalities on $\dot{E}(t)$ stated in the following Lemmas. 
\begin{lemma}
	\label{L02}. 
	The derivative $\dot{E}(t)$ of the functional (\ref{eq027}) satisfies the inequality
	\begin{eqnarray}\label{eq028}
		\dot{E}(t)\leq-\delta \|\dot{U}\|^2_{\ell^2}+\frac{\rho^2}{2}\|\dot{U}\|^2_{\ell^2}
		-\varepsilon\|J(U)\|_{\ell^2}^2 + \varepsilon
		\delta\|J(U)\|_{\ell^2} \ \|\dot{U}\|_{\ell^2},
	\end{eqnarray}
	for any solution of the problem (\ref{eq0})-(\ref{eq02})-(\ref{eq03}).
\end{lemma}
\begin{proof}
	We start by differentiating Eq.~(\ref{eq027}), with respect to time. Using (\ref{eq09}) and (\ref{eq012}), we see that $\dot{E}(t)$ satisfies
	\begin{eqnarray}\label{eq029}
		\dot{E}(t) = \frac{d}{dt}\left[\frac{1}{2} \|\dot{U}\|_{\ell^2}^2 - (F(U)-F(\Phi))\right]- \varepsilon 
		\frac{d}{dt}\langle J(U),\dot{U}\rangle_{\ell^2} = -\delta \langle \dot{U},\dot{U}\rangle_{\ell^2} - \varepsilon 
		\frac{d}{dt}\langle J(U),\dot{U}\rangle_{\ell^2}.
	\end{eqnarray}
	We will work on the second term of the right-hand side of (\ref{eq029}), which is 
	\begin{eqnarray}\label{eq030}
		\frac{d}{dt}\langle J(U),\dot{U}\rangle_{\ell^2} = \langle \dot{J}(U),\dot{U}\rangle_{\ell^2} + 
		\langle J(U),\ddot{U}\rangle_{\ell^2}.
	\end{eqnarray}
	Substitution of $\ddot{U}$ from $\ddot{U}=-\delta\dot{U}+J(U)$ [cf. Eq.~(\ref{eq0}) in short-hand notation] into (\ref{eq030}), gives 
	\begin{eqnarray}
		\label{eq031}
		\frac{d}{dt}\langle J(U),\dot{U}\rangle_{\ell^2} &=& \langle \dot{J}(U),\dot{U}\rangle_{\ell^2} + \langle J(U), 
		-\delta\dot{U} + J(U)\rangle_{\ell^2}\nonumber\\[2ex] &=&  \langle \dot{J}(U),\dot{U}\rangle_{\ell^2} + 
		\langle J(U),J(U)\rangle_{\ell^2} - \delta\langle J(U), \dot{U} \rangle_{\ell^2}\nonumber\\[2ex]
		&=& \langle \dot{J}(U),\dot{U}\rangle_{\ell^2}+ \|J(U)\|_{\ell^2}^2 - \delta\langle J(U), \dot{U} \rangle_{\ell^2}.
	\end{eqnarray}
	Inserting (\ref{eq031}) to (\ref{eq029}), we arrive to the equation
	\begin{eqnarray}\label{eq032}
		\dot{E}(t) = -\delta \langle \dot{U},\dot{U}\rangle_{\ell^2}- \varepsilon
		\langle \dot{J}(U),\dot{U}\rangle_{\ell^2}- \varepsilon\|J(U)\|_{\ell^2}^2 + \varepsilon
		\delta\langle J(U), \dot{U} \rangle_{\ell^2}.
	\end{eqnarray}
	We proceed, by estimating the last term of the right-hand side of 
	(\ref{eq032}) upon applying the 
	Cauchy - Schwartz inequality: 
	$$ |\langle J(U), \dot{U} \rangle_{\ell^2}| \leq \|J(U)\|_{\ell^2} \ \|\dot{U}\|_{\ell^2}. $$
	With the latter estimate, Eq.~(\ref{eq032}) becomes the inequality
	\begin{eqnarray}\label{eq033}
		\dot{E}(t)\leq -\delta \langle \dot{U},\dot{U}\rangle_{\ell^2}- \varepsilon
		\langle \dot{J}(U),\dot{U}\rangle_{\ell^2}- \varepsilon\|J(U)\|_{\ell^2}^2 + \varepsilon
		\delta\|J(U)\|_{\ell^2} \ \|\dot{U}\|_{\ell^2}.
	\end{eqnarray}
	Similarly, we estimate the quantity $\varepsilon
	\langle \dot{J}(U),\dot{U}\rangle_{\ell^2}$ appearing in the right-hand side of (\ref{eq033}), as
	\begin{eqnarray}\label{eq034}
		\varepsilon|\langle \dot{J}(U),\dot{U}\rangle_{\ell^2}| \leq \varepsilon \|\dot{J}(U)\|_{\ell^2} \ 
		\|\dot{U}\|_{\ell^2}.
	\end{eqnarray}
	Since $||\dot{U}||_{\ell^2}$ is uniformly bounded, to estimate further the quantity $\varepsilon
	\langle \dot{J}(U),\dot{U}\rangle_{\ell^2}$, we need an estimate for $\|\dot{J}(U)\|_{\ell^2}$. First, from the definition of the functional $J(U)$ in (\ref{eq026}), its time derivative is found to be
	\begin{eqnarray}
		\label{eq035}
		\dot{J}(U_n) = k \Delta_d \dot{U}_n + \omega_d^2 \dot{U}_n  - 
		3\beta \omega_d^2 U_n^2\dot{U}_n.
	\end{eqnarray}
	Note, that by implementing the embedding inequality (\ref{eq05}) with $p=\infty$ and $q=2$,
	\begin{eqnarray}
		\label{eq035A}
		||U^2\dot{U}||_{\ell^2}^2=\sum_{n=0}^{K+1}U_n^4\dot{U}_n^2\leq ||U||_{\ell^{\infty}}^4\sum_{n=0}^{K+1}\dot{U}_n^2\leq ||U||_{\ell^2}^4\sum_{n=0}^{K+1}\dot{U}_n^2=||U||_{\ell^2}^4||\dot{U}||_{\ell^2}^2,
	\end{eqnarray}
	we may observe that (\ref{eq035}) can be estimated as
	\begin{eqnarray}
		\label{eq036}
		\|\dot{J}(U)\|_{\ell^2}&\leq& k \|\Delta_d \dot{U}\|_{\ell^2} + \omega_d^2 \|\dot{U}\|_{\ell^2} + 
		3\beta \omega_d^2 \|U^2\dot{U}\|_{\ell^2}\nonumber\\[2ex]
		&\leq& 2k\|\dot{U}\|_{\ell^2} + \omega_d^2 \|\dot{U}\|_{\ell^2} + 3\beta \omega_d^2 
		\|U\|_{\ell^2}^2 \ 
		\|\dot{U}\|_{\ell^2}\nonumber\\[2ex] &\leq& c_4 \|\dot{U}\|_{\ell^2}, \qquad c_4=2k+\omega_d^2(3\beta c_1^2+1).
	\end{eqnarray}
	Multiplying (\ref{eq036}) by $\varepsilon$, we get that 
	\begin{eqnarray}\label{eq037}
		\varepsilon\|\dot{J}(U)\|_{\ell^2}\leq \frac{\rho^2}{2}||\dot{U}||_{\ell^2},\;\;\frac{\rho^2}{2}:=\varepsilon c_4.
	\end{eqnarray}
	Inserting (\ref{eq037}) into (\ref{eq034}), we derive the estimate for $\varepsilon|\langle \dot{J}(U),\dot{U}\rangle_{\ell^2}|$, 
	\begin{eqnarray}
		\label{eq038}
		\varepsilon|\langle \dot{J}(U),\dot{U}\rangle_{\ell^2}|\leq\frac{\rho^2}{2}\|\dot{U}\|^2_{\ell^2}.
	\end{eqnarray}
	Thus, by using (\ref{eq038}) in (\ref{eq033}), we conclude with
	$$\dot{E}(t)\leq-\delta \langle \dot{U},\dot{U}\rangle_{\ell^2}+\frac{\rho^2}{2}\|\dot{U}\|^2_{\ell^2}
	-\varepsilon\|J(U)\|_{\ell^2}^2 + \varepsilon
	\delta\|J(U)\|_{\ell^2} \ \|\dot{U}\|_{\ell^2},$$
	which is the claimed (\ref{eq028}). \end{proof}

Lemma \ref{L02}  will be used as an auxiliary tool for the proof of the second inequality on $\dot{E}(t)$, which plays a key-role on the proof of convergence $\lim_{t\rightarrow\infty}||U(t)-\overline{\Phi}||_{\ell^2}=0$. This second inequality is stated in
\begin{lemma}
	\label{L03}. 
	Let $\overline{U}(t)=(U(t), \dot{U}(t))\in \mathrm{C}([0,\infty),\ell^2_{K+2}\times\ell^2_{K+2})$, be the unique solution of the initial-boundary value problem  
	(\ref{eq0})-(\ref{eq02})-(\ref{eq03}), for any initial data $\overline{U}(0)=(U_{n,0}, U_{n,1})\in \ell^2_{K+2}\times\ell^2_{K+2}$. Then,  this solution satisfies the differential inequality
	\begin{eqnarray}\label{eq039}
		\dot{E}(t)\leq-\frac{\varepsilon}{4}\left(\|\dot{U}\|_{\ell^2}+\|J(U)\|_{\ell^2}\right)^2,\ \ \ 
		\forall t\geq0.
	\end{eqnarray}
\end{lemma}
\begin{proof}
	Inequality (\ref{eq039}) will be the outcome of a further estimation process, on the inequality (\ref{eq028}). Indeed, starting the process by handling the last term of (\ref{eq028}), we have (applying Young's inequality), that
	$$ \varepsilon \delta\|J(U)\|_{\ell^2} \ \|\dot{U}\|_{\ell^2}\leq 
	\frac{\varepsilon}{2}\|J(U)\|_{\ell^2}^2 + \frac{\delta^2 \varepsilon}{2}\|\dot{U}\|_{\ell^2}^2.$$
	Then, the inequality (\ref{eq028}) becomes 
	\begin{eqnarray}\label{eq040}
		\dot{E}(t)&\leq&-\delta \|\dot{U}\|^2_{\ell^2}+\frac{\rho^2}{2}\|\dot{U}\|^2_{\ell^2}
		-\varepsilon\|J(U)\|_{\ell^2}^2+\frac{\varepsilon}{2}\|J(U)\|_{\ell^2}^2 + 
		\frac{\delta^2 \varepsilon}{2}\|\dot{U}\|_{\ell^2}^2\nonumber\\[2ex]
		&=&\left(\frac{-2\delta+\rho^2+\delta^2\varepsilon}{2}\right) \|\dot{U}\|_{\ell^2}^2
		-\frac{\varepsilon}{2}\|J(U)\|_{\ell^2}^2\nonumber\\[2ex]
		&=&\left(\frac{-2\delta+2\varepsilon c_4+\delta^2\varepsilon}{2}\right) \|\dot{U}\|_{\ell^2}^2
		-\frac{\varepsilon}{2}\|J(U)\|_{\ell^2}^2,
	\end{eqnarray}
	by the definition of the constant $\rho^2=2\varepsilon c_4$ in (\ref{eq037}). 
	Requiring the coefficient of $||\dot{U}||_{\ell^2}^2$ in the right-hand side of (\ref{eq040}) to be negative, we select $\varepsilon$ so that 
	\begin{eqnarray}
		\label{eq041}
		0<\varepsilon<\frac{2\delta}{\delta^2+2c_4},
	\end{eqnarray}
	and thus  
	\begin{eqnarray}
		\label{eq043}
		\left(\frac{-2\delta+2\varepsilon c_4+\delta^2\varepsilon}{2}\right):= - \varepsilon_{1}<0.
	\end{eqnarray}
	Then, inequality (\ref{eq040}) becomes
	\begin{eqnarray}
		\label{eq044}
		\dot{E}(t)\leq -\varepsilon_{1}\|\dot{U}\|_{\ell^2}^2-\frac{\varepsilon}{2}\|J(U)\|_{\ell^2}^2.
	\end{eqnarray}
	In (\ref{eq044}), we further require  
	$\varepsilon_1>\varepsilon/2$, which results in the final assumption for $\varepsilon$:  
	\begin{eqnarray}
		\label{eq045}
		0<\varepsilon<\frac{2\delta}{\delta^2+2c_4+1}.
	\end{eqnarray}
	Let us note, that if (\ref{eq045}) holds, then (\ref{eq041}) is also readily satisfied. With the restriction (\ref{eq045}) at hand, (\ref{eq044}) implies the claimed (\ref{eq039})
	\begin{eqnarray}
		\label{eq046}
		\dot{E}(t)\leq -\frac{\varepsilon}{2}\left(\|\dot{U}\|_{\ell^2}^2+\|J(U)\|_{\ell^2}^2\right)\leq -\frac{\varepsilon}{4}\left(\|\dot{U}\|_{\ell^2}+\|J(U)\|_{\ell^2}\right)^2,
	\end{eqnarray}
	for which the inequality 
	$\left(\|\dot{U}\|_{\ell^2}^2+\|J(U)\|_{\ell^2}^2\right)\geq
	\frac{1}{2}\left(\|\dot{U}\|_{\ell^2}+\|J(U)\|_{\ell^2}\right)^2$ has been also used. 
\end{proof}
%
After the above preparations, we proceed to the proof of
\begin{theorem}
	\label{L04}	
	Let $\overline{U}(t)=(U(t), \dot{U}(t))\in \mathrm{C}([0,\infty),\ell^2_{K+2}\times\ell^2_{K+2})$, be the unique solution of the initial-boundary value problem  
	(\ref{eq0})-(\ref{eq02})-(\ref{eq03}), for any initial data $\overline{U}(0)=(U_{n,0}, U_{n,1})\in \ell^2_{K+2}\times\ell^2_{K+2}$. Then,
	\begin{eqnarray}
		\label{eq047} 
		\lim_{t\to +\infty} \|U(t)-\Phi\|_{\ell^2} = 0,
	\end{eqnarray} 
	where $\Phi$ is a solution of Eq. (\ref{eq063})-(\ref{eq063A}).
\end{theorem} 
\begin{proof}	
	The proof of (\ref{eq047}) is using the assumption $E(t)\geq 0$, $\forall t\geq 0$ (we only consider $E(t)>0$ since the case $E(t)=0$ is straightforward).  
	As we will see below, Lemma \ref{LINEQ} suggests to use the quantity 
	$[E(t)]^{1-\theta}$ [recall the definition of $E(t)$ in (\ref{eq027})]. We observe that it satisfies
	\begin{eqnarray}\label{eq048}
		[E(t)]^{1-\theta}=\left[\frac{1}{2}\|\dot{U}\|_{\ell^2}^2 - (F(U)-F(\Phi)) - \varepsilon 
		\langle J(U),\dot{U}\rangle \right]^{1-\theta}\leq
		\left[\frac{1}{2}\|\dot{U}\|_{\ell^2}^2 +|F(U)-F(\Phi)| + \varepsilon|\langle J(U),\dot{U}\rangle|
		\right]^{1-\theta}.
	\end{eqnarray}
	We  shall use the inequality $(a+b)^r\leq K_r(a^r+b^r)$, for some $K_r>0$ and all $a,b,r>0$ in the case $r=1-\theta$, to estimate the right-hand side of (\ref{eq048}). We get that
	\begin{eqnarray}\label{eq049}
		[E(t)]^{1-\theta}
		\leq K_{\theta}\left(\|\dot{U}\|_{\ell^2}^{2(1-\theta)}+|F(U)-F(\Phi)|^{(1-\theta)} +
		|\langle J(U),\dot{U}\rangle|^{1-\theta}\right).
	\end{eqnarray}
	Estimating the term $|\langle J(U),\dot{U}\rangle|^{1-\theta}$ of (\ref{eq049}), by the Cauchy - Schwartz inequality, we arrive at
	\begin{eqnarray}\label{eq050}
		[E(t)]^{1-\theta}\leq K_{\theta}\left(\|\dot{U}\|_{\ell^2}^{2(1-\theta)}+|F(U)-F(\Phi)|^{(1-\theta)} +
		\| J(U)\|_{\ell^2}^{1-\theta} \ \|\dot{U}\|_{\ell^2}^{1-\theta}\right).
	\end{eqnarray}
	The term $\| J(U)\|_{\ell^2}^{1-\theta} \ \|\dot{U}\|_{\ell^2}^{1-\theta}$ in (\ref{eq050}) is estimated by using Young's inequality $ab\leq a^p/p+b^q/q$, $1/p +1/q=1$, which holds for all $a,b>0$, as follows: set  $p=1/(1-\theta)>1$ and $q=1/\theta$, and obtain
	$$\begin{array}{llllllllllllllllllll}
	\| J(U)\|_{\ell^2}^{1-\theta} \ \|\dot{U}\|_{\ell^2}^{1-\theta} &\leq& 
	(1-\theta) \| J(U)\|_{\ell^2} + \theta \|\dot{U}\|_{\ell^2}^{(\frac{1-\theta}{\theta})}\\[2ex]
	&\leq&\| J(U)\|_{\ell^2}+\|\dot{U}\|_{\ell^2}^{(\frac{1-\theta}{\theta})}.
	\end{array}$$
	Consequently, inserting the above estimate into  (\ref{eq050}), it is found that $[E(t)]^{1-\theta}$ satisfies:
	\begin{eqnarray}\label{eq051}
		[E(t)]^{1-\theta}\leq K_{\theta}\left(\|\dot{U}\|_{\ell^2}^{2(1-\theta)}+|F(U)-F(\Phi)|^{(1-\theta)} +
		\| J(U)\|_{\ell^2}+\|\dot{U}\|_{\ell^2}^{(\frac{1-\theta}{\theta})}\right).
	\end{eqnarray}
	
	Now, by the definition of the $\omega$-limit set (\ref{limset}), and the fact that it contains all equilibria, there exists a sequence $t_n\rightarrow \infty$ as $n\rightarrow\infty$ such that $\lim_{n\rightarrow\infty}||U(t_n)-\Phi||_{\ell^2}=0$. In other words, there exists  $N\in\mathbb{N}^*$, such that, 
	\begin{eqnarray}
		\label{eq053}
		\|U(t_n)-\Phi\|_{\ell^2}<\frac{\tilde{\epsilon}}{2},\;\;\forall n\geq N,
	\end{eqnarray}
	for a sufficiently small $\tilde{\epsilon}>0$.  On the other hand, since $E(t)$ is decreasing --a consequence of Lemma \ref{L03}-- and $E(t)>0$, we have that
	$\lim_{t\to+\infty}E(t) = 0$.  This fact, together with Proposition \ref{L01}, allows us to consider among all the sequences $t_n\rightarrow\infty$,  the same sequence $t_n$ and $\tilde{\epsilon}$ as above, and derive that 
	\begin{eqnarray}\label{eq052}
		\|\dot{U}(t_n)\|_{\ell^2}\leq 1 \ \ \mbox{and} \ \ \ \frac{4\tilde{K}_\theta}{\theta\varepsilon}[E(t_n)]^{\theta}
		<\frac{\tilde{\epsilon}}{2},\;\;\forall n\geq N,
	\end{eqnarray}
	for some constant $\tilde{K}_\theta$ that will be defined below. 
	
	Now, we proceed by analyzing a contradiction argument, assuming that $U$ is not converging to $\Phi$, and that although (\ref{eq053}) and (\ref{eq052}) are satisfied, the solution $U$ may escape from the $\tilde{\epsilon}$-vicinity of $\Phi$. This escape could occur for $t>\bar{t}$, where  $\bar{t}$, and can be defined from  (\ref{eq053}), as
	\begin{eqnarray}
		\label{eq054}
		\bar{t}:= \sup\left\{t\geq t_N, \ \ \ \|U(s)-\Phi\|_{\ell^2}<\tilde{\epsilon} , \ \ \ s\in[t_N,t]\right\}.
	\end{eqnarray}
	Thus, if escape occurs then $\bar{t}$ is finite. With this assumption,
	we apply Lemma \ref{LINEQ} to deduce that  
	\begin{eqnarray}\label{eq055}
		\|J(U(t))\|_{\ell^2}\geq \nu_0|F(U(t))-F(\Phi)|^{1-\theta}, \ \ \forall t\in[t_N,\bar{t}).
	\end{eqnarray}
	Furthermore, the following inequalities are valid:
	\begin{eqnarray}\label{eq056}
		\|\dot{U}\|_{\ell^2}^{2(1-\theta)}\leq \|\dot{U}\|_{\ell^2} \ \ \mbox{and} \ \ \ 
		\|\dot{U}\|_{\ell^2}^{(\frac{1-\theta}{\theta})}\leq\|\dot{U}\|_{\ell^2}, \ \ \ \forall
		t\in[t_N,\bar{t}).
	\end{eqnarray}
	Consequently, using (\ref{eq055}) and (\ref{eq056}), the inequality (\ref{eq051}) becomes:
	\begin{eqnarray}
		\label{ALI}
		[E(t)]^{1-\theta}\leq \tilde{K}_\theta\left(\|\dot{U}\|_{\ell^2} + \| J(U)\|_{\ell^2}\right),
		\ \ \ \ \forall t\in[t_N,\bar{t}),
	\end{eqnarray}
	where $\tilde{K}_\theta=\max\{\nu_0,2\}K_{\theta}$.
	The estimate (\ref{ALI}) implies that
	\begin{eqnarray}\label{eq057}
		[E(t)]^{\theta-1} = \frac{1}{[E(t)]^{1-\theta}}\geq \frac{1}{\tilde{K}_\theta}
		\frac{1}{\left(\|\dot{U}\|_{\ell^2} + \| J(U)\|_{\ell^2}\right)},\ \ \ \ 
		\forall t\in[t_N,\bar{t}).
	\end{eqnarray}
	Differentiating the quantity $[E(t)]^{\theta}$ with respect to time, and using  (\ref{eq057}) and Lemma \ref{L03}, we observe that it satisfies the inequality 
	\begin{eqnarray}
		\label{eq058}
		-\frac{d}{dt}[E(t)]^{\theta}&=& -\theta \dot{E}(t)[E(t)]^{\theta-1}\nonumber\\[2ex]
		&\geq& \frac{\theta\varepsilon}{4}\left(\|\dot{U}\|_{\ell^2} + \| J(U)\|_{\ell^2}\right)^2 
		[E(t)]^{\theta-1}\nonumber\\[2ex]
		&\geq& \frac{\theta\varepsilon}{4}\left(\|\dot{U}\|_{\ell^2} + \| J(U)\|_{\ell^2}\right)^2 
		\frac{1}{\tilde{K}_\theta\left(\|\dot{U}\|_{\ell^2} + \| J(U)\|_{\ell^2}\right)},\;\;\forall t\in[t_N,\bar{t}).
	\end{eqnarray}
	Therefore we arrive to the differential inequality for $[E(t)]^{\theta}$:
	\begin{eqnarray}\label{eq059}
		-\frac{d}{dt}[E(t)]^{\theta}\geq\frac{\theta\varepsilon}{4 \tilde{K}_\theta}
		\left(\|\dot{U}\|_{\ell^2} + \| J(U)\|_{\ell^2}\right),\;\;\forall t\in[t_N,\bar{t}).
	\end{eqnarray}
	Integrating (\ref{eq059}) with respect to $t$ in $(t_N,\bar{t})$,  we get
	\begin{eqnarray} 
		\label{eq060}
		[E(t_N)]^{\theta}\geq\frac{\theta\varepsilon}{4 \tilde{K}_\theta}
		\int_{t_N}^{\bar{t}}\|\dot{U}(t)\|_{\ell^2} dt + \frac{\theta\varepsilon}{4 \tilde{K}_\theta}
		\int_{t_N}^{\bar{t}} \| J(U(t))\|_{\ell^2} dt +[E(\bar{t})]^{\theta}.
	\end{eqnarray}
	From the positivity of the  last two integral terms of  (\ref{eq060}), the inequality
	\begin{eqnarray}\label{eq061}
		\int_{t_N}^{\bar{t}}\|\dot{U}(t)\|_{\ell^2} dt\leq \frac{4 \tilde{K}_\theta}{\theta\varepsilon}
		[E(t_N)]^{\theta},
	\end{eqnarray}
	readily follows. Since 
	$$ U(\bar{t})-\Phi = U(t_N)-\Phi + \int_{t_N}^{\bar{t}}\frac{d}{dt}(U(t)-\Phi) dt,$$ we obtain the inequality
	\begin{eqnarray}
		\label{eq062A}
		\|U(\bar{t})-\Phi\|_{\ell^2} 
		\leq \|U(t_N)-\Phi\|_{\ell^2} + \int_{t_N}^{\bar{t}}\|\frac{d}{dt}(U(t)-\Phi)\|_{\ell^2} dt.
	\end{eqnarray}
	Furthermore, from (\ref{eq052}) we have that 
	$\frac{4 \tilde{K}_\theta}{\theta\varepsilon}
	[E(t)]^{\theta}<\frac{\tilde{\epsilon}}{2}$,
	not only for $t\in[t_N, \bar{t}]$, but for all $t\geq t_N$. Hence, (\ref{eq061}) implies also that   
	\begin{eqnarray}\label{eq062}
		\int_{t_N}^{\bar{t}}\|\dot{U}(t)\|_{\ell^2} dt<\frac{\tilde{\epsilon}}{2}.
	\end{eqnarray}
	Besides, $\|U(t_N)-\Phi\|_{\ell^2}<\frac{\tilde{\epsilon}}{2}.$ 
	Inserting the latter, as well as (\ref{eq062}), into (\ref{eq062A}) and summing, we obtain  
	\begin{eqnarray}
		\label{contra}
		\|U(\bar{t})-\Phi\|_{\ell^2}<\tilde{\epsilon},
	\end{eqnarray}
	The fact that the supremum $\bar{t}$ still satisfies (\ref{contra}), establishes that $\bar{t}$ can be continued arbitrarily, extending $\bar{t}\rightarrow\infty$.  Thus, (\ref{contra}) is valid for all $t\in [t_N,\infty)$, concluding the proof of (\ref{eq047}).
\end{proof}
%

\section{Bifurcations of nonlinear equilibria}
\setcounter{equation}{0}
In this section, we will study the existence of the second class of equilibria for the problem 
(\ref{eq0})-(\ref{eq02})-(\ref{eq03}) which, as noted previously, have the form 
$\bar{\Phi}=( \Phi, \mathbf{0})$, where $\Phi$ are solutions of the nonlinear algebraic system 
(\ref{eq063})-(\ref{eq063A}). 
%
%
%
%
%
%
The result that we will prove here
not only shows the existence of solutions of (\ref{eq063})-(\ref{eq063A}), but 
also justifies that these equilibrium solutions bifurcate from the 
eigenvalues of the linear discrete eigenvalue-problem: 
\begin{eqnarray}
\label{eq063L}
-k\Delta_d \varphi_{n}^j &=& E\varphi_n^{j},\\[2ex]
\label{eq063AL}
\varphi_{0}^j&=&\varphi_{K+1}^j=0.
\end{eqnarray}
We recall that the eigenvalues of (\ref{eq063L})-(\ref{eq063AL}) are \cite{JNLS2013},
\begin{eqnarray}
\label{EIGS}
	E_j (h)=\frac{4}{h^2}\sin^2\left(\frac{j\pi h}{2L}\right)=\frac{4(K+1)^2}{L^2}\sin^2\left(\frac{j\pi}{2(K+1)}\right),\;\;j=1,\ldots,K.
\end{eqnarray}
Thus, the principal eigenvalue $E_1$ is:
\begin{eqnarray}
\label{case1.1}
E_1(h)=\frac{4}{h^2}\sin^2\left(\frac{\pi h}{2L}\right)=\frac{4(K+1)^2}{L^2}\sin^2\left(\frac{\pi}{2(K+1)}\right).
\end{eqnarray}
It will be also useful 
to discuss the behavior of the eigenvalues in the various discreteness regimes. The {\em discrete} regime  corresponds to the case $h=O(1)$. In the {\em continuum} limit 
of $h\rightarrow 0$, and the {\em anti-continuum} limit of
$h\rightarrow\infty$ respectively,  we observe that
\begin{eqnarray}
\label{case1.2A}
&&\lim_{h\rightarrow 0}E_1(h) =\lambda_1=\frac{\pi^2}{L^2},\;\;\mbox{(continuum limit)},\\
\label{case1.3A}
&&\lim_{h\rightarrow \infty}E_1(h)=0,\;\; \mbox{(anti-continuum limit)}.
\end{eqnarray}
We  also recall 
the variational characterization of
$E_1>0$, 
\begin{eqnarray}
\label{eigchar}
E_1=\inf_{
	\begin{array}{c}
	X \in \ell^2_{K+2} \\
	X\neq 0
	\end{array}}\frac{(-k\Delta_d X,X)_{\ell^2}}{\sum_{n=0}^{K+1}|X_n|^2},
\end{eqnarray}
which implies the inequality 
\begin{eqnarray}
\label{crucequiv}
E_1\sum_{n=0}^{K+1}|X_n|^2\leq
k(-\Delta_dX,X)_{\ell^2}\leq 4k \sum_{n=0}^{K+1}|X_n|^2.
\end{eqnarray}
The bifurcation of nonlinear equilibria, will be a consequence of the Rabinowitz bifurcation Theorem \cite{RB71,Smo94,ZeiV1}, stated below.
\begin{theorem}
\label{th1}
Assume that $\cx$ is a Banach space with norm\ $||\cdot||_\cx$.  Consider
the map $\mathcal{F}(\mu,\cdot): \cx\rightarrow \cx$, $\mu\in\mathbb{R}$,
\begin{eqnarray}
\label{eq063B}
\mathcal{F}(\mu,\cdot)=\mu \mathcal{R}(\cdot) + \mathcal{W}(\mu,\cdot),
\end{eqnarray}
where $\mathcal{R}: \cx\rightarrow \cx$ is a compact linear map and
$\mathcal{W}(\mu,\cdot): \cx\rightarrow \cx$ is compact and satisfies
\begin{eqnarray}
\label{Order}
\lim_{||u||_\cx \to 0} \frac{||\mathcal{W}(\mu,u)||_{\cx}}{||u||_\cx}=0.
\end{eqnarray}
If $\frac{1}{\lambda^*}$ is a simple eigenvalue of $\mathcal{R}$, then the closure of the set
\begin{eqnarray*}
	C=\{ (\mu ,u) \in \mathbb{R} \times \cx : (\mu,u)\;\;
	\mbox{solves}\;\; u-\mathcal{F}(\mu,u)=0,\; u \not\equiv 0 \},
\end{eqnarray*}
possesses a maximal continuum (i.e. connected branch) of
solutions $C$ which branches out of $(\lambda^*,0)$ and $C$ either:

(i)\ meets infinity in\ $\mathbb{R} \times \cx$\ or,

(ii)\ meets\ $u=0$ in a point $(\hat{\mu},0)$ where $\hat{\mu}\neq\lambda^*$ and $\frac{1}{\hat{\mu}}$
is an eigenvalue of $\mathcal{R}$. 
\end{theorem}
 
To apply Theorem \ref{th1}, we need some preparations so as to rewrite Eq.~(\ref{eq063}) in the form:
 \begin{eqnarray}\label{eq066}
\Phi-\mathcal{F}(\mu, \Phi)=0, 
\end{eqnarray} 
requested by the theorem. The first step, is to define and discuss the properties of the linear operator $\mcR$: The linear operator $\mathcal{R}$, will be the inverse of the operator 
 \begin{eqnarray}\label{eq067}
 \ct(\Phi_n) = -k\Delta_d \Phi_n.
\end{eqnarray}
For instance, the above discussion on the eigenvalues $E_n$ of the linear eigenvalue problem (\ref{eq063A})-(\ref{eq063AL}), implies the following:  $\ct$ is self-adjoint on the Hilbert space $\cx:=\ell^2_{K+2}$, and positive; the latter follows from inequality (\ref{crucequiv}). Its eigenvalues are all simple, and are given by (\ref{EIGS}), which can be ordered as
$$ 0<E_1<E_2< \cdots <E_{K}.$$
Clearly, the operator $\ct$ is invertible, and we define $\mcR:=\ct^{-1}:\ell^2_{K+2}\rightarrow \ell^2_{K+2}$, its inverse. Now, we rewrite Eq.~(\ref{eq063}) in the operator-form:
 \begin{eqnarray}\label{eq070}
\ct(\Phi) - \omega_d^2 \Phi -\mathcal{G}(\omega_d^2, \Phi)=0, \ \ \ \Phi\in\ell^2_{K+2},
\end{eqnarray}
where $\mathcal{G}:\ell^2_{K+2} \to \ell^2_{K+2}$ is the non-linear operator 
$$ \mathcal{G}(\omega_d^2,\Phi_n) = -\omega_d^2\beta \Phi_n^3.$$
Next, we apply to Eq.~(\ref{eq070}), the operator $\mathcal{R}$,
and we get the equation 
 \begin{eqnarray}\label{eq071}
\Phi-\omega_d^2\mathcal{R}(\Phi) - \mathcal{R}\mathcal{G}(\omega_d^2,\Phi)=0.
\end{eqnarray}
Observe that Eq.~(\ref{eq071}) is in conformity with (\ref{eq063B})-(\ref{eq066}), with $ \mcR(\Phi)=\ct^{-1}(\Phi)$, and 
$\cw(\omega_d^2, \Phi)=\mathcal{R}\mathcal{G}(\omega_d^2,\Phi)$. We are ready to implement Theorem \ref{th1} to Eq. (\ref{eq071}) and prove 
\begin{proposition}
\label{RT1}	
There exists a maximal continuum of solutions $C_{E_j}$ of equation (\ref{eq063}), 
	$j=1,\dots,K+1$, bifurcating from $(E_j,0)$ and $C_{E_j}$ either (1) meets infinity in $\bbR\times\ell^2_{K+2}$
	or (2) meets $\Phi=0$ in a point $(\hat{\omega_d^2},0),$ where $\hat{\omega_d^2}\neq E_j$ and $\frac{1}{\hat{\omega_d^2}}$ 
	is an eigenvalue of $\mcR.$
\end{proposition}
\begin{proof} Since Eq. (\ref{eq071}) is considered on the finite-dimensional space $\ell^2_{K+2}$, the operators $\mathcal{R}$, and $\mathcal{W}=\mathcal{RG}$ are compact. Besides, the eigenvalues of $\mathcal{R}$ are $1/E_j$, $j=1,...,K+1$, and are all simple. Hence, it only remains to check the growth condition (\ref{Order}): 
we observe that
\begin{eqnarray*}
	\frac{\|\cw(\omega_d^2,\Phi)\|_{\ell^2}}{\|\Phi\|_{\ell^2}} 
	&=& \frac{\|\mathcal{R}\mathcal{G}(\omega_d^2, \Phi)\|_{\ell^2}}{\|\Phi\|_{\ell^2}}\\[2ex]
	&\leq& \frac{\|\mathcal{R}\|_{\ell^2}\ \|\mathcal{G}(\omega_d^2,\Phi)\|_{\ell^2}}{\|\Phi\|_{\ell^2}}
	\\[2ex]
	&\leq& \frac{\omega_d^2\beta\|\mathcal{R}\|_{\ell^2}\  \|\Phi\|_{\ell^2}^3}{\|\Phi\|_{\ell^2}}
	\\[2ex] 
	&=&\omega_d^2\beta  \|\mathcal{R}\|_{\ell^2}\ \|\Phi\|_{\ell^2}^2.
\end{eqnarray*}
For the above estimate, we used inequality (\ref{eq07}), with $q=3$ and $p=2$.
Letting $||\Phi||_{\ell^2}\rightarrow 0$, we see that
$ \lim_{\|\Phi\|_{\ell^2}\to0} \frac{\|\cw(\omega_d^2,\Phi)\|_{\ell^2}}{\|\Phi\|_{\ell^2}} 
=0$, concluding the proof of the proposition. 
\end{proof}
We proceed by discussing some geometric characteristics of the branches $\bfC_{E_j}$. For this purpose, it will be necessary to recall some further properties of the linear discrete eigenvalue problem (\ref{eq063L})-(\ref{eq063AL}). First, the Krein--Rutman theorem implies that the principal eigenfunction $\phi_n^1$ associated to the principal 
eigenvalue $E_1$ is positive, in the sense that $\phi_n^1>0$ for all $n=1,\dots K$,  except for the nodes $n=0$ and $n=K+1$, where it satisfies the boundary conditions (\ref{eq063AL}). On the other hand, the eigenvalue problem (\ref{eq063L})-(\ref{eq063AL}) is the discrete analogue of the 
Sturm--Liouville  eigenvalue problem:
\begin{eqnarray}\label{eq068}
-\psi''(x) &=& \lambda \psi(x), \ \ \ \ -L/2<x<L/2,\\[2ex]
\label{eq068A}
\psi(-L/2)&=&\psi(L/2)=0,
\end{eqnarray}
with the countable sequence of eigenvalues $\lambda_j=\frac{j^2\pi^2}{L^2}$, $j=1,2,\cdots$, and corresponding eigenfunctions 
$\psi_j(x)=\sin\left(\frac{j\pi x}{L}\right)$, having exactly $j-1$ nodal points.  It is also important to recall that the discrete eigenfunctions $\phi_n^j$, $j=1,2,\cdots,K+1$ of the discrete eigenvalue problem (\ref{eq063L})-(\ref{eq063AL}), trace a discretized analogue of their continuous counterparts, and consequently, they also have $j-1$ sign-changes. 

Motivated by the above properties of the discrete eigenfunctions of the problem (\ref{eq063L})-(\ref{eq063AL}),  we define the sets in $\ell^2_{K+2},$
\begin{eqnarray}\label{eq069}
S_j:= \left\{X\in\ell^2_{K+2} : X_0 = X_{K+1}=0 \ \ \ \mbox{with exactly}\ \ \ 
j-1\ \mbox{sign-changes}.
\right\}
\end{eqnarray}
This set is open in $\ell^2_{K+2}$, since for every $X\in S_j$, the open ball 
$$B(X,\varrho) = \left\{ X\in S_j : \|X-Y\|_{\ell^2_{K+2}}<\varrho,\;Y\in \ell^2_{K+2}\right\}$$
lying in $S_j$ by considering $\varrho$ sufficiently small. For instance, for such a 
$\varrho$, we get sufficiently small perturbations of the coordinates of $X$ in $S_j$. Thus, all the vectors of $\ell^2_{K+2}$ being in $B(X,\varrho)$ have the same number of sign-changes. 

Returning to the solutions $\Phi_n$ of the  nonlinear stationary problem (\ref{eq063})-(\ref{eq063A}), the local 
bifurcation theory and the implicit function theorem \cite{RB71,Smo94,ZeiV1}, guarantee that the branch $\bfC_{E_j}$ can be 
locally represented by the $C^1$- curve 

$$(\mu,\Phi):(-\gamma,\gamma)\to\bbR\times\ell^2_{K+2},$$
for some $\gamma$ sufficiently small. For instance, this representation has the following properties: 
\begin{eqnarray}
\label{eq072A}
\mu(0) &=& E_j , \ \ \ \ \chi(0)=0,\\[2ex]
\label{eq072}
(\mu(s),\Phi(s)) &=& (\mu(s),s(\varphi^j +\chi(s))), \ \ \ |s|<\gamma. 
\end{eqnarray}

Here, $\mu(s):=\omega_d^2(s)$ and $\|\chi(s)\|_{\ell^2_{K+2}}=O(|s|),$ in the neighborhood of the bifurcation point $(E_j,0)$. Furthermore, there is a neighborhood of $(E_j,0),$ such that any solution of (\ref{eq063})-(\ref{eq063A}) (or equivalently, of the nonlinear operator equation (\ref{eq071})), lies on this curve, or is exactly $(E_j,0)$. The next proposition, refers to a local concavity property of the branch $\bfC_{E_j}$.
\begin{proposition}
\label{RT2}	
 Consider the local representation (\ref{eq072A})-(\ref{eq072}) of branch $\bfC_{E_j}.$ Then, $\mu'(0)=0,$ $\mu''(0)>0,$ i.e., the local representation is concave-up. 
\end{proposition}
\begin{proof}
We insert the local representation of the branch $(\mu(s),\Phi(s)) = (\mu(s),s(\varphi^j +\chi(s)))$ in (\ref{eq063}), and we divide by $s>0$. Thus, we have that
$$-k\Delta_d (\varphi_n^j + \chi_n(s)) + \mu(s)\beta s^2(\varphi_n^j+\chi_n(s))^3 = \mu(s)(\varphi_n^j+\chi_n(s)).$$
Differentiating the above equation with respect to $s$, we get
\begin{eqnarray}
\label{eq073}
-k\Delta_d\chi_n'(s) &+& \mu(s)\left[2s\beta(\varphi_n^j+\chi_n(s))^3 +3\beta s^2(\varphi_n^j+\chi_n(s))^2 \chi_n'(s)\right] 
+ \mu'(s)[\beta s^2(\varphi_n^j+\chi_n(s))^3]\nonumber\\[2ex]  &=&\mu'(s)(\varphi_n^j+\chi_n(s)) + \mu(s)\chi_n'(s).
\end{eqnarray}
Setting $s=0$ to (\ref{eq073}), and using the relations $\mu(0)=E_j$, $\chi_n(0)=0$, we derive that
\begin{eqnarray}
\label{eq074}
-k\Delta_d\chi_n'(0)-\mu'(0)\varphi_n^j = E_j \chi_n'(0).
\end{eqnarray}
Multiplication of (\ref{eq074}) by $\varphi^j$, and summation, yields 
$$ -k\sum_{n=0}^{K+1}\Delta_d\chi_n'(0)\varphi_n^j - \sum_{n=0}^{K+1}\mu'(0) (\varphi_n^j)^2 = 
\sum_{n=0}^{K+1} E_j \chi_n'(0)\varphi_n^j,$$ 
while summation by parts in the first term  of the above equation, implies
\begin{eqnarray}
\label{eq075}
-k\sum_{n=0}^{K+1}\chi_n'(0)\Delta_d \varphi_n^j- \sum_{n=0}^{K+1}\mu'(0) (\varphi_n^j)^2 = 
\sum_{n=0}^{K+1} E_j \chi_n'(0)\varphi_n^j.
\end{eqnarray}
Since $E_j$ and $\varphi_n^j$ solve the linear discrete eigenvalue problem (\ref{eq063L})-(\ref{eq063AL}),
we have that
$$ -k\sum_{n=0}^{K+1}\chi_n'(0)\Delta_d \varphi_n^j = \sum_{n=0}^{K+1} E_j \chi_n'(0)\varphi_n^j.$$
Therefore, (\ref{eq075}) results in
$$\sum_{n=0}^{K+1}\mu'(0)|\varphi_n^j|^2 = 0.$$ 
The latter implies that $\mu'(0)=0.$ To evaluate $\mu''(s)$, we differentiate (\ref{eq073}) with respect to $s$: 
\begin{eqnarray}
\label{eq076}
-k\Delta_d\chi_n''(s) &+&  \mu(s)[2\beta(\varphi_n^j+\chi_n(s))^3 + 6\beta s(\varphi_n^j+\chi_n(s))^2 \chi_n'(s)]\nonumber\\[2ex]
&+&2\mu'(s)\beta s(\varphi_n^j+\chi_n(s))^3\nonumber\\[2ex] 
&+&\mu(s)[6\beta s(\varphi_n^j+\chi_n(s))^2\chi_n'(s)+6\beta s^2(\varphi_n^j+\chi_n(s))\chi_n'(s)^2+3\beta s^2(\varphi_n^j+\chi_n(s))^2\chi_n''(s)]\nonumber\\[2ex]
&+& 3\beta s^2\mu'(s)(\varphi_n^j+\chi_n(s))^2\chi_n'(s)\nonumber\\[2ex]
&+&\mu'(s)[2\beta s(\varphi_n^j+\chi_n(s))^3 + 3\beta s^2(\varphi_n^j+\chi_n(s))^2\chi_n'(s)] \nonumber\\[2ex]
&+& \mu''(s)(\beta s^2(\varphi_n^j+\chi_n(s))^3) \nonumber\\[2ex]
&=& \mu''(s)(\varphi_n^j+\chi_n(s))+\mu'(s)\chi_n'(s) + 
\mu'(s)\chi_n'(s)+\mu(s)\chi_n''(s).
\end{eqnarray}
We now set $s=0$ in (\ref{eq076}), use $\mu'(0)=0$ (as proved above) and the relations 
$\mu(0)=E_j$, $\chi_n(0)=0$, and derive the equation:
\begin{eqnarray}
\label{eq077}
-k\Delta_d\chi_n''(0) + 2\beta E_j (\varphi_n^j)^3 = \mu''(0)\varphi_n^j + E_j \chi_n''(0).
\end{eqnarray}
Handling of (\ref{eq077}) similarly to Eq. (\ref{eq075}), yields:
$$ -k\sum_{n=0}^{K+1}\chi_n''(0)\Delta_d \varphi_n^j + 2\beta \sum_{n=0}^{K+1}E_j (\varphi_n^j)^4 = 
\sum_{n=0}^{K+1}\mu''(0)(\varphi_n^j)^2 + \sum_{n=0}^{K+1}E_j \chi_n''(0)\varphi_n^j.$$ 

Then, using again that $(E_j, \phi_n^j)$ are eigensolutions of (\ref{eq063L})-(\ref{eq063AL}), we eventually derive the equation:
\begin{equation} \mu''(0)\sum_{n=0}^{K+1}(\varphi_n^j)^2 = 2\beta E_j\sum_{n=0}^{K+1} (\varphi_n^j)^4.\label{m2}\end{equation} 
The above equation clearly implies that $\mu''(0)>0.$
\end{proof}

Using Eq. (\ref{m2}) we can acquire an estimation about the value of $\mu''(0)$. 
Equation~\ref{m2} can be written as
\begin{equation}\mu''(0) = 2\beta E_j\frac{\sum_{n=0}^{K+1} (\varphi_n^j)^4}{\sum_{n=0}^{K+1}(\varphi_n^j)^2}.\label{m22}\end{equation}
From the local representation (\ref{eq072}) is implied that $\varphi_n^j$ is of order one, and thus  
we can approximate the fraction in the above formula as
\begin{equation}\displaystyle\frac{\sum_{n=0}^{K+1} (\varphi_n^j)^4}{\sum_{n=0}^{K+1}(\varphi_n^j)^2}=\frac{\int_0^\pi \alpha^4\sin^4(jx)dx}{\int_0^\pi \alpha^2\sin^2(jx)dx}=\frac{3\alpha^2}{4},\label{sums}\end{equation}
where $\alpha$ is the normalization amplitude for the discrete eigenfunctions $\varphi_n^j$, i.e., 
\begin{eqnarray}
\label{defa}
\alpha=\frac{1}{||\phi^j||_{\ell^2}}.
\end{eqnarray}	
 
Thus, Eq.~(\ref{m22}) becomes:
\begin{equation}
\mu''(0) = \frac{3}{2}\alpha^2\beta E_j.
\label{conc}
\end{equation}
It is observed that the concavity of the local representation of each $\bfC_{E_j}$ is analogous of 
the value of the corresponding $E_j$ \footnote{Note that the numerical calculation of the left part 
of (\ref{sums}) with the summations provides the same result as the middle part with the integrals.}. 
%
%

We conclude, with the proof that the nonlinear equilibrium branches $\bfC_{E_j}$ satisfy the scenario (i) of Theorem \ref{th1}.
\begin{proposition}
\label{RT3}
For any $k>0$, the maximal continuum of solutions $\bfC_{E_j}$ of Eq.~(\ref{eq063}) 
bifurcating from $(E_j,0)$ meets infinity in $\bbR$.  
\end{proposition}
\begin{proof}
(a)  Recall that any solution $(\mu,\Phi)$ close to $(E_j,0)$ has the same number of sign-changes as the eigenstate $\Phi^j$ corresponding to the eigenvalue $E_j$. This is due to the local $C^1$ - representation of the solution $\Phi$, as $\Phi_n(s) = s\Phi^j_n + s\chi_n(s)$. For instance, each linear state $\Phi^j$ belongs to the set $S_n$ defined in eq~(\ref{eq069}) and $\|\chi\|_{\ell^2_{K+2}}=O(|s|).$ It follows then, that the solution $\Phi$ satisfies the estimate $\|\Phi(s)\|_{\ell^2_{K+2}}\leq |s|\; \|\Phi^j\|_{\ell^2_{K+2}}+O(s^2)$ in the neighborhood of the bifurcation point $(E_j,0)$. Therefore, since the set $S_j$ is open, we get from the above estimate, that $\Phi\in S_j$ for $|s|<\gamma$. 
(b) Now, for all $(\mu,\Phi)\in\bfC_{E_j}$ and each $j=1,2,....K$, we consider the indicator function: 
$$f(\mu,\Phi) =\left\{ 
\begin{array}{lllllllllllllllllll} &1,& \mbox{if} \; \; \;\Phi\in S_j \\[2ex]
&0,& \mbox{if}\; \; \; \Phi=0, \; \; \mu=E_j, \; \; m\neq j,  
\end{array}\right. $$
that is $f(\mu,\Phi)=0$ if the branch $\bfC_{E_j}$, meets the axis $(\mu,0)$ in another eigenvalue $E_m\neq E_j$. Note that $f$ is well defined due to the two possibilities described by the previous proposition. From (a),  we have that if $(\mu,\Phi)$ is in a small neighborhood of $(E_j,0)$, then $f(\mu,\Phi)=1$. Thus, the function $f$ is constant in the small neighborhood of $(E_j,0)$ and cannot change value in this small neighborhood, i.e $f$ is locally constant. The set $S_j$ is open and the function $f$ is locally constant on the connected set $\bfC_{E_j}$. Both facts clearly imply that $f$ is continuous. Therefore, $f(\bfC_{E_j})$ should be also connected, since the image of a connected set through a continuous function should be connected. However, $f$ is integer valued, and the fact that $f(\bfC_{E_j})$ is connected, implies that $f$ should be constant, e.g.,  $f=1$, for all $(\mu,\Phi)\in\bfC_{E_j}$. Therefore, $\bfC_{E_j}$ cannot contain a point $(E_m,0)$ with $E_m\neq E_j$, and $\bfC_{E_j}$ should be unbounded.
\end{proof}

\begin{figure}[tbh!]
	\centering
	\includegraphics[scale=0.06]{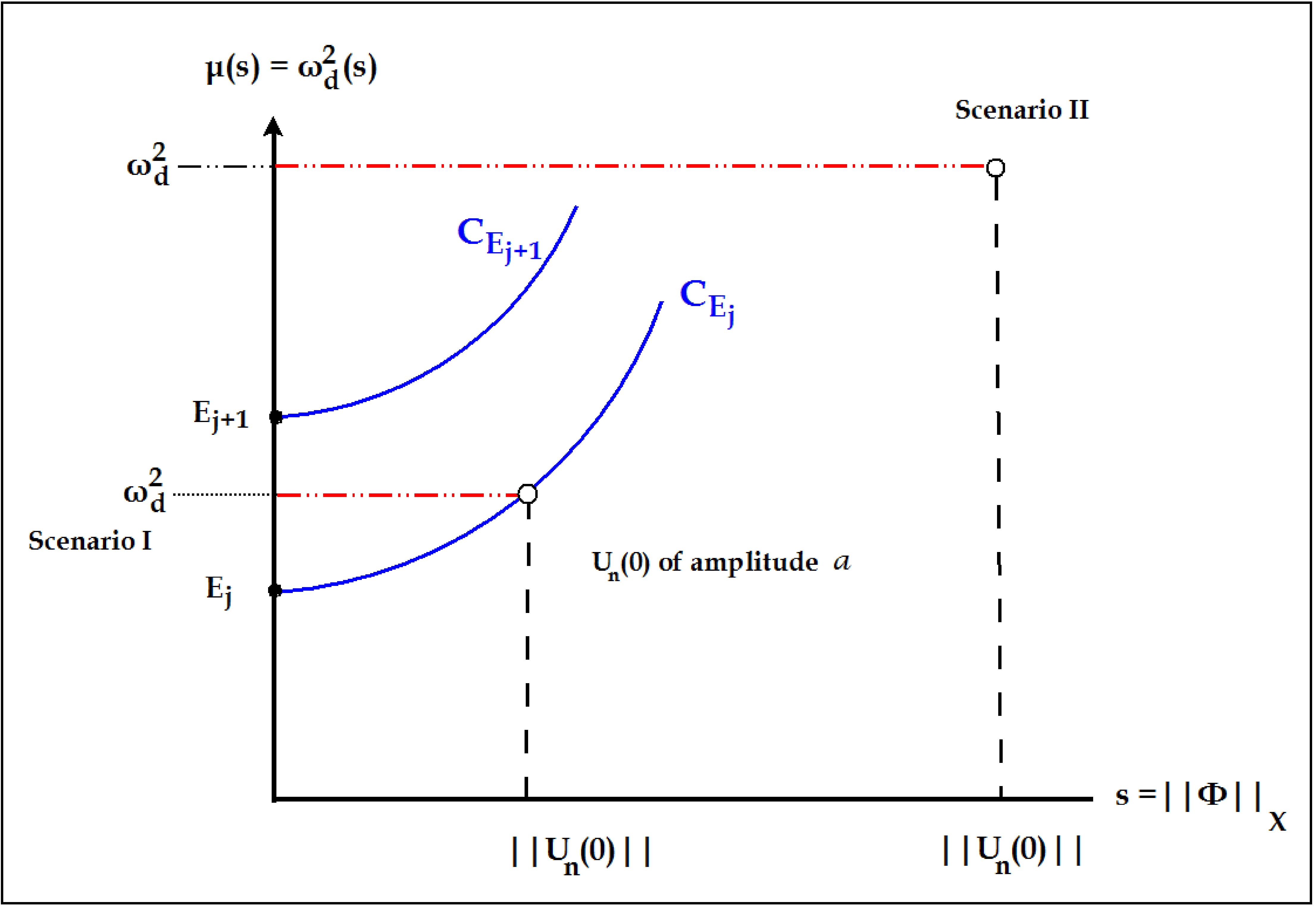} \vspace{0cm}
	\caption{(Color Online) Visualization of the scenarios for the selection of spatially extended initial conditions (\ref{ic1}), based on the local stationary bifurcation diagram (cartoon). Locally concave-up (in the sense of Proposition \ref{RT2}) branches $\mathbf{C}_{E_{j}}$ of nonlinear equilibria, are bifurcating from the
          (linear problem) eigenvalues $E_j$ [continuous (blue) curves]. Scenario I: The initial condition (\ref{ic1}) of amplitude $a$ and norm $||U_n(0)||$, has the same number of sign-changes, as the nonlinear steady state of the branch $\mathbf{C}_{E_{j}}$. The point $(||U_n(0)||,\omega_d^2 )$ [dot ``on'' the branch $\mathbf{C}_{E_{j}}$], corresponds to an extended initial condition, being close in norm, to an equilibrium of the branch $\mathbf{C}_{E_{j}}$. Scenario II: the initial condition (\ref{ic1}) is as in Scenario I having a large norm, but $\omega_d^2$ is far from the bifurcation value $E_j$. The point $(||U_n(0)||,\omega_d^2 )$ is far from the local bifurcation diagram. 
%
}
	\label{fig2cart}
\end{figure}


The local representation of the equilibrium branches $\bfC_{E_j}$ is schematically visualized in the cartoon of Fig.~\ref{fig2cart}. A  numerical justification of this schematic representation will be one of our aims in the numerical study that follows.

\section{Numerical Study}
\setcounter{equation}{0}
In this section, we present numerical results concerning the dynamics of the DKG chain (\ref{eq0})-(\ref{eq02})-(\ref{eq03}), with two principal aims. The first, 
concerns the numerical study of a potential ``dynamical stability'' property of the  equilibrium branches, identified by the bifurcation results of Propositions \ref{RT1}-\ref{RT2}. The second,  concerns the study of the structure of the equilibrium branches with respect to discreteness, and the strength of nonlinearity. 

\subsection{Existence and linear stability of bifurcating branches}
We start the presentation of our numerical results, by depicting a numerical justification of the global bifurcation analysis discussed in the previous section.  Figure~\ref{fig2} shows the numerically computed equilibrium branches  $\bfC_{E_j}$ ($j=1,\ldots,6$), for $K=59$, $L=60$ and $\beta=1$. Let us define the branch $\bfC_{E_0}$ [$\omega_d^2$ --vertical (red)  axis] corresponding to the trivial equilibrium ${\Phi}_0=\mathbf{0}$.  The numerical results are in full agreement with the analytical predictions, on the local geometric structure of the branches (which were schematically summarized in Fig. \ref{fig2cart}): the numerically computed branches bifurcate from the corresponding eigenvalues $\omega_d^2=E_j$ of the linear discrete eigenvalue problem  (\ref{eq063L})-(\ref{eq063AL}), and are locally concave up, as it was analytically shown in Propositions \ref{RT1}-\ref{RT2}. Furthermore, it is observed that the concavity of the various branches increases as the value of $\omega_d^2$ increases, in accordance to Eq.~(\ref{conc}), and they are indeed  unbounded, as it was proved in Proposition \ref{RT3}. Finally, in the right panel of Fig.~\ref{fig2}, we compare the numerically computed branches (depicted with solid lines) against the corresponding ones to the analytical estimation derived from the Taylor expansion up to the second order, i.e.,
\begin{eqnarray*}
\mu(s)=\mu(0)+\mu'(0)s+\mu''(0)\frac{s^2}{2}+{\cal O}(s^3).	
\end{eqnarray*}	
The coefficients of the expansion are calculated by our previous results, as follows: $\mu(0)=E_j$, $\mu'(0)=0$, $\mu''(0)$ is given by using Eq. (\ref{conc}), and the parameter $\alpha$ defined in (\ref{defa}), has the value $\alpha=0.182$ for the above set of lattice parameters. We observe a very good agreement between the two lines at least for small values of the norm where the higher order corrections are negligible. 
\begin{figure}[tbp!]
	\centering
	\includegraphics[height=7cm]{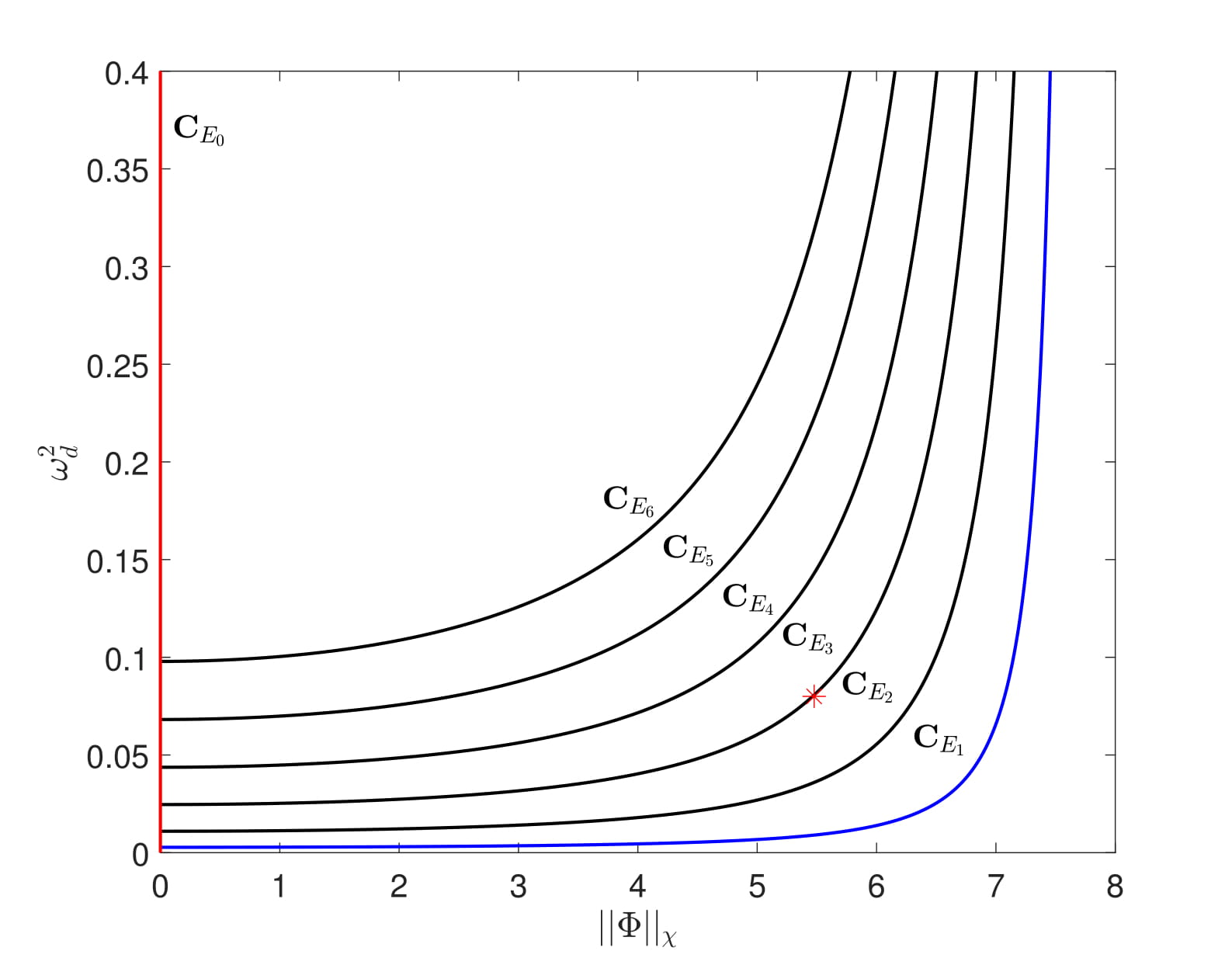} \hspace{0cm}\includegraphics[height=7cm]{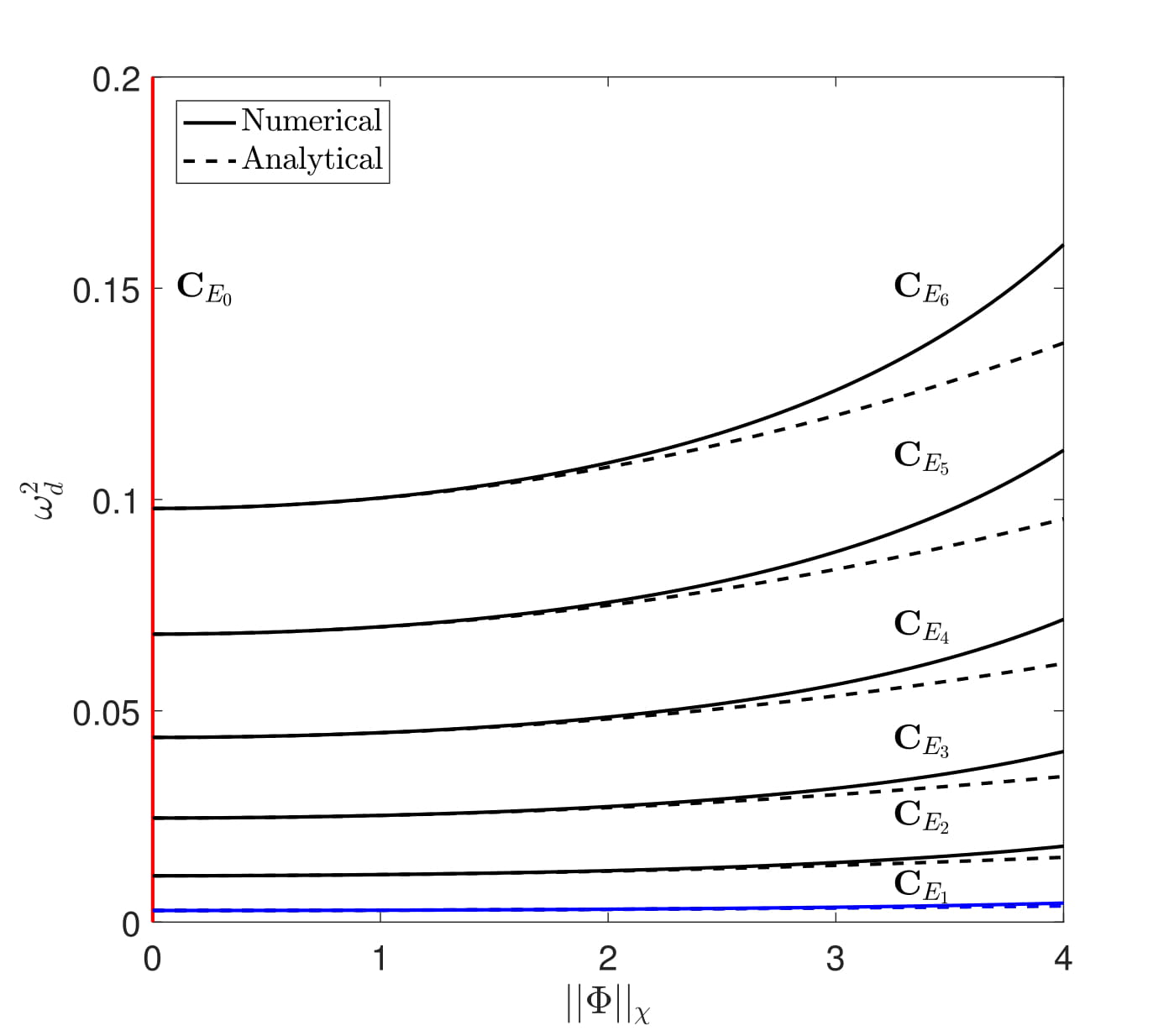}
	\caption{Numerically computed equilibrium branches  $\mathbf{C}_{E_{j}}$ which are bifurcating from the values of the parameter $\omega_d^2$, corresponding to the eigenvalues of the discrete eigenvalue problem (\ref{eq063L})-(\ref{eq063AL}). The figures show the equilibrium branches bifurcating from  $\omega_d^2=E_j$ for $j=1,\ldots,6$, where $E_1=0.0027, E_2=0.011, E_3=0.025$, $E_4=0.044, E_5=0.068$ and $E_6=0.098$. Parameters: $K=59$, $L=60$, $\beta=1$. In the left panel we show the branches for the full range of the values of the norm for equilibrium states. The  marked (by a star) point  $(||U_n(0)||, \omega_d^2)=(5.48, 0.08)$ stands for an initial condition  close in norm (with accuracy of order $10^{-2}$), to an equilibrium solution $\Phi_3$ of $\mathbf{C}_{E_{3}}$). In the right panel we compare the numerically calculated branches (solid line) with the analytically predicted by Eq.~(\ref{conc}) ones (dashed line) for smaller values of the norm.}
	\label{fig2}
\end{figure}

Next, we consider the linear stability of the equilibrium branches $\bfC_{E_j}$.
To each equilibrium branch $\bfC_{E_j}$ bifurcating when $\omega_d^2$ is crossing the linear mode  $E_j$, a number of $j-1$ unstable eigenvalues emerge, and consequently, the corresponding equilibria $\Phi_j$ possess a $j-1$-dimensional unstable manifold.  Such an instability structure is similar to that described by \cite[Theorem 24.14, pg. 538]{Smo94} for continuous gradient systems, and insight can be given by discussing the linearization spectrum of the branch  $\bfC_{E_0}$: Figure~\ref{fig2b}, depicts the birth of unstable eigenvalues of $\bfC_{E_0}$ as the bifurcation parameter $\omega_d^2$ is increasing, as well as the growth of their positive real part. In the inset of Fig.~\ref{fig2b}, we observe that for $0<\omega_d^2<E_1 \approx 0.003$, there exist no unstable eigenvalues of $\bfC_{E_0}$. 
\begin{figure}[tbp!]
	\centering
	\includegraphics[scale=0.22]{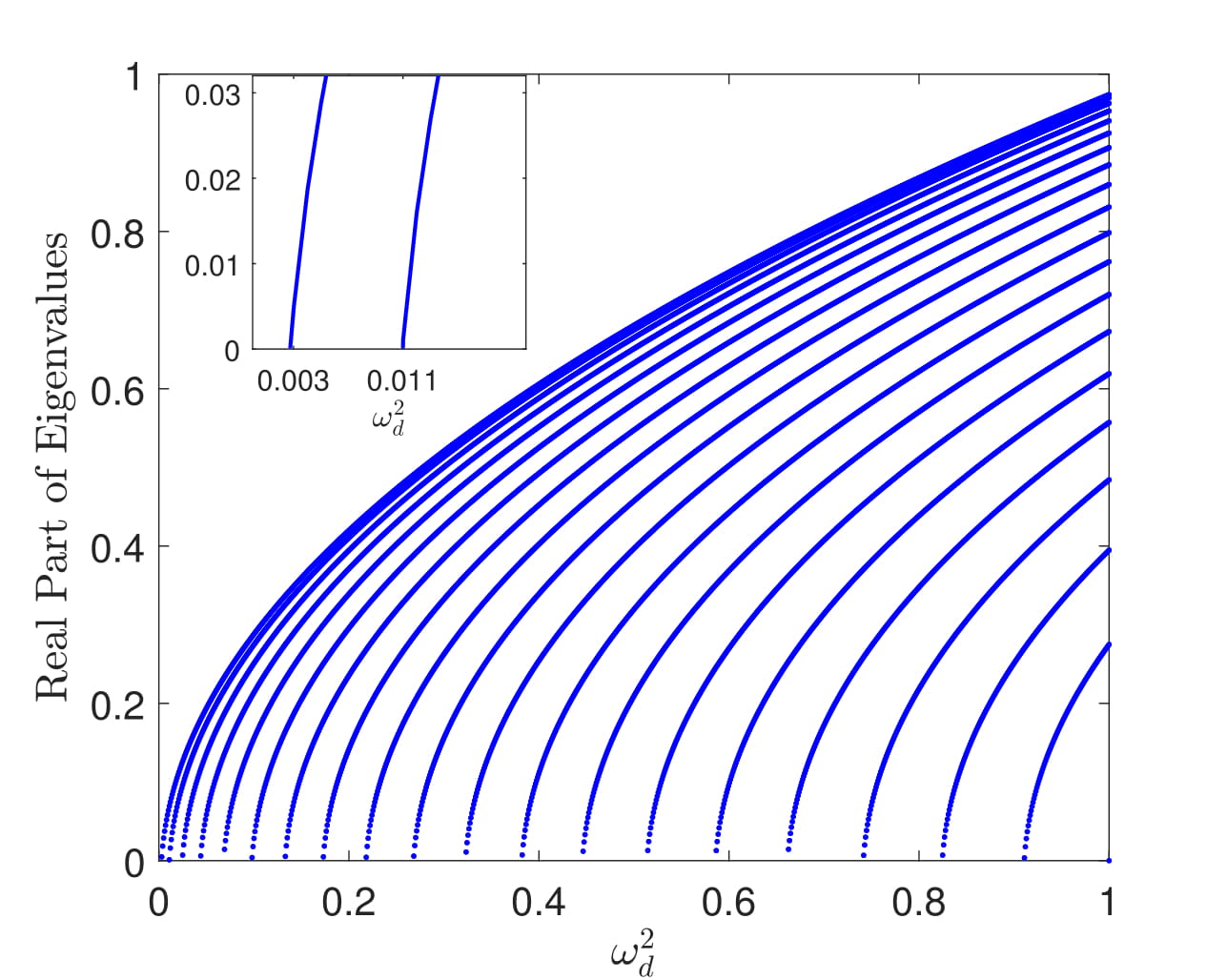} \hspace{0cm}
	\caption{(Color Online) Numerical justification for the emergence of unstable eigenvalues of the branch $\bfC_{E_0}$, as the bifurcation parameter is crossing the linear modes $E_j$. The growth of their positive real part is also portrayed. Same parameters as in the previous figure.}
	\label{fig2b}
\end{figure}
The first unstable eigenvalue of $\bfC_{E_0}$ is numerically detected when $\omega_d^2=E_1$, i.e., when the first branch $\bfC_{E_1}$  bifurcates. The branch $\mathbf{C}_{E_1}$ is stable,  while $\bfC_{E_0}$ acquires one unstable eigenvalue as a result of the bifurcation. The second  unstable eigenvalue of $\bfC_{E_0}$ is detected when $\omega_d^2=E_2 \approx 0.01$, which is the value 
of $\omega_d^2$ when $\mathbf{C}_{E_2}$ bifurcates as well.
This new branch $\mathbf{C}_{E_2}$ inherits one positive eigenvalue (from $\bfC_{E_0}$), giving rise to its one-dimensional unstable manifold, while $\bfC_{E_0}$  has now acquired two unstable eigenvalues. Figure~\ref{fig2b}, verifies numerically that, inductively, as soon as $\omega_d^2=E_j$, the $\bfC_{E_0}$ branch acquires its $j$-th unstable eigenvalue. Consequently, the associated bifurcating branch $\bfC_{E_j}$ inherits $j-1$ unstable eigenvalues, giving rise to its $j-1$-dimensional unstable manifold.
We note in passing that
the results of Sturm-Liouville theory
(for discrete, self-adjoint operators) via the Sturm comparison theorem
can be used to lead to the same conclusions for the number of unstable
eigenvalues of each $\bfC_{E_j}$ branch.
\subsection{Evolution of spatially extended initial conditions}

In light of the instability properties of the equilibrium branches discussed above, it is interesting to numerically investigate the 
local attractivity properties of the relevant equilibria. We   consider a set of extended initial conditions of the form
\begin{eqnarray}
\label{ic1}
U_n(0)=U_{n,0}=a\,\sin\left(\frac{j\pi hn}{L}\right),\;\;j=1,...,K
\end{eqnarray}
where $h=\frac{L}{K+1}$ is the lattice spacing and $a>0$ is the amplitude of (\ref{ic1}); we also assume zero initial velocity,  
\begin{eqnarray}
\label{ic4}
\dot{U}_n(0)=U_{n,1}=\mathbf{0}.
\end{eqnarray}
For this type of initial data we introduce the notation 
	\begin{eqnarray}
	\label{eqnota}
	U_n(0)\simeq\mathbf{C}_{E_{j}},
	\end{eqnarray} 
to describe the following property of (\ref{ic1}): 
$U_n(0)$ has the same sign-changes as a nonlinear steady state $\Phi_{j}\in\mathbf{C}_{E_{j}}$ of the DKG chain (\ref{eq0})-(\ref{eq02})-(\ref{eq03}) (solution of the nonlinear stationary problem (\ref{eq063})-(\ref{eq063A})). Let us recall that a nonlinear steady-state $\Phi_{j}\in\mathbf{C}_{E_{j}}$ has $j-1$ sign-changes. We call such an initial condition, \emph{similar} to a branch $\mathbf{C}_{E_{j}}$. For such initial conditions we consider two different scenarios described as follows. 
\begin{itemize}
	\item Scenario I: We study the dynamics when the initial condition $U_n(0)$ is such that $U_n(0)\simeq\mathbf{C}_{E_{j}}$ and the bifurcation parameter $\omega_d^2$ is selected so that the point $(||U_n(0)||, \omega_d^2)$ belongs to the bifurcation diagram (see Fig.~\ref{fig2cart}). In this situation, the norm of the initial condition is close to that of a solution of the branch $\mathbf{C}_{E_{j}}$ up to numerical accuracy of $10^{-2}$. 
\end{itemize}


\begin{figure}[tbh!]
	\centering
\includegraphics[scale=0.11]{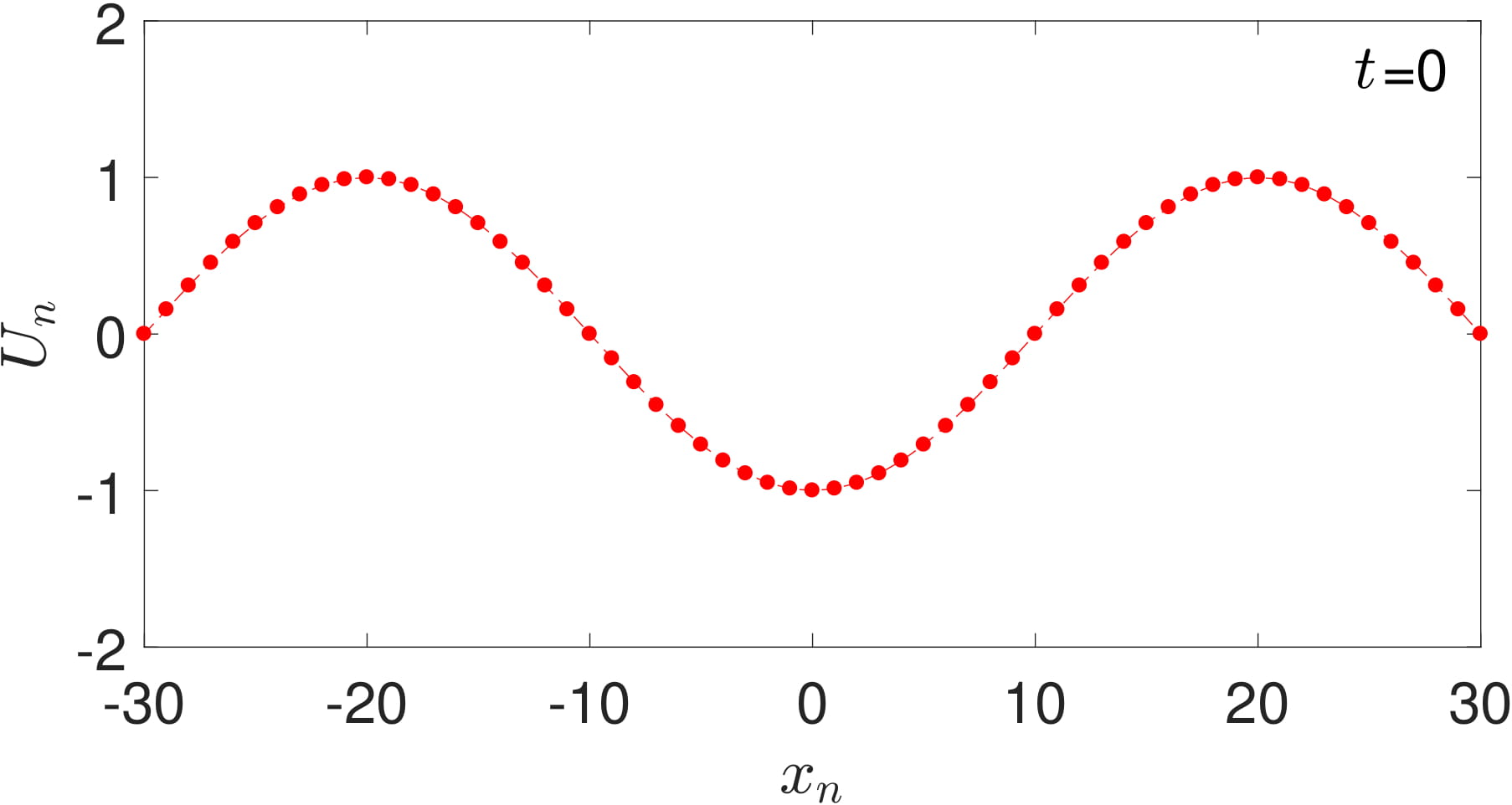}
\quad
\includegraphics[scale=0.11]{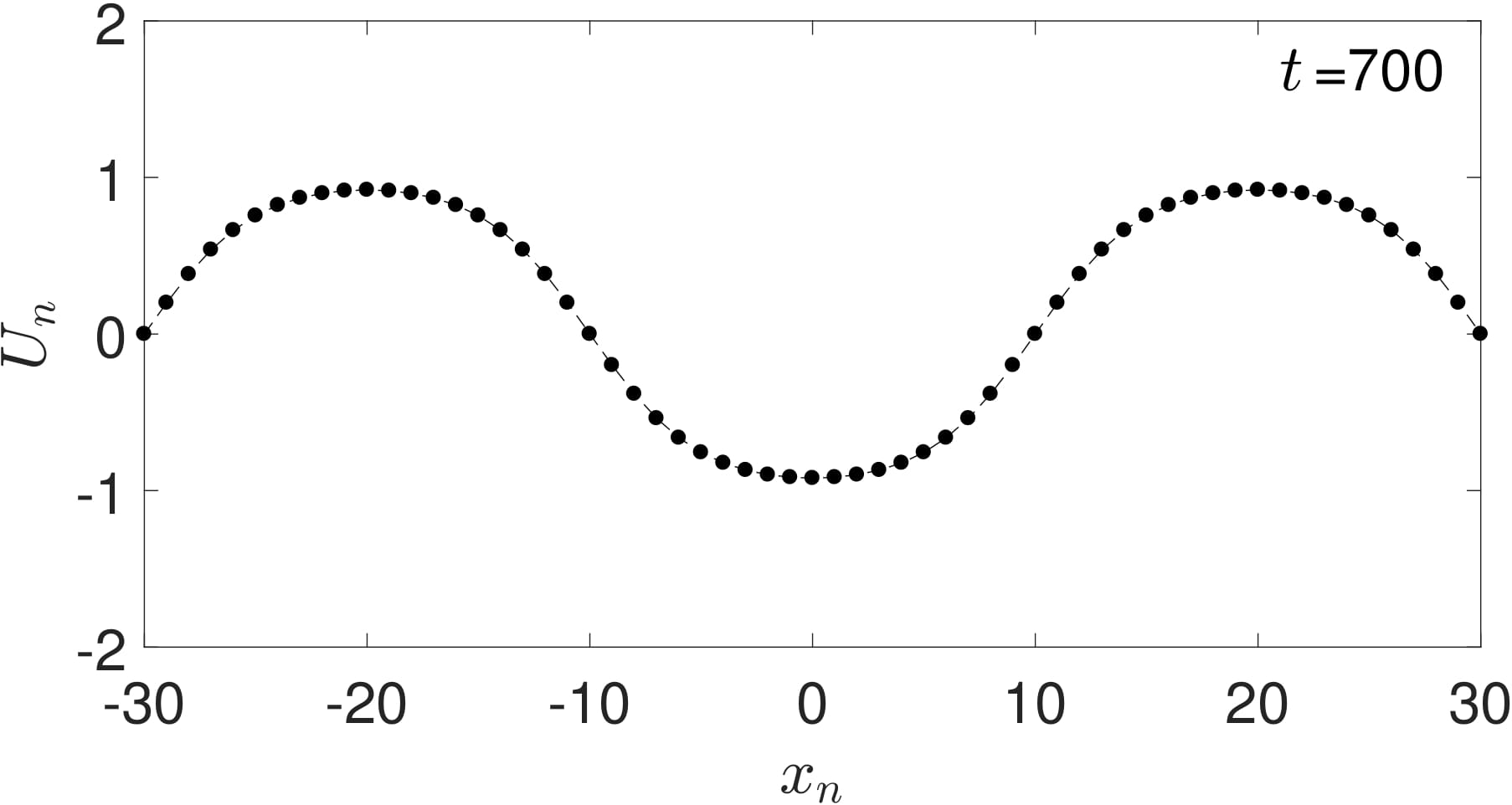}
\\[4ex]
\includegraphics[scale=0.15]{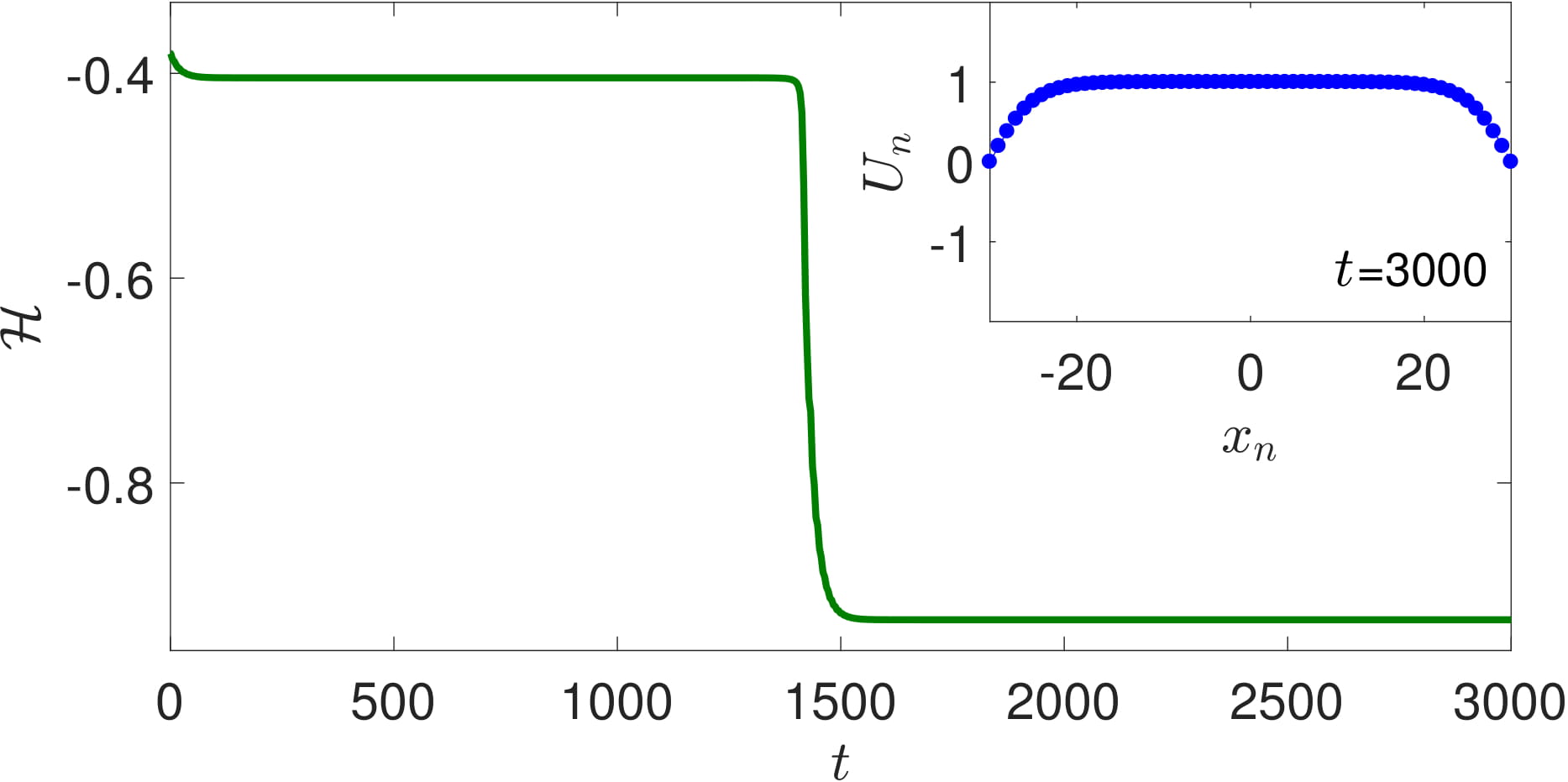} \vspace{0cm}
	\caption{(Color Online) Metastable dynamics and convergence to the  state of the branch $\bfC_{E_1}$ (2nd row), for an initial condition (\ref{ic1}), $U_n(0)\simeq\mathbf{C}_{E_{3}}$, and amplitude $a=1$ (1st row, left panel). Other parameters:  $K=59$, $L=60$, $\omega_d^2=0.08$, $\beta=1$, $\delta=0.05$. Metastability is observed in the evolution of the Hamiltonian energy [continuous (green) curve].}
	\label{Fig4}
\end{figure}
%
To be more specific, such a study wishes to take advantage of the local structure of the branches (as shown in Fig.~\ref{fig2}), and of the associated linear stability analysis, in predicting the state of convergence for points $(||U_n(0)||, \omega_d^2)$ sufficiently close to an equilibrium $\Phi_j$ and within the local bifurcation diagram (the latter is relevant for moderate values of the initial norm $||U_n(0)||$). 

For this scenario, we start with the initial condition  $U_n(0)\simeq\mathbf{C}_{E_{3}}$ of amplitude $a=1$ and bifurcation parameter $\omega_d^2=0.08$. This choice defines a point $(||U_n(0)||,\omega_d^2)=(5.48,0.08)$ (visualized by a star) in the bifurcation diagram of Fig.~\ref{fig2}, which is close 
to an equilibrium solution of the branch $\mathbf{C}_{E_{3}}$. The rest of the parameter values are 
$K=59$, $L=60$, $\beta=1$, and $\delta=0.05$. 
The dynamics of this state is shown in Fig.~\ref{Fig4}.
The top left panel of this figure shows the profile of the above initial  condition at $t=0$. The system was integrated up to $t=3000$, and the resulting dynamics is summarized in the bottom panel, portraying 
the Hamiltonian energy as a function of $t$ [continuous (green curve)], and the ultimate equilibrium of convergence (inset).  The Hamiltonian energy used herein as a diagnostic (attaining a constant value when convergence to an equilibrium is reached), captures the metastable dynamics, as verified numerically by the two different plateaus attained in its graph. In this case, metastability involves an equilibrium $\Phi_3\in \bfC_{E_3}$ (located at energy level $\mathcal{H}
\approx -0.4$), and the ground state $\Phi_1\in  \bfC_{E_1}$ (located at the energy level $\mathcal{H}=-0.93$). Such metastable dynamics indicate that the point $(||U_n(0)||,\omega_d^2)=(5.48,0.08)$ defines an orbit traced close to the stable manifold of the equilibrium $\Phi_3\in \mathbf{C}_{E_{3}}$ for $0<t\lesssim 1500$, but for sufficiently long times
$t\gtrsim 1500$,  the unstable manifold of the state eventually
takes over leading eventually to convergence to the ground state $\Phi_1\in  \bfC_{E_1}$. The dynamics is in accordance with the linear stability analysis of the branches: recall that solutions of $\mathbf{C}_{E_{3}}$ are unstable (possesing a 2D-unstable manifold), while solutions of $\mathbf{C}_{E_{1}}$ are linearly stable.
Note that although (both in the present and in the following investigations) we use a final integration time of $t=3000$, once the system reaches a stable configuration it is not expected to depart from it. 
%

		\begin{figure}[tbh!]
			\centering

\includegraphics[scale=0.125]{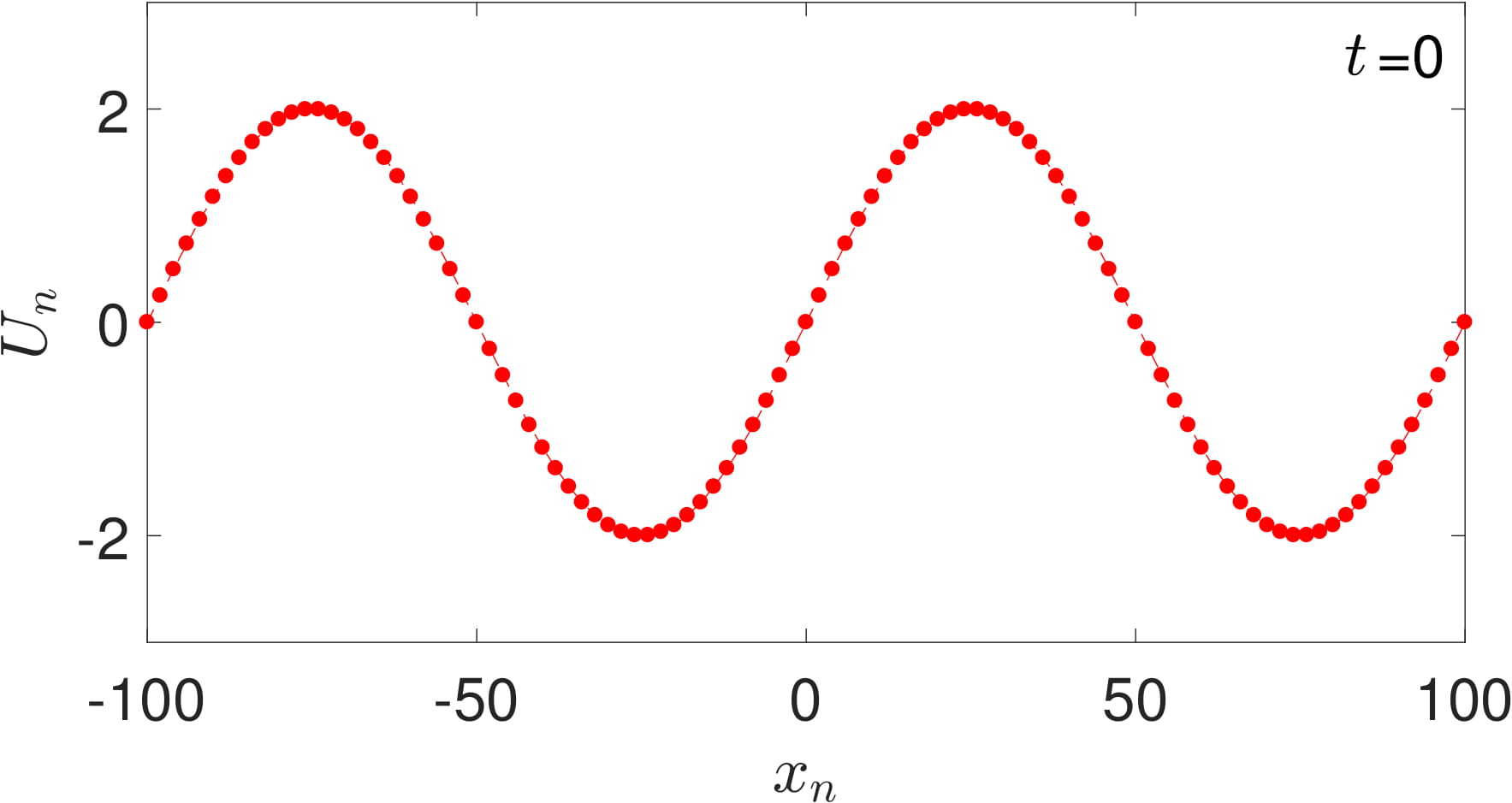}
\quad
\includegraphics[scale=0.125]{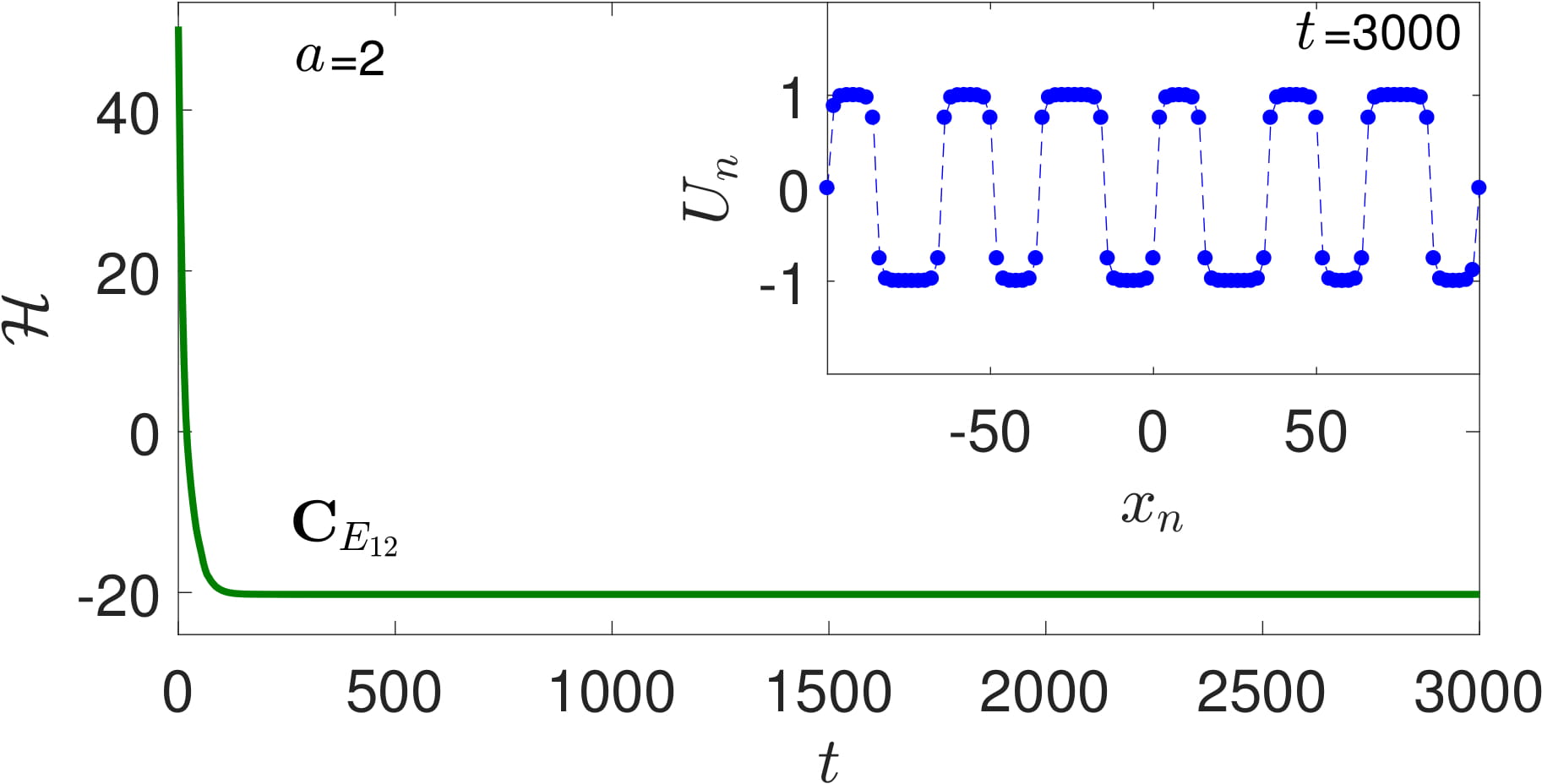}
\\[4ex]
\includegraphics[scale=0.125]{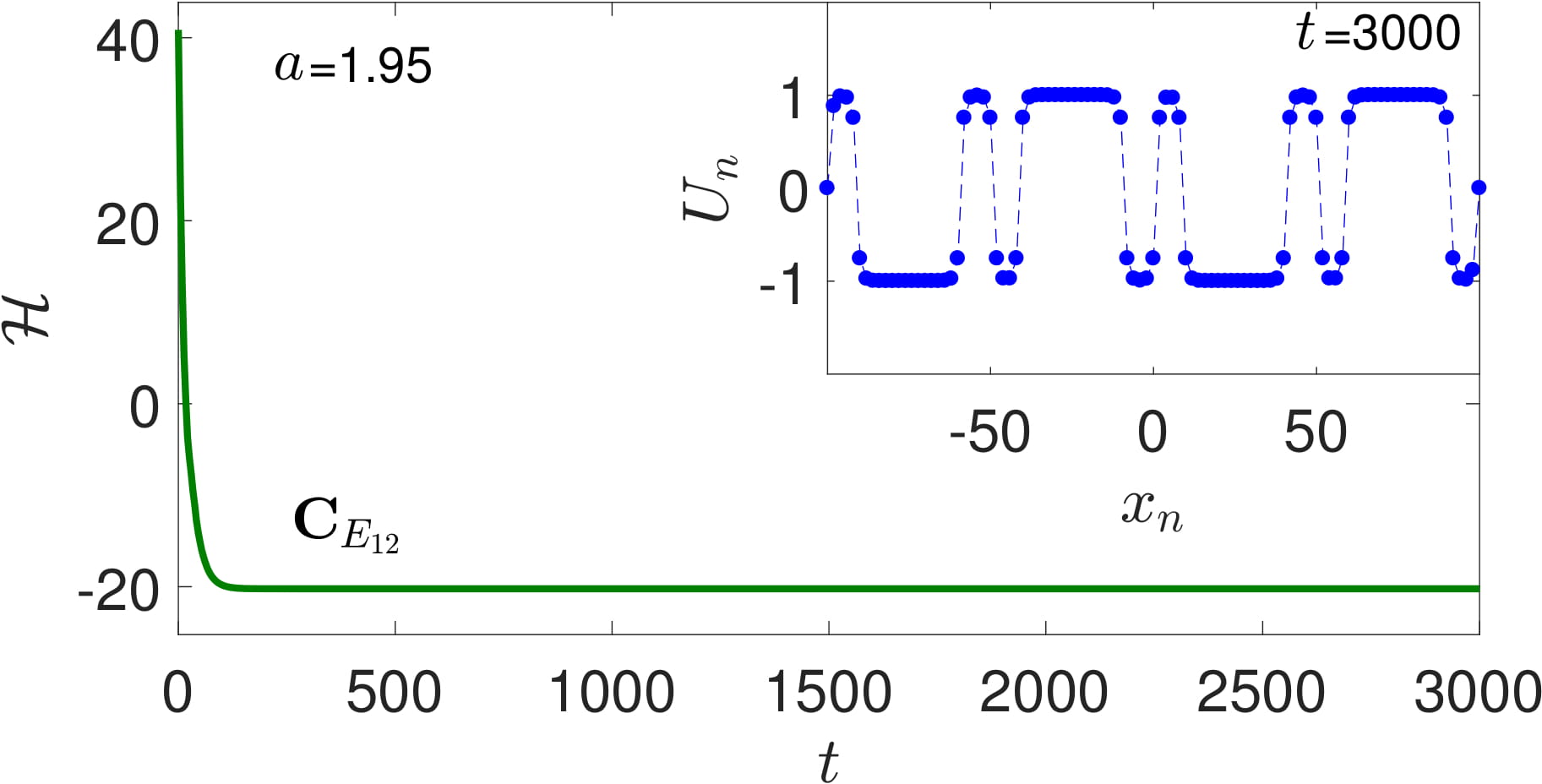}
\quad	
\includegraphics[scale=0.125]{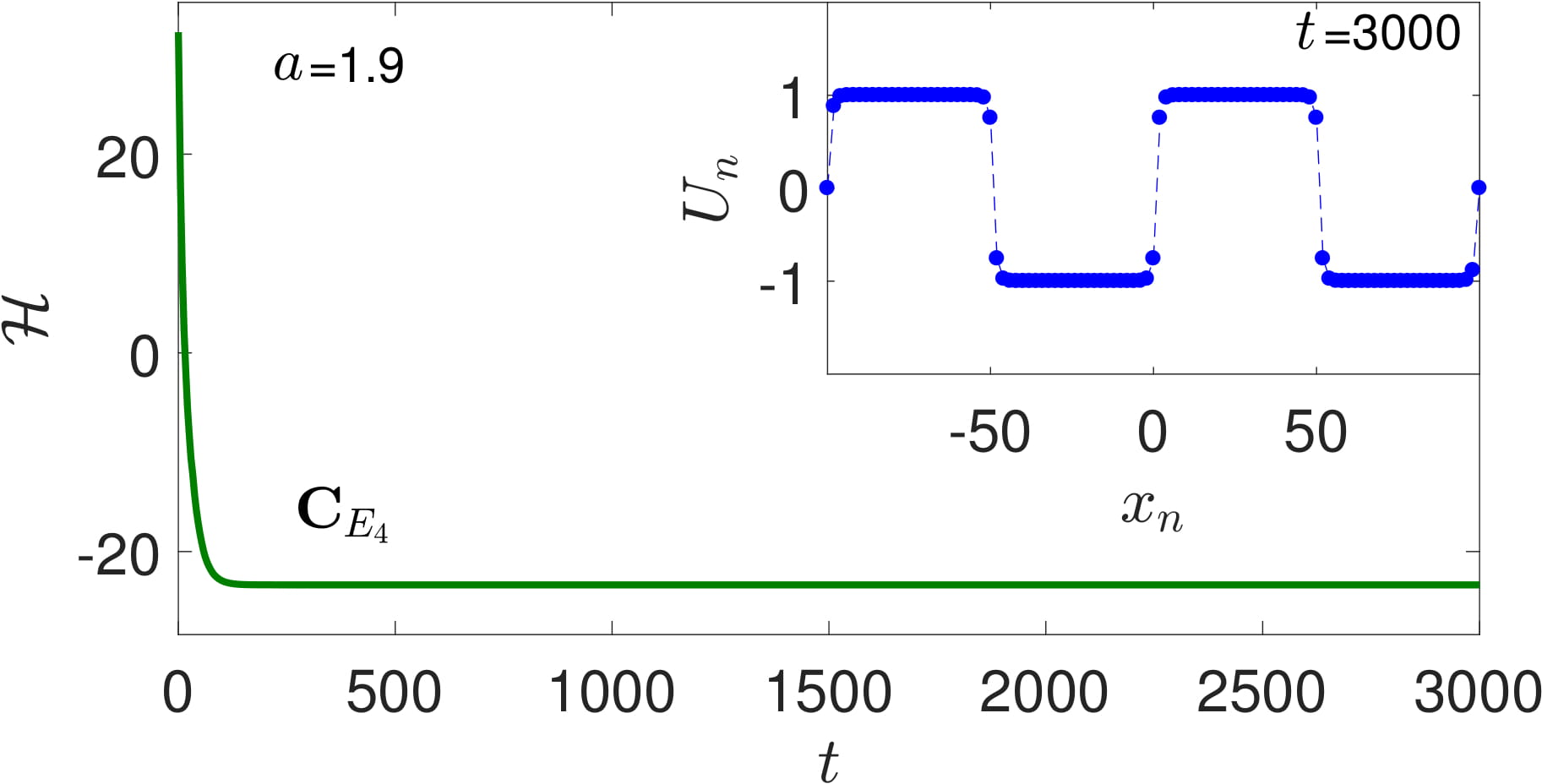}
\\[4ex]
\includegraphics[scale=0.125]{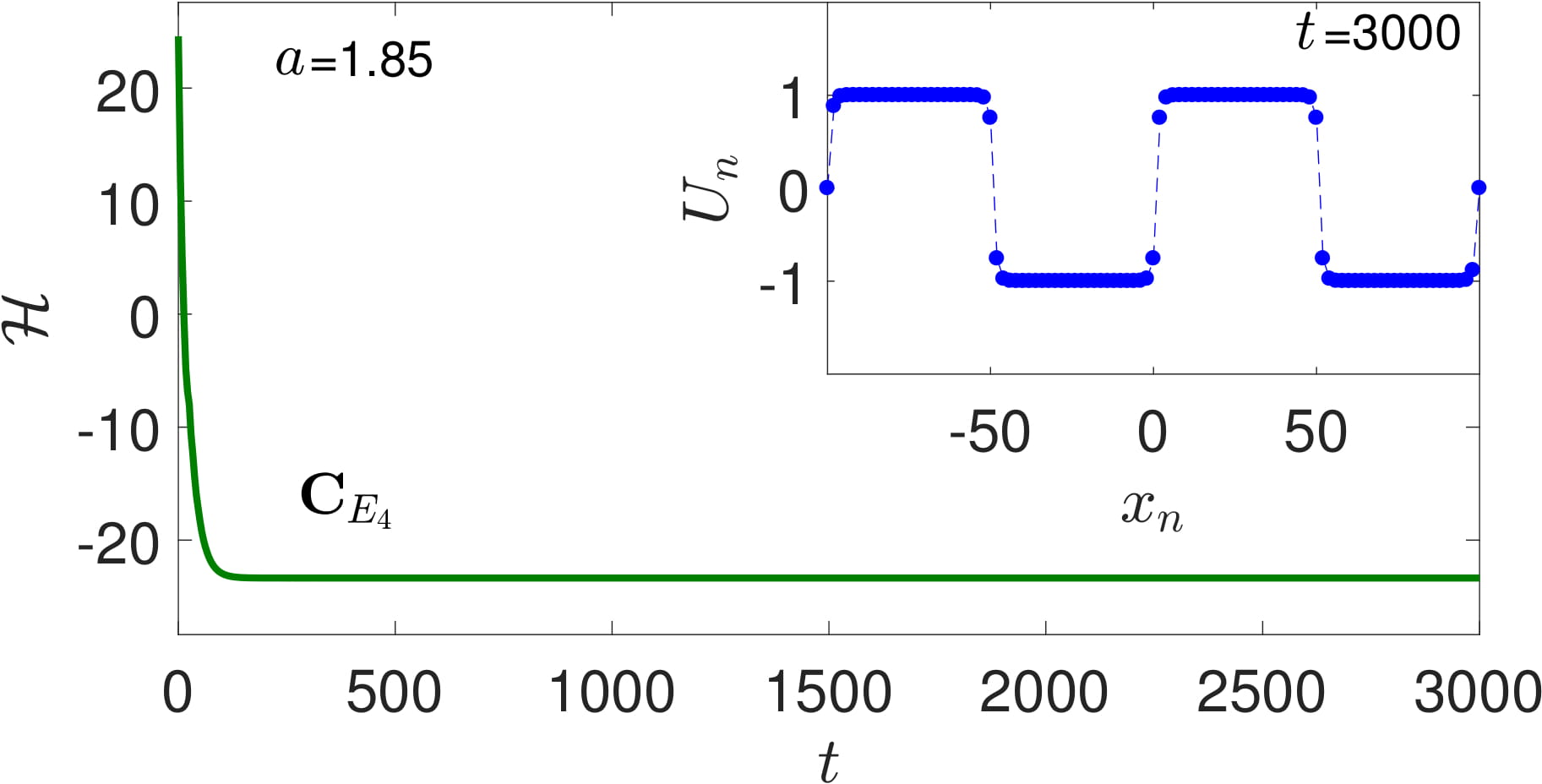}
\quad	
\includegraphics[scale=0.125]{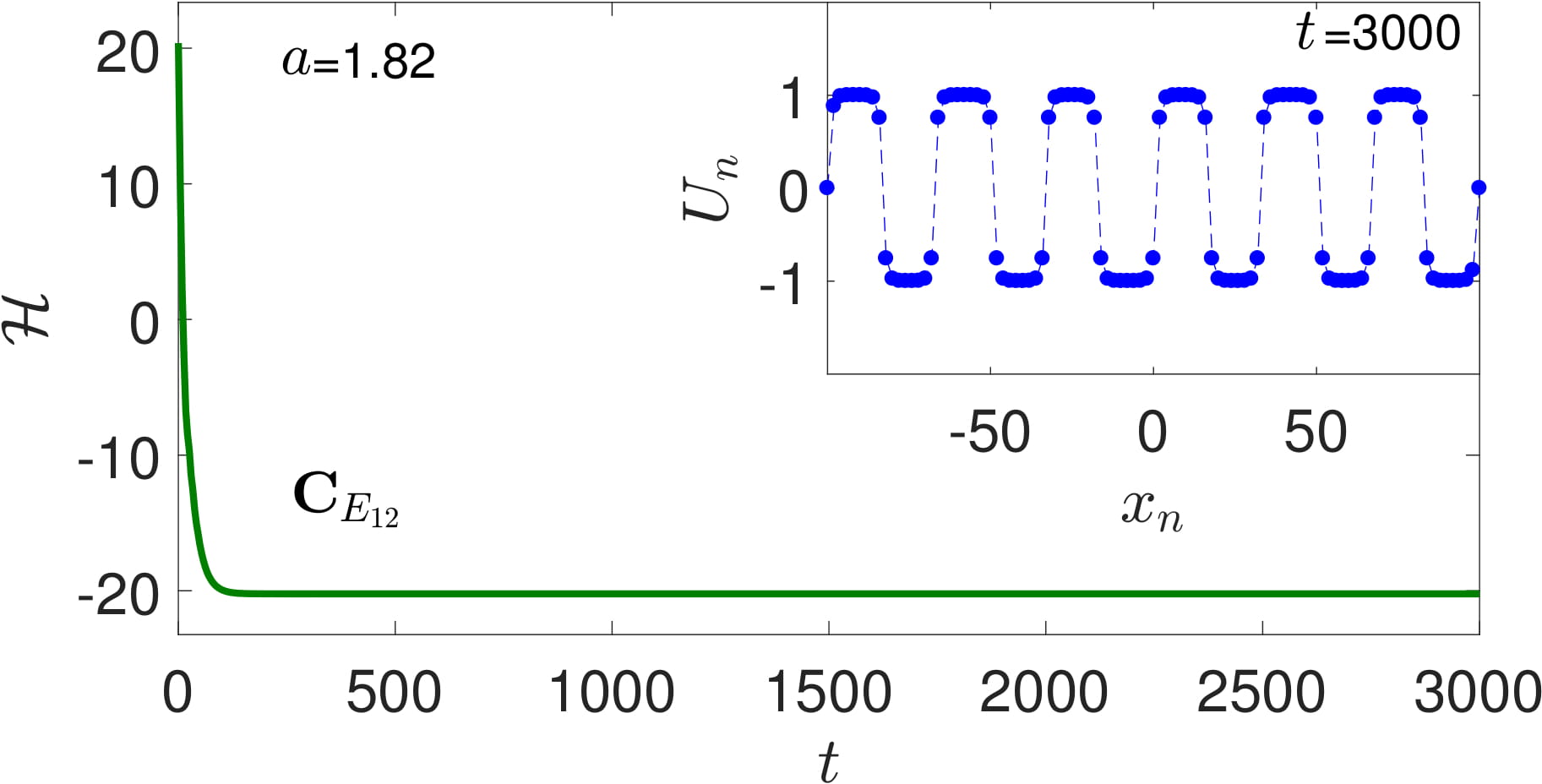}
			\vspace{0cm}
			\caption{(Color Online) Dynamics for an initial condition (\ref{ic1}), $U_n(0)\simeq\mathbf{C}_{E_{4}}$, for $\omega_d^2=1$, of varying amplitudes. Descending amplitudes, we observe convergence to equilibria of different branches than $\bfC_{E_4}$, and the appearance of amplitude values $a$ for which the solution converges to a nonlinear equilibrium $\overline{\Phi}_{4}\in \bfC_{E_{4}}$. Other parameters:  $K=99$, $L=200$, $\beta=1$, $\delta=0.05$.}
			\label{Figu5}
		\end{figure}

We proceed now to the second scenario of initial conditions, namely:
\begin{itemize}		
\item Scenario II. We study the dynamics corresponding to initial conditions (\ref{ic1}) of varying amplitudes, such that $U_n(0)\simeq\mathbf{C}_{E_{j}}$, in the case where the point $(||U_n(0)||,\omega_d^2)$ lies outside the local bifurcation diagram of Fig.~\ref{fig2}. 
\end{itemize}
Scenario II modifies Scenario I, as follows. 
Expecting that the global behavior of the branches may be more complicated for larger values of the parameter $\omega_d^2$ and norm $||U_n(0)||$ than those shown in Fig.~\ref{fig2}, we wish to explore the convergence dynamics for  points $(||U_n(0)||, \omega_d^2)$ far from the 
local bifurcation diagram. Our aim is to reveal the 
role of the amplitude of the extended initial condition, as well as of its  similarity to a branch, 
in the selection of the ultimate state of convergence.
%

For the numerical investigation on this scenario, we consider two cases of the initial conditions (\ref{ic1}), namely: $U_n(0)\simeq\bfC_{E_4}$ and $U_n(0)\simeq\bfC_{E_{9}}$, 
for varying amplitude 
$a$ and fixed $\omega_d^2=1$.  
The rest of the parameter values of the system are fixed as: $K=99$, $L=200$, $\beta=1$, and \textcolor{black}{$\delta=0.05$}. The system was again integrated up to $t=3000$.

The numerical results of the dynamics of the system for the initial conditions $U_n(0)\simeq\mathbf{C}_{E_{4}}$  are summarized in 
Figs.~\ref{Figu5} and \ref{Figu6}. 
We start with an amplitude $a=2$, corresponding to a norm $\|U_n(0)\|=14.14$. Clearly, the point $(\|U_n(0)\|,\omega_d^2)=(14.14,1)$, lies outside the local bifurcation diagram which is still similar (for $K=99$ and $L=200$) to that  of Fig.~\ref{fig2}: in this case, 
the maximum value of $E_j$ is $E_{99}=0.99$, while the maximum value of the norm of a solution family is $\|\Phi_n\|$=10.
The associated orbit  converges to an equilibrium of a $\bfC_{E_{12}}$ branch, of reduced norm $\|\Phi_{12}\|=9.37$, due to the dissipative nature of the system. The convergence occurs without metastable transition, as detected by the evolution of the Hamiltonian energy, shown in the  upper right panel of 
Fig.~\ref{Figu5}. The   symmetric nonlinear equilibrium  $\Phi_{12}\in \bfC_{E_{12}}$ attained, corresponds to a plateau in the energy at $\mathcal{H}=-20.22$ and it is displayed in the inset of the panel. The dynamical behavior of the system is in agreement with its linear stability analysis, since $\Phi_{12}$  is linearly stable with all the corresponding eigenvalues having negative real part. This is also true for all the converging equilibria of the numerical investigation that follows.
\begin{figure}[tbh!]
	\centering
\includegraphics[scale=0.125]{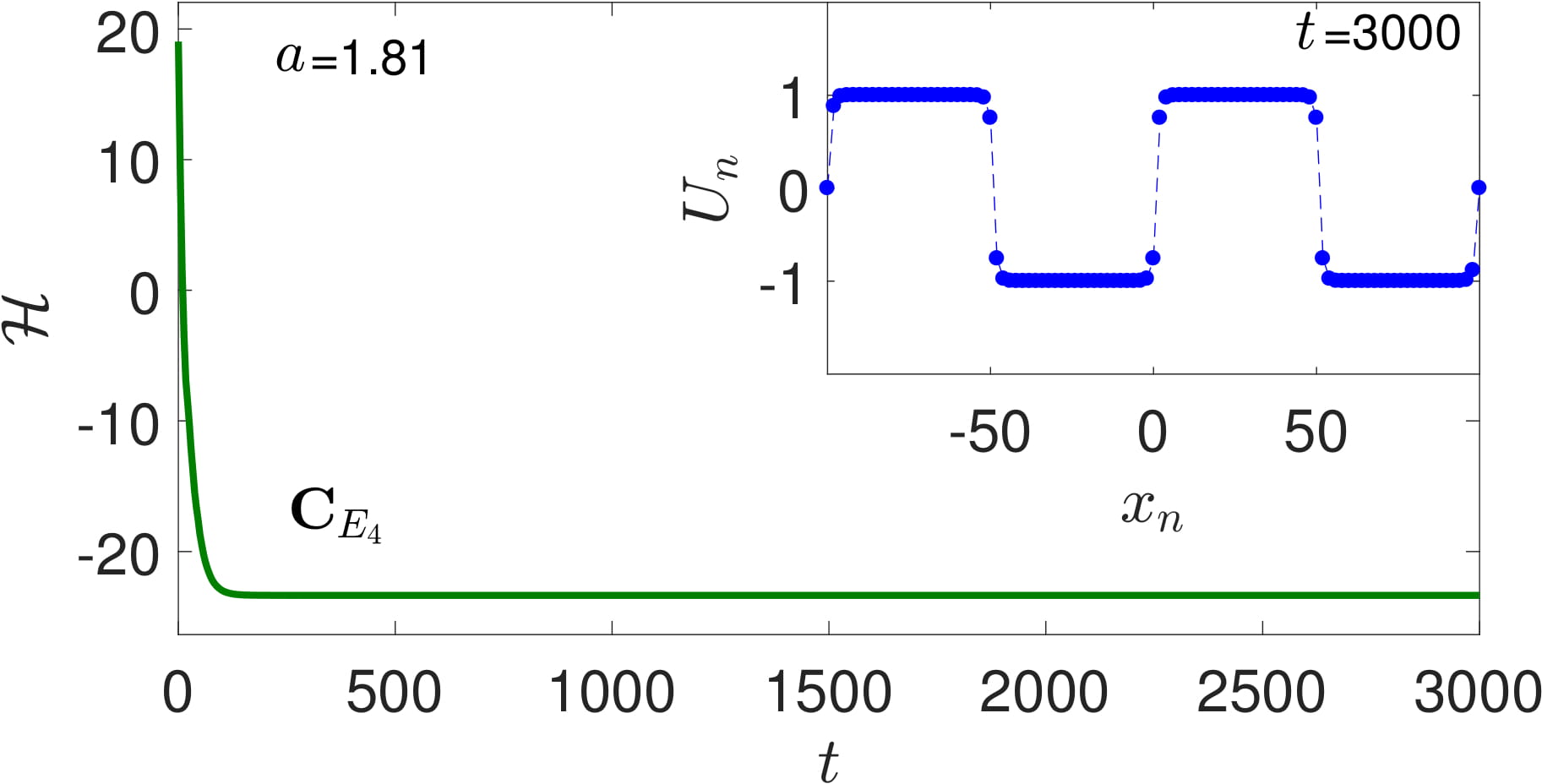}
\quad
\includegraphics[scale=0.125]{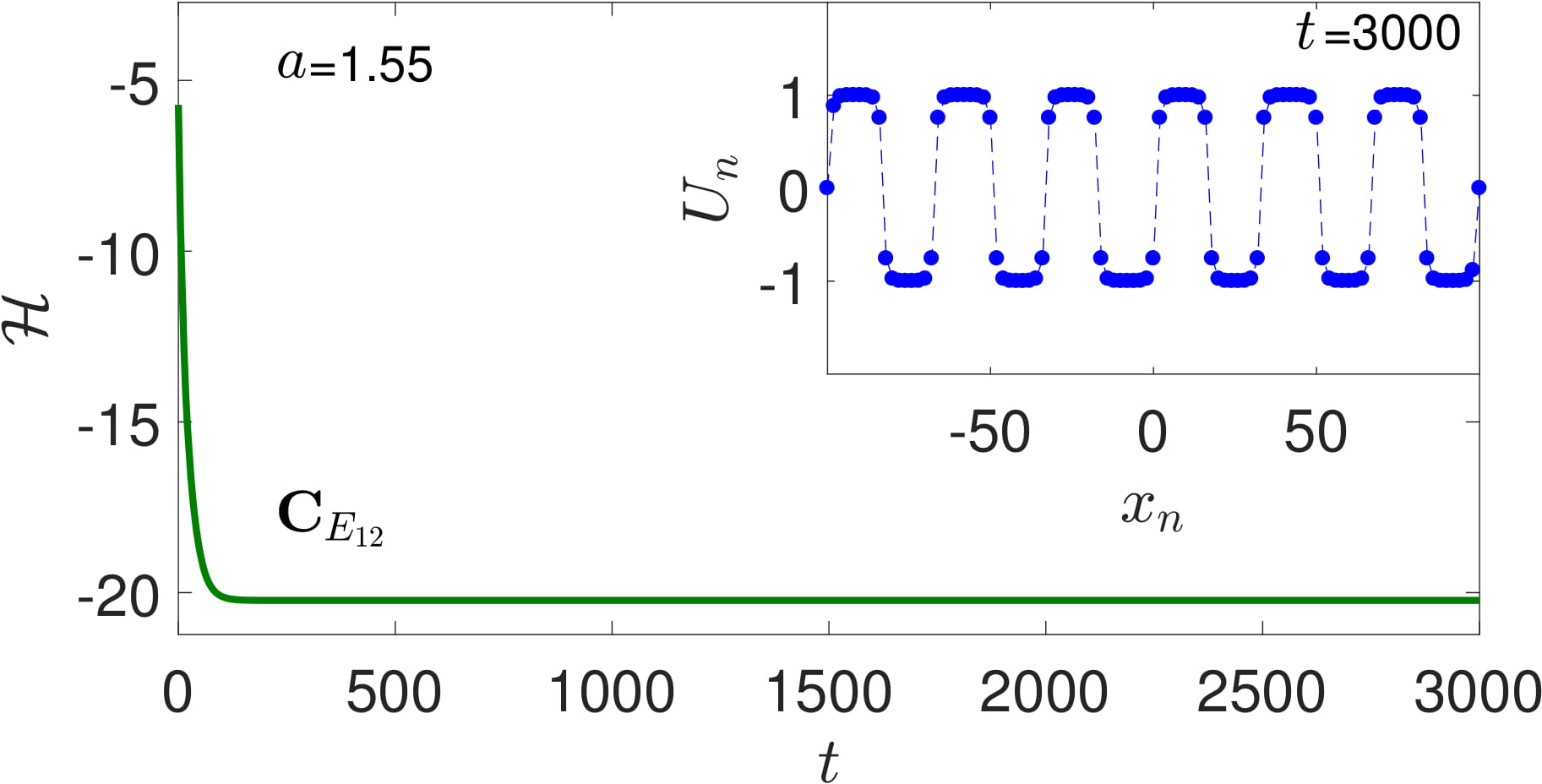}
\\[4ex]
\includegraphics[scale=0.125]{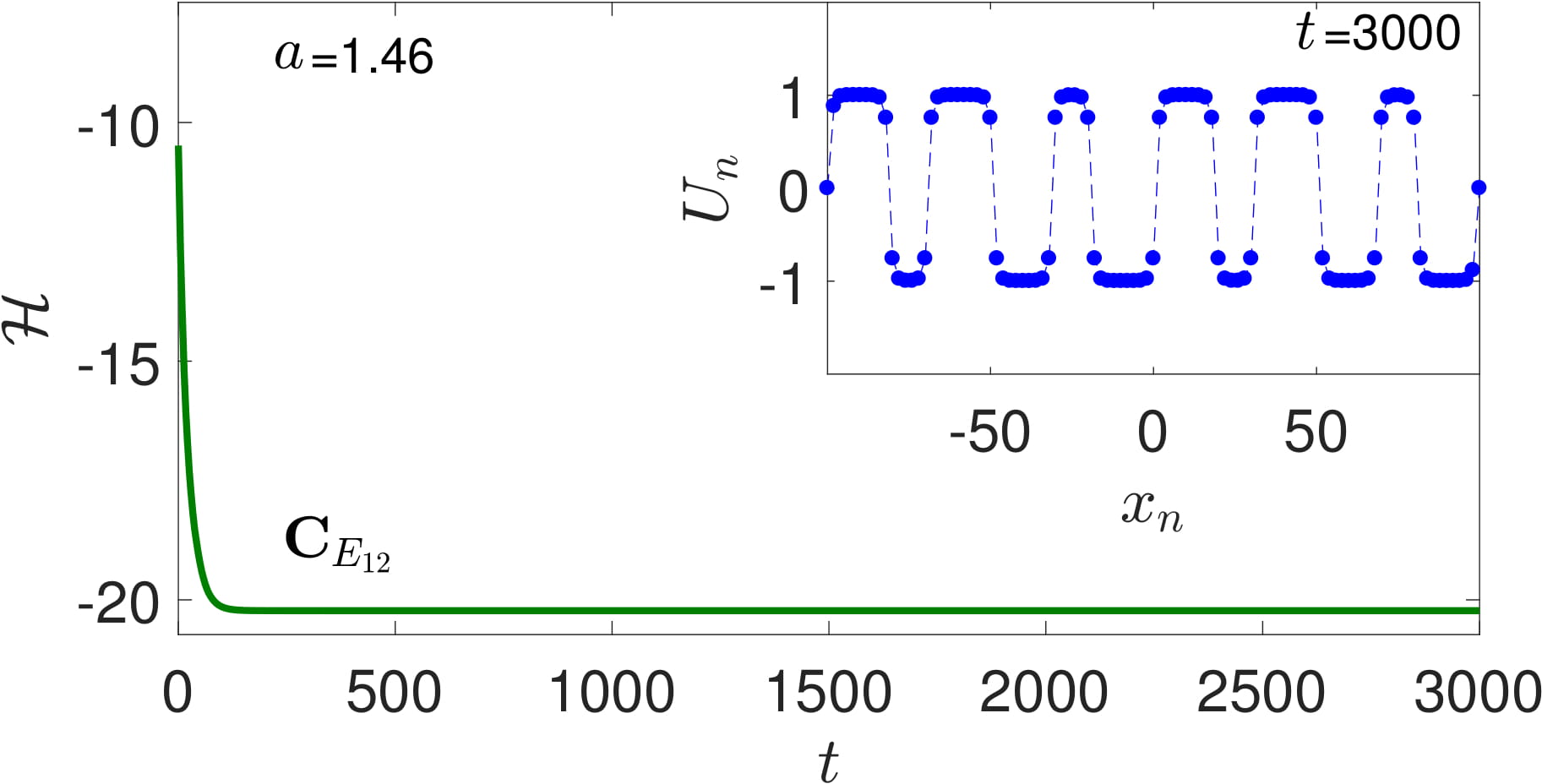}
	\quad	
\includegraphics[scale=0.125]{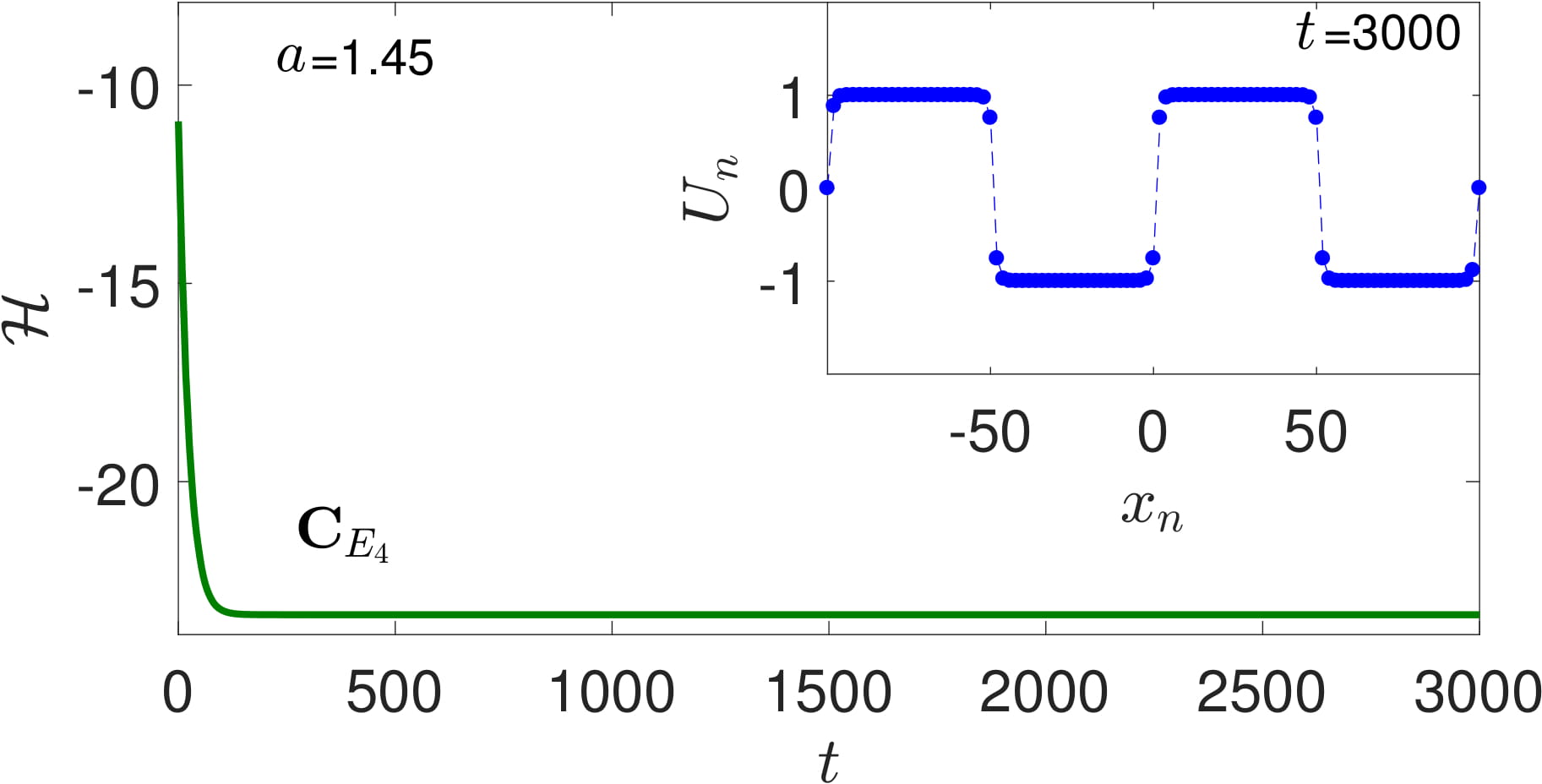}
	\vspace{0cm}
	\caption{(Color Online) Fig. \ref{Figu5} continued: Detection of the stability threshold $a^*=1.45$: For $a\leq a^*$, any initial condition (\ref{ic1}), $U_n(0)\simeq\mathbf{C}_{E_{4}}$, converges to a nonlinear equilibrium $\overline{\Phi}_{4}\in \bfC_{E_{4}}$.}
	\label{Figu6}
\end{figure}

By decreasing the amplitudes (using a $10^{-2}$ decrease step), remarkable features may be observed. For an initial condition of $a=1.95$ corresponding to norm $\|U_n(0)\|=13.8$, the dynamics are depicted in the  middle left panel of  Fig. \ref{Figu5}. We still observe convergence to a nonlinear equilibrium of $\bfC_{E_{12}}$, sharing the same norm $\|\Phi_{12}\|=9.37$ and the same energy $\mathcal{H}=-20.23$ with the previous case, but with a totally different, asymmetric profile. This feature
showcases the existence of different $\Phi_{E_{12}}$-type linearly stable configurations
which emerge from those described in Fig.~\ref{fig2} through symmetry breaking bifurcations or, independently, through saddle-node bifurcations. These branches, as it is evident, may lie  very close (in our case $10^{-2}$-close) to each other, both in energy and  norm.


Decreasing further the amplitude of the initial condition, we identify an interval of values of $a$, namely $1.83\leqslant a\leqslant 1.94$, for which the initial conditions $U_n(0)\simeq\bfC_{E_4}$ 
, converge to the same nonlinear equilibrium of $\bfC_{E_4}$, with $\|\Phi_{4}\|=95.6$ and $\mathcal{H}=-23.36$, i.e., to a final state with the same sign changes as the initial conditions. In the middle right  and bottom left panels of Fig.~\ref{Figu5} the results of two indicative values of $a$ i.e.~$a=1.9$ and $a=1.85$ are shown. 
On the other hand, in the $a$-values interval $1.46\leqslant a\leqslant1.82$ the initial conditions $U_n(0)\simeq\bfC_{E_4}$ converges interchangeably either in the previously mentioned $\Phi_4\in\bfC_{E_{4}}$ state or in different $\bfC_{E_{12}}$ branches. The results for the values $a=1.82, 1.81, 1.55~ \text{and}~ 1.46$ of this interval are shown in Figs.~\ref{Figu5} and \ref{Figu6}.


\begin{figure}[tbh!]
	\centering

\includegraphics[scale=0.125]{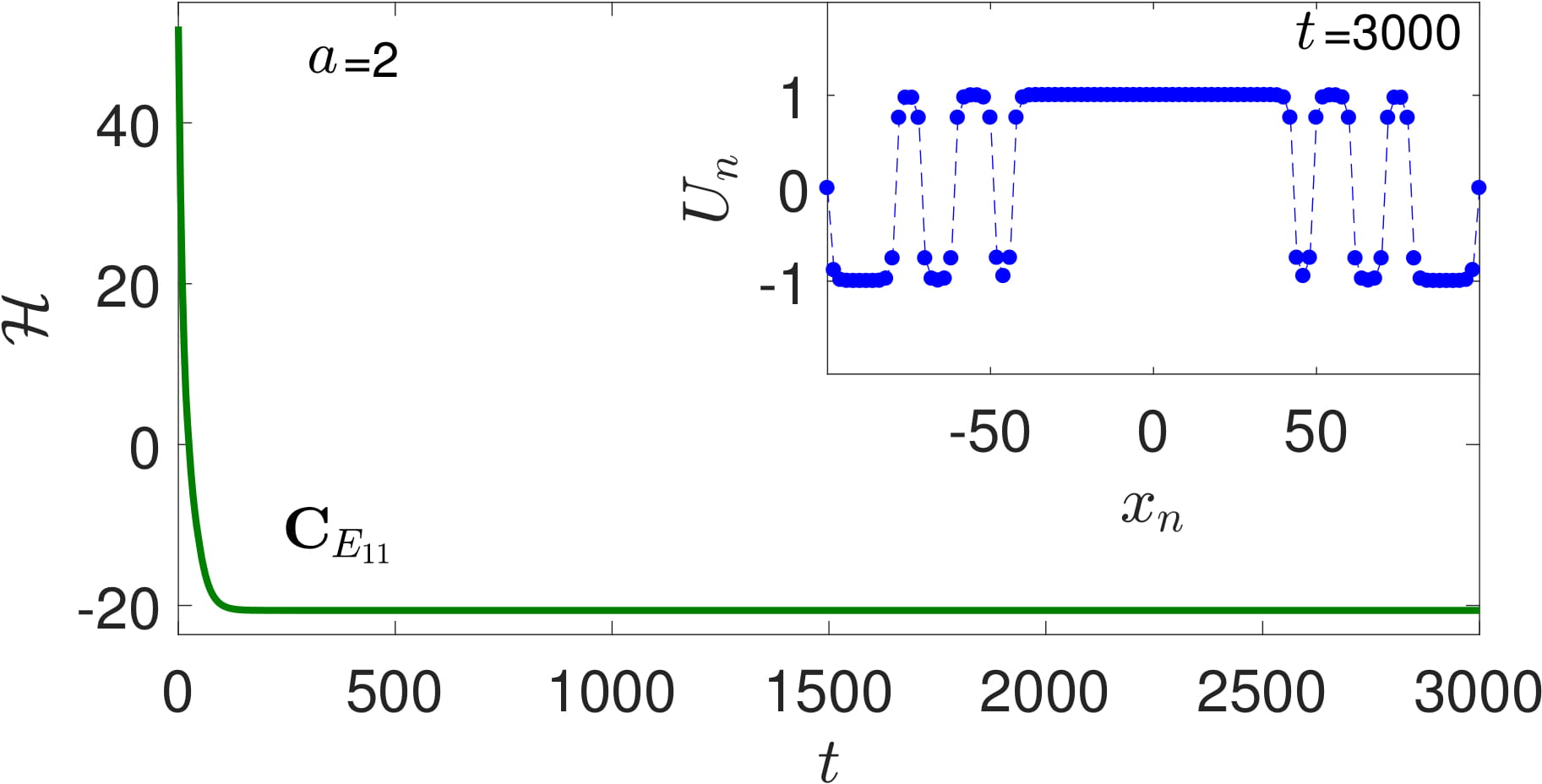}
\quad
\includegraphics[scale=0.125]{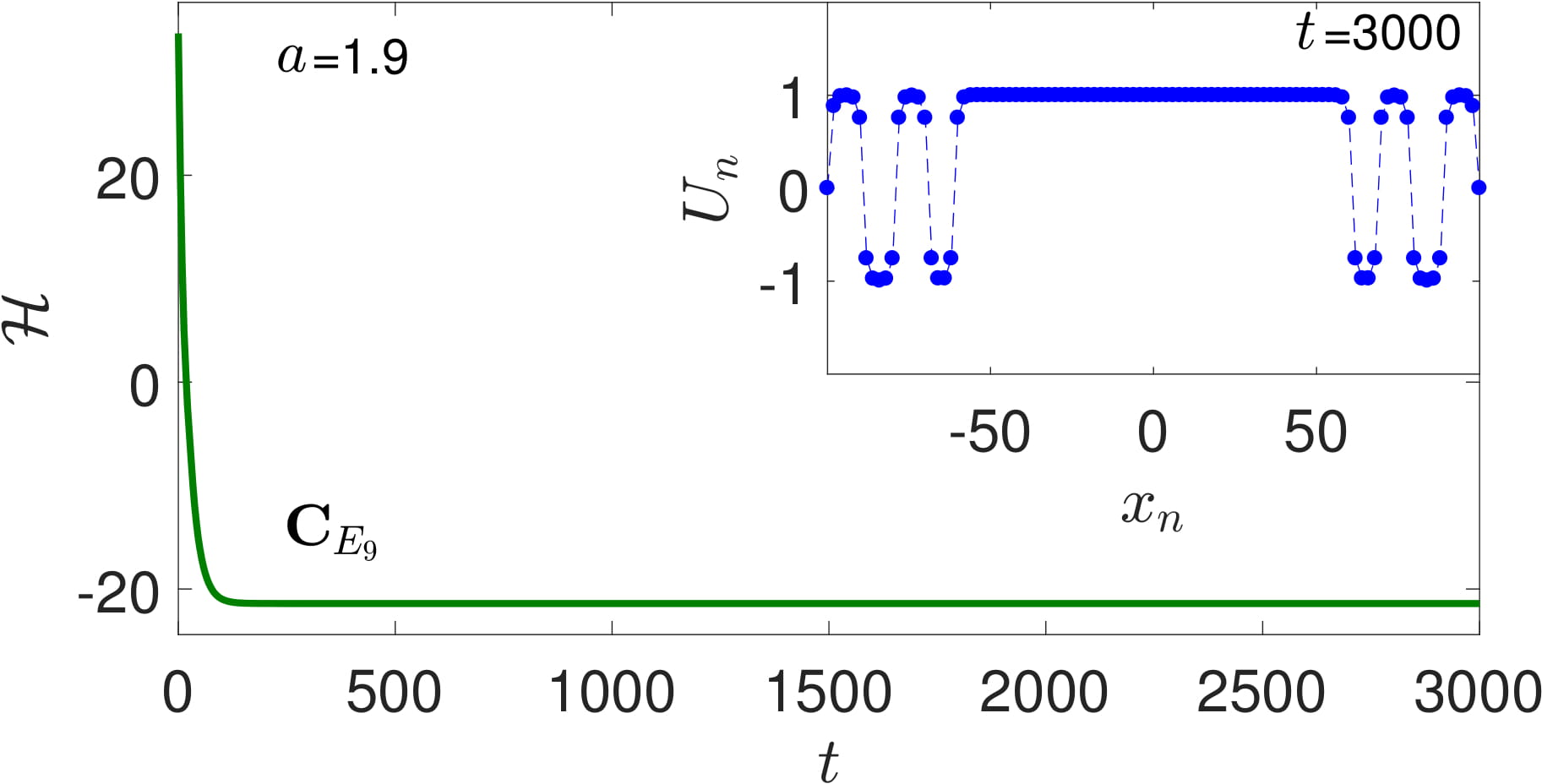}
\\[4ex]
\includegraphics[scale=0.125]{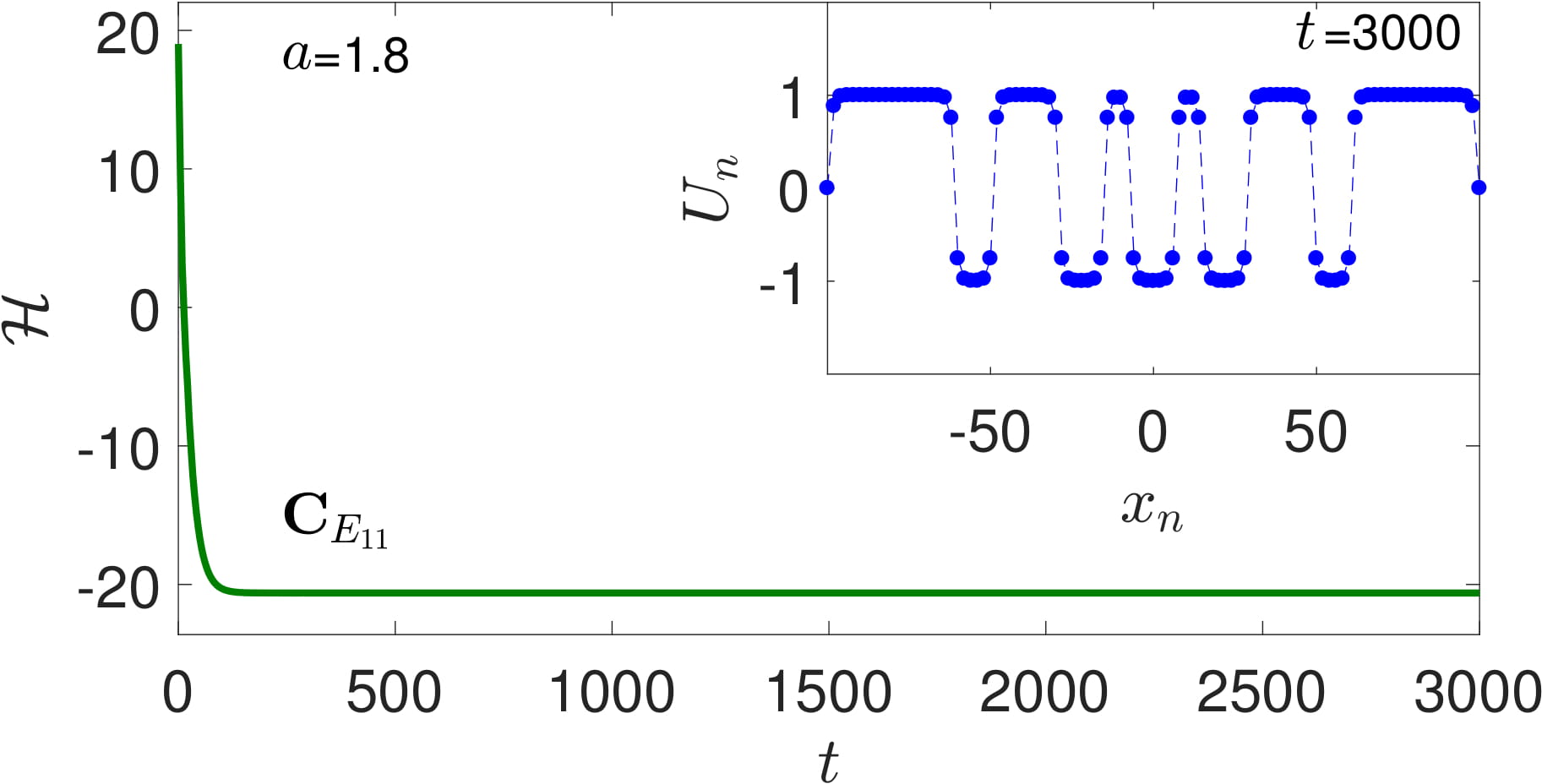}
	\quad	
\includegraphics[scale=0.125]{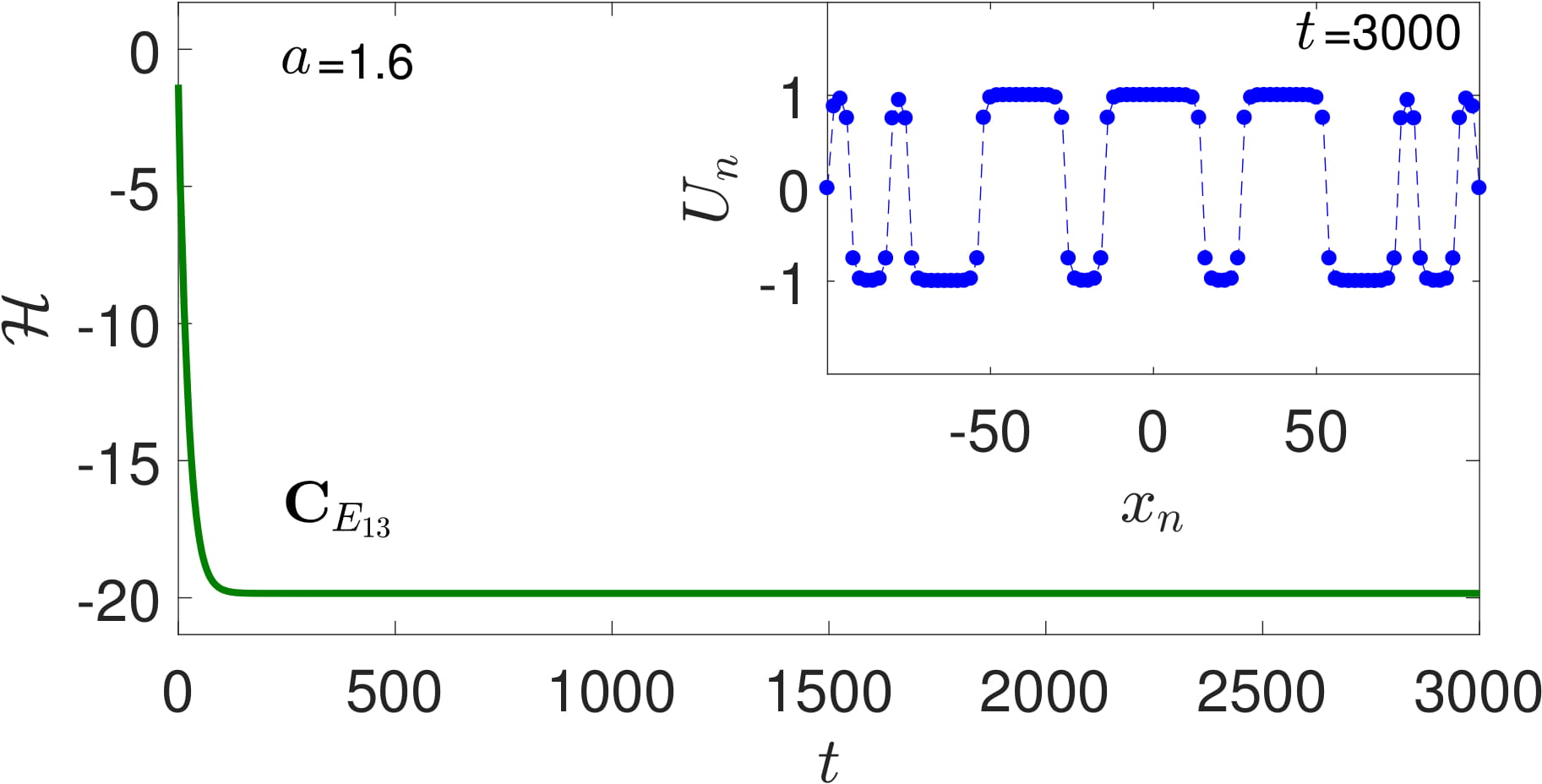}
\\[4ex]
\includegraphics[scale=0.125]{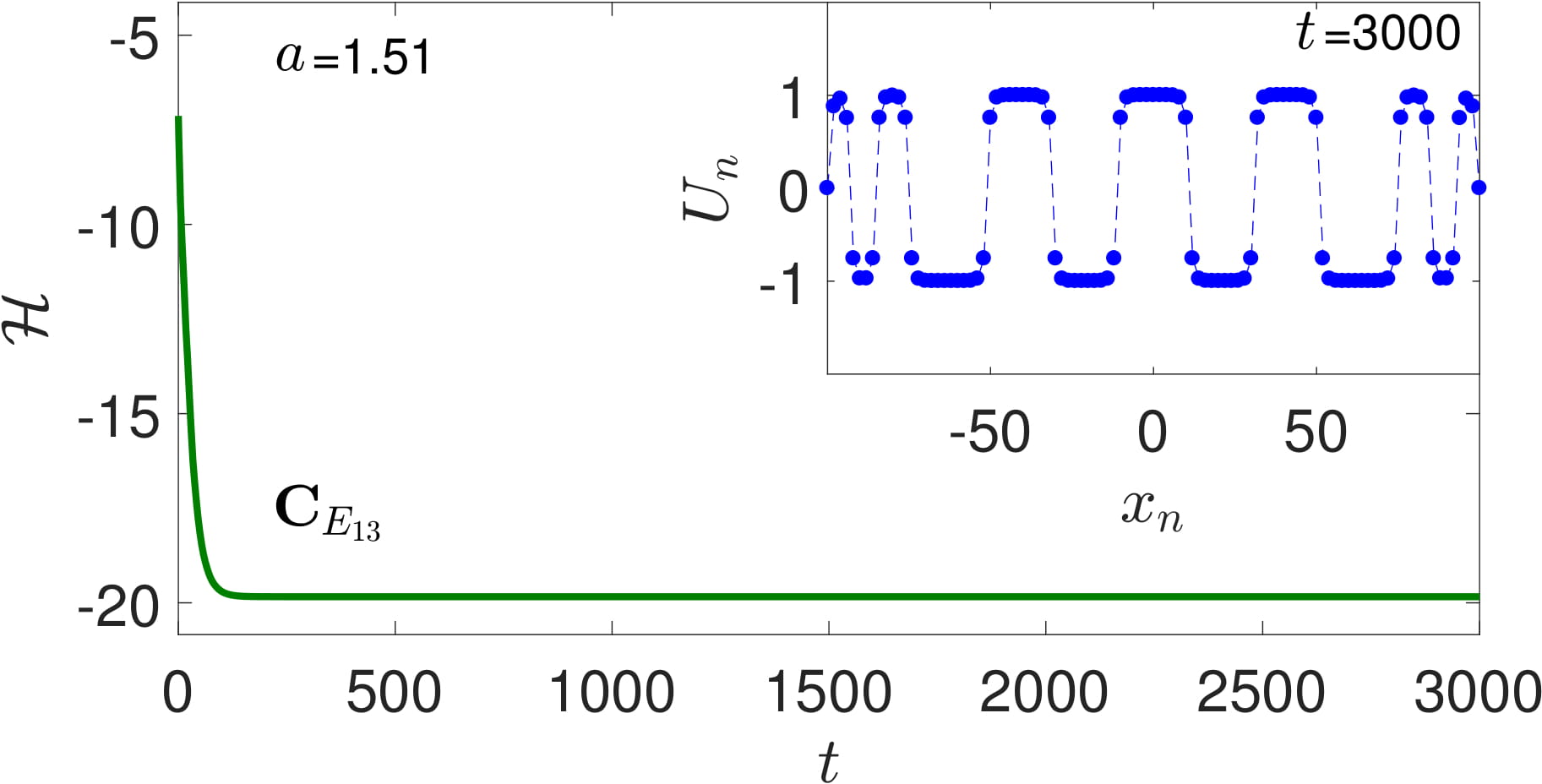}
\quad	
\includegraphics[scale=0.125]{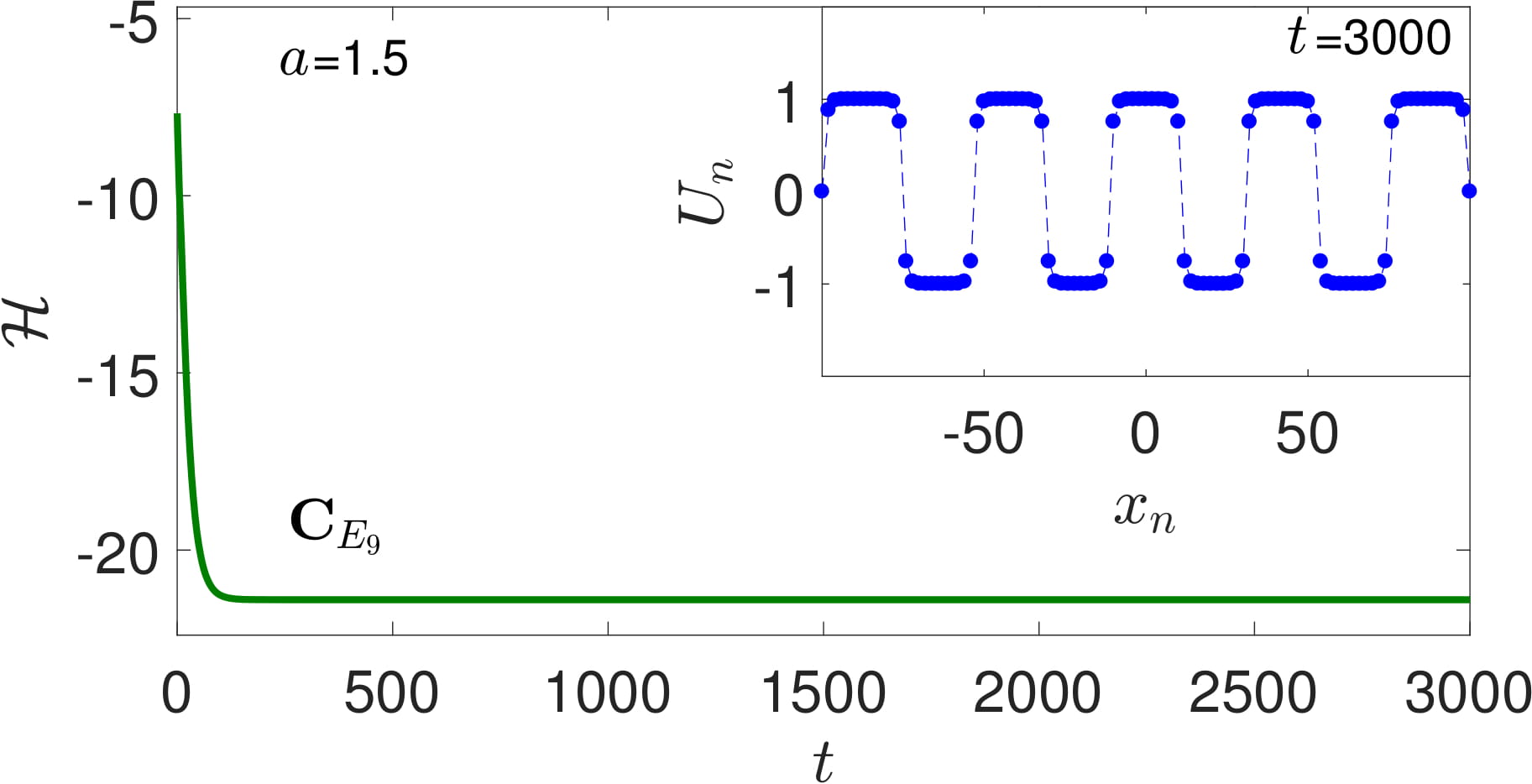} 
	\vspace{0cm}
	\caption{(Color Online)  Dynamics for an initial condition (\ref{ic1}), $U_n(0)\simeq\mathbf{C}_{E_{9}}$, for $\omega_d^2=1$, of varying amplitudes. Descending amplitudes, we observe convergence to equilibria of different branches than $\bfC_{E_{9}}$, and the appearance of amplitude values $a$ for which the solution converges to a nonlinear equilibrium $\overline{\Phi}_{9}\in \bfC_{E_{9}}$. The stability threshold is detected at $a^*=1.5$. Other parameters:  $K=99$, $L=200$, $\beta=1$, $\delta=0.05$.}
	\label{Figu8}
\end{figure}
The amplitude value $a^*=1.45$ defines a {\em threshold amplitude} below which  
the initial condition $U_n(0)\simeq\bfC_{E_4}$ converges to the equilibrium $\Phi_{4}\in \bfC_{E_4}$, shown in the inset of the  bottom right panel of Fig.~\ref{Figu6}. Observe that the portrayed $\Phi_{4}$-state, was also the attracting state detected in the cases of $a=1.9, 1.85, 1.81$. The showcased dynamics indicate that by selecting a similar initial condition to the branch $\bfC_{E_4}$, and  by reducing its amplitude, we may approach the stable manifold of  $\bfC_{E_4}$. Here we note that, as explained in the previous subsection, for small values of $\omega_d^2$ the $\bfC_{E_4}$ family possesses three unstable eigenvalues; however, for 
$\omega_d^2=1$ it is linearly stable.
Thus, the similarity of the initial condition to a branch
with a given number of sign-changes, as well as the choice of its amplitude below a certain threshold value, may define a criterion for generating orbits within the stable manifolds of the equilibria associated with that branch.

Another example supporting the above conjecture 
concerns the dynamics of initial condition $U_n(0)\simeq\mathbf{C}_{E_{9}}$, for the same value of $\omega_d^2=1$. The relevant dynamics are summarized in Fig. \ref{Figu8}. For decreasing values of amplitudes, we observe convergence to nonlinear equilibria of different branches, possessing a remarkably rich spatial structure:
for $a=2$ corresponding to norm $\|U_n(0)\|=14.14$, we observe convergence to a steady-state $\Phi_{11}\in \bfC_{E_{11}}$, with $\|\Phi_{11}\|=9.42$ and energy $\mathcal{H}=-20.62$. For $a=1.9$, where $\|U_n(0)\|=13.4$, the solution converges to a steady state $\Phi_{9}\in \bfC_{E_{9}}$ with $\|\Phi_{9}\|=9.53$, energy $\mathcal{H}=-21.4$ and an interesting profile characterized by a big central plateau. For $a=1.8$ corresponding to $\|U_n(0)\|=12.7$, convergence occurs again to a steady-state $\Phi_{11}$, which possesses  a  profile completely different from that of the steady state attained in the case of 
$a=2$. This again suggests the emergence of a bifurcation. 
For $a=1.6$ ($\|U_n(0)\|=11.3$), the system converges to a $\bfC_{E_{13}}$ steady-state with $\|\Phi_{13}\|=9.3209$ and energy $\mathcal{H}=-19.84$. Decreasing further the amplitude to the value $a=1.51$ (where $\|U_n(0)\|=10.7$), the corresponding solution converges again to a steady-state  $\Phi_{13}\in \bfC_{E_{13}}$ with $\|\Phi_{13}\|=9.3214$. Eventually, the threshold value for the dynamical stability of $\bfC_{E_9}$ was found at amplitude $a^*=1.5$, for which $\|U_n(0)\|=10.6$. 

The convergence to different states in both of the examined examples, as well as the richness of the dynamically emerging states, implies that a complete bifurcation analysis would be useful in order to acquire a more
systematic understanding of the structure, stability and of the basins of attraction of the various steady-states, especially for high values of the nonlinearity. This exceeds the scope of this work and it is left for future investigation.
\section{Conclusions}

In the present work, we have studied the dynamical transitions between the equilibria of the dissipative Klein-Gordon chain supplemented with Dirichlet boundary conditions, and have attempted to shed some light on
the complex energy landscape dictating the corresponding dynamics.

First, we have discussed the convergence to a single, non-trivial equilibrium, as the system falls within the class of a second-order gradient system, for which a discrete version of the \L{}ojasiewicz inequality is applicable. Then, we have applied in the discrete set-up, global bifurcation theory, to prove that global branches of nonlinear equilibria, are bifurcating from the linear eigenstates of the Dirichlet discrete Laplacian. Consequently, we have characterized the bifurcating branches of equilibrium states (and their stability), by their number of sign-changes. As such, the relevant equilibria may define nontrivial topological interpolations between the adjacent minima of the on-site $\phi^4$ potential.

Then, examining direct numerical simulations,
in conjunction with the linear stability analysis of the branches, we managed to reveal important features about the complicated structure of the convergence dynamics, depending on discreteness, nonlinearity
and dissipation. In the numerical experiments, we first 
considered spatially extended initial conditions, sharing the number of sign-changes with 
the equilibria of a specific branch. Varying the strength of nonlinearity and discreteness, 
we revealed the richness 
of the dynamical transition of the convergence dynamics to equilibrium configurations 
of distinct branches, and the variety of the potential
spatial structure of equilibria
even within elements of the same branch. We also gained insight on the role of the amplitude 
of the initial condition towards the
convergence to different final states on the branches of the emerging
complex bifurcation diagram.


Summarizing, in a simple dissipative lattice describing transition state phenomenology, 
we have shown that 
while the global attractor is trivial with respect 
to the phase-space topology, it may give rise to a highly non-trivial energy landscape,  
defined by the rich structure of the equilibrium set, and the variety of possibilities 
for convergence dynamics.  
 
The results presented in this work may pave the way for future work in many interesting directions. 
A first important one is to investigate further --both analytically and numerically-- 
the structure of the global bifurcation diagram, motivated by the strong indications 
for some of its fundamental features showcased in the present study. 
Another natural possibility is to extend the present 
considerations to higher-dimensional settings, where the role of the
	coupling will be more significant, due to the geometry enforcing additional
	neighbors. 
It is also relevant to consider, via the methods and diagnostics
	presented herein, the phenomenology of different types of potentials
	including, e.g., the Morse potential, which is relevant to DNA denaturation
	as modeled by the Peyrard-Bishop model~\cite{PB}. We speculate that
	the techniques used herein may turn out to be more broadly relevant to 
	such problems. 
Finally, it is certainly relevant to investigate the convergence dynamics 
to other second-order, or higher-order damped lattices, e.g., relevant to the dynamics of nonlinear metamaterials \cite{P4,P7}. Such studies are currently in progress and will be presented
	in future publications.
\appendix
\section{Discrete version of the \L{}ojasiewicz inequality}
\label{appen1}
\setcounter{section}{5}
In this appendix, we discuss briefly the discrete version of the \L{}ojasiewicz inequality, stated in Lemma \ref{LINEQ}. It is actually a direct corollary of the following abstract result \cite[Corollary 5.5, pg. 2839]{Jen2011}, for generalized analytic functions on Hilbert spaces: We consider $\mathcal{V}$ and $\cx$ two real Hilbert spaces with the inclusion $\mathcal{V}\hookrightarrow \cx$ being dense and continuous. The topological dual of $\mathcal{V}$ is denoted by $\mathcal{V}^{*}$. Thus, we may consider the evolution triple $\mathcal{V}\hookrightarrow \cx\equiv \cx^*\hookrightarrow \mathcal{V}^*$, still with continuous and dense inclusions. 
\begin{theorem}
\label{LIG}
Let $F:\mathcal{U}\rightarrow\mathbb{R}$ be an analytic
functional where $\mathcal{U}\subset \mathcal{V}$ is an open neighborhood of $\phi$, and assume that the (Fr\'{e}chet) derivative of $F$ at $\phi$, $DF(\phi) = 0$. We denote by $\mathcal{L}:\mathcal{V}\rightarrow \mathcal{V}^{*}$, the linear operator, defined by the linearization of $DF:\mathcal{V}\rightarrow \mathcal{V}^{*}$, at $\phi$. We assume the following
two conditions: (i) $\mathrm{ker}\mathcal{L}$ is finite-dimensional and (ii) for some linear compact operator $\mathcal{K}:\mathcal{V}\rightarrow \mathcal{V}^{*}$, the operator $\mathcal{L} +\mathcal{K}$ is invertible.

Then, there exist $\theta\in (0, 1/2)$, a neighborhood $\mathcal{Q}$ of $\phi$ and $\nu_0 >0$, for which 
\begin{eqnarray}
\label{abi}
\forall u\in \mathcal{Q},\;\;||DF(u)||_{\mathcal{V}^{*}}\geq \nu_0|F(u)-F(\phi)|^{1-\theta}.
\end{eqnarray}	
\end{theorem} 
Due to the finite-dimensionality of the phase space $\mathcal{V}=\ell^2_{K+2}\equiv \mathbb{R}^{K+2}\equiv \cx\equiv \mathcal{V}^*$, Lemma \ref{LINEQ}, follows as a straight-forward application of Theorem \ref{LIG}, if implemented to the functionals $F$ defined in (\ref{eq010}), and  $J$ defined in (\ref{eq026}). First, it follows by repeating the lines of \cite{JNLS2013,K1}, that $J(U)=DF(U)$, that is, the derivative of $F$. Then clearly, any solution $\Phi$ of the stationary problem (\ref{eq063})-(\ref{eq063A}) satisfies $J(\Phi)=\mathbf{0}$. 

Next, the linearization of $J$ on the equilibrium $\Phi$, is the operator $\mathcal{L}:\ell^2_{K+2}\rightarrow \ell^2_{K+2}$,
\begin{eqnarray*}
\label{linA}
\mathcal{L}(U_n)=-k\Delta_d U_n-f'(\Phi_n)U_n,\;\;\forall\;\; U\in\ell^2_{K+2},
\end{eqnarray*}
where $f$ is given in (\ref{eq017}). It can be easily checked, that it is well defined due to the equivalence of norms (\ref{eq07})  [or by applying (\ref{eq06})] for every equilibrium $\Phi$. Evidently, its kernel $\mathrm{ker}\mathcal{L}$ is finite dimensional, thus condition (i) is satisfied.
To check condition (ii), we consider the operator $\mathcal {K}=\lambda I$, for some suitable $\lambda>0$, and $I:\ell^2_{K+2}\rightarrow \ell^2_{K+2}$-the identity mapping. Again, since $\ell^2_{K+2}$ is finite-dimensional, $\mathcal {K}$ is compact.  Furthermore, there exists $\lambda>0$, such that the operator $\Lambda=\mathcal{L}+\mathcal {K}$ is coercive, and hence, invertible. Indeed, we observe that 
\begin{eqnarray*}
\left(\Lambda(U),U\right)_{\ell^2}&=& k(-\Delta_dU,U)_{\ell^2}-\omega_d^2\sum_{n=0}^{K+1}|U_n|^2+3\omega_d^2\beta\sum_{n=0}^{K+1}|\Phi_n|^2|U_n|^2+\lambda\sum_{n=0}^{K+1}|U_n|^2\\
&&\geq E_1\sum_{n=0}^{K+1}|U_n|^2-\omega_d^2\sum_{n=0}^{K+1}|U_n|^2+\lambda\sum_{n=0}^{K+1}|U_n|^2.
\end{eqnarray*}
Then, by choosing $\lambda>\omega_d^2$, the coercivity condition $\left(\Lambda(U),U\right)_{\ell^2}>\Omega ||U||_{\ell^2}^2$, is satisfied for $\Omega=E_1+(\lambda-\omega_d^2)>0$. Hence, both conditions (i) and (ii) of Theorem \ref{LIG}, being verified, the discrete inequality (\ref{LSin}) of Lemma  \ref{LINEQ}, readily follows, by applying the inequality (\ref{abi}), to the functionals discussed above.

Let us recall that the prototypical \L{}ojasiewicz inequality \cite{LL1,LL2} links the norm of the analytic function $F:\mathbb{R}^N\rightarrow\mathbb{R}$ (potential) in the neighborhood $\mathcal{Q}$ of a critical point $\phi\in\mathbb{R}^N$, with the norm of its gradient $\nabla F(x)$, according to the  estimate 
\begin{eqnarray}
\label{abi2}
||\nabla F(x)||\geq |F(x)-F(\phi)|^{1-\theta},\;\;\forall\;\;x\in \mathcal{Q},
\end{eqnarray}
and some $\theta\in (0, 1/2)$. 
Thus, its discrete version (\ref{LSin}), as derived by the application of Theorem \ref{LIG}, is an extension  of (\ref{abi2}) in finite nonlinear coupled lattices.
\section*{Acknowledgments}
%
The authors D.J.F., N.I.K., P.G.K., and V.K. 
acknowledge the support by NPRP grant {\#} [9-329-1-067] from Qatar National Research Fund (a member of Qatar Foundation). The findings achieved herein are solely the responsibility
of the authors.
The authors D.J.F. and P.G.K. gratefully acknowledge the support of the ``Greek Diaspora Fellowship Program'' 
of Stavros Niarchos Foundation.
The author K.V. gratefully acknowledges the support of the ``$\Upsilon\Pi\mathrm{ATIA}$ Doctoral Fellowship Program'' 
of the Research Unit of the University of the Aegean.

\bibliographystyle{amsplain}

\end{document}